\documentclass[oneside,final, letterpaper]{ucr}

\usepackage{color}
\usepackage[draft]{fixme}

\usepackage{graphicx}
\usepackage{amsmath,amssymb,amsfonts,amsthm}
\usepackage{mathtools}
\usepackage{ucs}
\usepackage[utf8x]{inputenc}
\usepackage[svgnames]{xcolor}
\usepackage[all,2cell,cmtip]{xy}
\usepackage{blkarray}
\usepackage{graphicx}
\usepackage{authblk}
\usepackage{multicol}
\UseTwocells

\definecolor{myurlcolor}{rgb}{0,0,0.8}
\definecolor{mycitecolor}{rgb}{0,0,0.8}
\definecolor{myrefcolor}{rgb}{0,0,0.8}
\usepackage[pagebackref]{hyperref}
\hypersetup{colorlinks,
linkcolor=myrefcolor,
citecolor=mycitecolor,
urlcolor=myurlcolor}

\usepackage{tikz}
\usetikzlibrary{arrows,backgrounds,circuits,circuits.ee.IEC,shapes,fit,matrix,decorations.markings, positioning}
\pgfdeclarelayer{edgelayer}
\pgfdeclarelayer{nodelayer}
\pgfsetlayers{edgelayer,nodelayer,main}
\tikzstyle{main node} =[circle,fill=white!20,draw,font=\sffamily\Large\bfseries]
\tikzstyle{terminal}=[circle,fill=white!20,draw,font=\sffamily\Large\bfseries,color=purple,fill=none]
\tikzstyle{shadecircle}=[circle,fill=DodgerBlue,draw=Black]
\tikzstyle{lipid}=[-,draw=Black,line width=0.750]
\tikzstyle{empty}=[inner sep=3pt]
\definecolor{lblue}{rgb}{0,250,255}
\tikzstyle{species}=[circle,fill=yellow,draw=black,scale=.75]
\tikzstyle{transition}=[rectangle,fill=lblue,draw=black,scale=1.15]
\tikzstyle{none}=[inner sep=0pt]
\tikzstyle{empty}=[circle,fill=none, draw=none]
\tikzstyle{inputdot}=[circle,fill=black,draw=black, scale=.25]
\tikzstyle{tick}=[-,draw=black,postaction={decorate},decoration={markings,mark=at position .5 with {\draw (0,-0.1) -- (0,0.1);}},line width=1.000]
\tikzstyle{inputarrow}=[->,draw=black, shorten >=.05cm]
\tikzstyle{simple}=[-,draw=black,line width=1.000]

\tikzstyle{inarrow}=[->, >=stealth, shorten >=.03cm,line width=1]

\newcommand{\FinSet}{\mathtt{FinSet}}
\newcommand{\Set}{\mathtt{Set}}
\newcommand{\Cospan}{\mathtt{Cospan}}

\newcommand{\Rel}{\mathtt{Rel}}

\newcommand{\RNet}{\mathtt{RNet}}
\newcommand{\RxNet}{\mathtt{RxNet}}

\newcommand{\relto}{\nrightarrow}

\newcommand{\CC}{\mathtt{C}}
\newcommand{\D}{\mathtt{D}}

\newcommand{\Circ}{\mathtt{Circ}}

\newcommand{\LinRel}{ \mathtt{LinRel}}

\newcommand{\FinPopSet}{\mathtt{FinSet}_{\epsilon}}

\newcommand{\Mark}{\mathtt{Mark}}
\newcommand{\OpenMark}{\mathtt{OpenMark}}
\newcommand{\DetBalMark}{\mathtt{OpenMark}_{\epsilon}}

\newcommand{\Dynam}{\mathtt{Dynam}}

\newcommand{\OpenSys}{\mathtt{OpenSys}}
\newcommand{\Behave}{\mathtt{Behavior}}
\newcommand{\SemialgRel}{\mathtt{SemiAlgRel}}

\newcommand*{\darkgraysquare}{\textcolor{gray}{\blacksquare}}
\newcommand*{\vdarkgraysquare}{\textcolor{darkgray}{\blacksquare}}
\newcommand*{\graysquare}{\textcolor{lightgray}{\blacksquare}}

\newcommand{\im}{\text{Im}}
\let\hom\relax
\DeclareMathOperator{\hom}{\mathrm{Hom}}

\newcommand{\beq}{\begin{equation}}
\newcommand{\eeq}{\end{equation}}

\newcommand{\define}[1]{{\bf \boldmath{#1}}}
\newcommand{\maps}{\colon}

\newcommand{\N}{\mathbb{N}}

\newcommand{\R}{\mathbb{R}}

\newcommand{\Graph}{\mathrm{Graph}}

\newcommand{\f}{f_i}
\newcommand{\fp}{f'_i}
\newcommand{\fj}{f_j}
\newcommand{\fpj}{f'_j}

\begin{document}


\title{Open Markov Processes and Reaction Networks}
\author{Blake Stephen Swistock Pollard}
\degreemonth{September}
\degreeyear{2017}
\degree{Doctor of Philosophy}
\chair{Prof. John C. Baez}
\othermembers{Prof. Shan-Wen Tsai\\
Prof. Nathaniel Gabor}
\numberofmembers{3}
\field{Physics}
\campus{Riverside}

\maketitle
\copyrightpage{}
\approvalpage{}

\degreesemester{Spring}

\begin{frontmatter}

\begin{acknowledgements}
Thanks to my advisor, Prof. John C. Baez, for his guidance and for opening many doors to me.  Thanks to Daniel Cicala, Kenny Courser, Brandon Coya, Jason Erbele, Brendan Fong, Jayne Thompson, Whei Yip and Mile Gu for many great conversations. I thank the Physics Department at U.C. Riverside for their support.  Part of this work was done at the Centre for Quantum Technologies at the National University of Singapore supported by a National Science Foundation East Asia and Pacific Summer Institutes Award, the National Research Foundation of Singapore, and an FQXi minigrant. Part of this work was also performed with funding from a subcontract with Metron Scientific Solutions working on DARPA's Complex Adaptive System Composition and Design Environment (CASCADE) project.

The material on entropy production in open Markov processes in Chapter \ref{ch:entropymark} appears in my paper `A Second Law for open Markov proceses' \cite{Pollard} and the example of membrane diffusion appears in my paper `Open Markov processes: a compositional perspective on non-equilibrium steady states in biology' \cite{PollardBio}.  Chapter \ref{ch:catopen} recalls some relevant results from Brendan Fong's theory of decorated cospans \cite{Fong}. Chapters \ref{ch:openmark} and \ref{ch:blackbox} arose from joint work with John Baez and Brendan Fong and has been published previously in our paper `A compositional framework for Markov processes' \cite{BaezFongP}.  Chapter \ref{ch:rx} consists of joint work with John Baez appearing in our paper `A compositional framework for reaction networks' \cite{BaezPRx}.

\end{acknowledgements}

\begin{dedication}
\null\vfil
{\large
\begin{center}
To my parents for all their love and support and to my sister, for being her.
\end{center}}
\vfil\null
\end{dedication}

\begin{abstract}


We begin by defining the concept of `open' Markov processes, which are continuous-time Markov chains where probability can flow in and out through certain `boundary' states. We study open Markov processes which in the absence of such boundary flows admit equilibrium states satisfying detailed balance, meaning that the net flow of probability vanishes between all pairs of states. External couplings which fix the probabilities of boundary states can maintain such systems in non-equilibrium steady states in which non-zero probability currents flow. We show that these non-equilibrium steady states minimize a quadratic form which we call `dissipation.' This is closely related to Prigogine's principle of minimum entropy production. We bound the rate of change of the entropy of a driven non-equilibrium steady state relative to the underlying equilibrium state in terms of the flow of probability through the boundary of the process.

We then consider open Markov processes as morphisms in a symmetric monoidal category by splitting up their boundary states into certain sets of `inputs' and `outputs.' Composition corresponds to gluing the outputs of one such open Markov process onto the inputs of another so that the probability flowing out of the first process is equal to the probability flowing into the second. Tensoring in this category corresponds to placing two such systems side by side.

We construct a `black-box' functor characterizing the behavior of an open Markov process in terms of the space of possible steady state probabilities and probability currents along the boundary. The fact that this is a functor means that the behavior of a composite open Markov process can be computed by composing the behaviors of the open Markov processes from which it is composed. We prove a similar black-boxing theorem for reaction networks whose dynamics are given by the non-linear rate equation. Along the way we describe a more general category of open dynamical systems where composition corresponds to gluing together open dynamical systems.

\end{abstract}

\tableofcontents
\end{frontmatter}


\chapter{Introduction}\label{ch:intro}

This thesis is part of a larger effort utilizing category theory to unify a variety of diagrammatic approaches found across the sciences including, but not limited to, electrical circuits, control theory, Markov processes, reaction networks, and bond graphs \cite{BaezCoyaRebro, BaezEberleControl, BaezFongCirc, JasonThesis, FongThesis, SpivakWiring}. We hope that the categorical approach will shed new light on each of these subjects as well as their interrelation. Just as Feynman diagrams provide a notation for calculations in quantum field theory, various graphical or network representations exist for describing proofs, processes, programs, computations, and more. Governing many such diagrammatic approaches is the theory of symmetric monoidal categories \cite{BaezStay}. In this thesis we focus on two basic types of systems which admit some graphical syntax: probabilistic systems satisfying the Markov property, hereafter called Markov processes and systems of interacting entities called reaction networks. Markov processes can be represented by labelled, directed graphs which serve as a notation for a set of coupled linear differential equations. Reaction networks can be represented by certain bipartite graphs commonly referred to as Petri nets which, for the purposes of this thesis, provide a graphical syntax for sets of coupled non-linear differential equations which describe for instance the evolution of a number of well-mixed, interacting chemical species.  

We study these as `open' systems; open in the sense that they can interact with other systems, the environment, or even a user. We thus generalize the graphical syntax for such systems and consider certain types of `open graphs'. Gluing such graphs together along their common interfaces corresponds to interconnection of the open systems they represent. This is formalized using category theory, by considering open systems as morphisms. This setting allows us to represent the composition of systems by the composition of their respective morphisms. Just as complicated functions of many variables can be decomposed in a variety of ways, so can a system be thought of as a composite system in many different ways depending on how one draws the boundaries designating components. We go on to study the `behaviors' of these systems and subsystems in a way which leverages the structure of composition, using functors to compute properties or behaviors of these open systems. Such functors map the `syntax' representing some class of open systems to some notion of `semantics' or behavioral interpretation of such syntax.

In our approach, syntax typically corresponds to some category of `open graphs.' Semantics could correspond, for instance, to sets of coupled differential equations represented by such open graphs or to properties of these equations and their steady states. The fact that these behaviors are computed functorially means that the behavior of a composite system can be computed as the composite of the behaviors of its constituent systems. The behavior categories typically forget many details of the systems themselves, indicating that for complex systems composition of computed behaviors of constituent systems is more efficient than direct computation of the behavior of the large complex system. Using functors to assign meaning or semantics to some syntax is commonly referred to as `functorial semantics,' inspired by F.\ W.\ Lawvere's thesis \cite{Lawvere}. 

This thesis is structured as follows. In Chapter \ref{ch:mark} we provide some preliminaries regarding Markov processes and explain how to define a Markov processes from a given directed graph whose edges are labelled by non-negative real numbers. In Chapter \ref{ch:entropymark} we introduce the concept of an open Markov process in which probability can flow in and out of the process at certain boundary states. We prove a number of theorems bounding the rate of change of various entropic quantities in open Markov processes. By externally fixing the probabilities at the boundary states of an open Markov process one generally induces a non-equilibrium steady state for which non-zero probability currents flow through the boundary of the process in order to maintain such a steady state. We show that these steady states minimize a quadratic form which we call `dissipation.' For steady states near equilibrium, dissipation approximates the rate of entropy production and Prigogine's principle of minimum entropy production coincides with the principle of minimum dissipation. Steady states far from equilibrium minimize dissipation, not the entropy production. This exemplifies the well-known fact that Prigogine's principle of minimum entropy production is valid only in the vicinity of equilibrium. We show that entropy production and dissipation can both be seen as special cases of the rate of generalized entropy production. Dissipation plays a key role in our functorial characterization of steady states in open Markov processes viewed as morphisms in a symmetric monoidal category.
 
We then turn our attention to the compositional modeling of such systems using symmetric monoidal categories. To this end we apply techniques developed by Brendan Fong for constructing categories whose morphisms correspond to open systems \cite{Fong, FongThesis}. These are called `decorated cospan' categories. Roughly speaking, the `cospan' contains the information regarding the `interfaces' of the open system, while the `decoration' carries the information regarding the details of the systems. In Chapter \ref{ch:catopen} we introduce decorated cospan categories and recall some relevant theorems regarding the construction of decorated cospan categories and functors between them. In Chapter \ref{ch:openmark} we explain how open Markov processes can be viewed as morphisms in a decorated cospan category. Here we introduce a slightly refined notion of open Markov processes where we restrict our attention to open Markov processes whose underlying closed Markov process admits a particularly nice type of equilibrium distribution satisfying a property called `detailed balance' for which all probability currents vanish. 

In Chapter \ref{ch:blackbox} we prove a `black-box' theorem for open detailed balanced Markov processes which characterizes an open Markov processes in terms of the subspaces of possible steady state probabilities and probability currents along the boundary of the process. We call the subspace of possible steady state boundary probabilities and probability currents the `behavior' of an open Markov process.  We accomplish this by constructing a `black-box functor' sending an open Markov process to its behavior, viewed as a morphism in the category $\LinRel$ whose objects are vector spaces and whose morphisms are linear relations. The fact that this is a functor means that the steady state behaviors of a composite open Markov process can be computed by composing the behaviors of the constituent subprocesses in $\LinRel$. 

In Chapter \ref{ch:rx} we construct a decorated cospan category where the morphisms are open reaction networks, certain types of bipartite graphs providing a graphical notation for a network of interacting entities. Reaction networks are used in modeling a variety of systems from chemical reactions to the spread of diseases to ecological models of interacting species. Probabilities in a Markov process evolve in time according to the linear `master equation.' We consider the situation in which the dynamics of a reaction network are given by the non-linear `rate equation' and prove a similar black-box theorem characterizing the steady state behaviors of reaction networks as subspaces defined by semialgebraic relations.

\chapter{Markov processes}\label{ch:mark}

A continuous-time Markov chain, or a Markov process, is a way to specify the dynamics of a probability distribution which is spread across some finite set of states. Probability can flow between the states. The larger the probability of being in a state, the more rapidly probability
flows out of the state. Because of this property, under certain conditions the
probability of the states tend toward an equilibrium where at any state the
inflow of probability is balanced by its outflow. The majority of the content in this chapter is standard and can be found, for instance, in Frank Kelly's book \emph{Reversibility and Stochastic Networks} \cite{Kelly}. 

In applications to statistical mechanics, we are often interested in equilibria such that the flow from $i$ to $j$ equals the flow from $j$ to $i$ for any pair of states. An equilibrium of a continuous-time Markov chain with this property is called `detailed balanced.' In a detailed balanced equilibrium, the net flow of probability vanishes between any pair of states. Probability distributions which remain constant in time via non-vanishing probability currents between pairs of states are called `non-equilibrium steady states' (NESS). Such states are of interest in the context of non-equilibrium statistical mechanics. 

Markovian or master equation systems have numerous applications across a wide-range of disciplines \cite{Gardiner, VanKampen}. We make no attempt to provide a complete review of this line of work, but here mention a few references relevant in the context of networked biological systems.  Schnakenberg, in his paper on networked master equation systems, defines the entropy production in a Markov process and shows that a quantity related to entropy serves as a Lyapunov function for master equation systems \cite{SchnakenRev}.  Schnakenberg, in his book \emph{Thermodynamic Network Analysis of Biological Systems} \cite{SchnakenBook}, provides a number of biochemical applications of networked master equation systems. Oster, Perelson and Katchalsky developed a theory of `networked thermodynamics' \cite{OPK}, which they went on to apply to the study of biological systems \cite{OPKBio}. Following the untimely passing of Katchalsky, Perelson and Oster went on to extend this work into the realm of chemical reactions \cite{OPChem}.

Starting in the 1970's, T.\ L.\ Hill spearheaded a line of research focused on what he called `free energy transduction' in biology. A shortened and updated form of his 1977 text on the subject \cite{Hill} was republished in 2005 \cite{HillDover}. Hill applied various techniques, such as the use of the cycle basis, in the analysis of biological systems, see for example his model of muscle contraction \cite{HillScience}.

\section{The master equation}
Consider a probability distribution spread across some finite set of states $V$. To specify a Markov process on $V$ is to specify probabilistic rates between every pair of states. We can think of this as an $V$ by $V$ matrix with non-negative entries, let's call this matrix of probabilistic rates $H \in \R^{V \times V}$ the \define{Hamiltonian}. An entry $H_{ij} \in [0,\infty)$ represents the rate at which probability flows from state $j \in V$ to state $i \in V$. Since probabilities are dimensionless, the units of $H_{ij}$ are simply inverse time. The diagonal entries of the matrix are determined by the off-diagonal entries via the condition that the columns of $H$ sum to zero, 
\[ \sum_i H_{ij} = 0. \]
A matrix whose off-diagonal entries are non-negative and whose columns sum to zero is called \define{infinitesimal stochastic}. The diagonal entry $H_{jj}$ is thus given by
\[ H_{jj} = - \sum_{i \neq j} H_{ij}. \] 
Noting that $H_{ij}$ is the rate at which probability flows from $j$ to $i$, one can see that the right hand side above and hence the entry $H_{jj}$ is the rate of total outflow from vertex $j$. 
Given some initial probability distribution $p(t=0) \in \R^V$, the time evolution of the probability distribution is given by the \define{master equation} which can be written in matrix form as
\[ \frac{dp}{dt}  = Hp \]
with solution
\[ p(t) = p(0) e^{Ht}. \]
Writing the master equation in component form we have
\[ \frac{dp_i}{dt} = \sum_j H_{ij} p_j. \]
Summing over all states gives
\[ \sum_i \frac{dp_i}{dt} = \sum_{i,j} H_{ij} p_j =0 \]
where the last equality follows from the fact that the columns of $H$ sum to zero. This implies the conservation of probability for systems whose time evolution is generated by an infinitesimal stochastic matrix.

\section{Markov processes from labeled, directed graphs}
There exists a graphical notation for Markov processes where the states of the process are represented by the nodes of a directed, labeled graph. The directed edges represent transitions and are labeled by non-negative numbers corresponding to transition rates between states. More precisely:

\begin{defn}
A \define{Markov process} $M$ is a diagram
\[ \xymatrix{ (0,\infty) & E \ar[l]_-r \ar[r]<-.5ex>_t  \ar[r] <.5ex>^s & V }  \]
where $V$ is a finite set of \define{vertices} which correspond to \define{states} of the Markov process, $E$ is a finite set of \define{edges} corresponding to transitions between the states, $s,t \maps E \to V$ assign to each edge its \define{source} and \define{target}, and $r \maps E \to (0,\infty)$ assigns a \define{rate constant} $r_e$ to each edge $e \in E$. 

If $e \in E$ has source $i$ and target $j$, we write $e \maps i \to j$. We sometimes summarize the above by writing $M=(V,E,s,t,r)$ for a \define{Markov process on $V$}.
\end{defn}

\begin{defn}
The \define{underlying graph} $G$ of a Markov process $M = (V,E,s,t,r)$ is the directed graph $G=(V,E,s,t)$.
\end{defn}

From a Markov process $M$ we can construct an infinitesimal stochastic Hamiltonian $H \maps \R^V \to \R^V$. If $i \neq j$ we define 
\[ H_{ij} = \sum_{e \maps j \to i} r_e. \]
The diagonal terms are defined so as to make $H$ infinitesimal stochastic
\[ H_{ii} = -\sum_j \sum_{e \maps i \to j} r_e. \]
We see that $-H_{ii}p_i$ has the interpretation of the net outflow from the $i^{th}$ state. A simple example illustrates this graphical syntax for Markov processes. Consider the graph
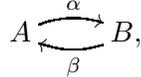
\[ \xymatrix{  A \ar@/^/[r]^{\alpha} & B \ar@/^/[l]^{\beta}, } \]
Following the above procedure we get the following Hamiltonian
\[ H = \left( \begin{array}{cc} -\alpha & \beta \\ \alpha & -\beta \end{array} \right) \]
generating the time evolution of $p(t) \in \R^2 $ via the master equation
\[ \frac{d}{dt} \left( \begin{array}{c} p_A(t) \\ p_B(t) \end{array} \right) = \left( \begin{array}{cc} -\alpha & \beta \\ \alpha & -\beta \end{array} \right) \left( \begin{array}{c} p_A(t) \\ p_B(t) \end{array} \right). \]

Given any directed, labelled graph $(V,E,s,t,r)$, following the above prescription, one can write down an infinitesimal stochastic Hamiltonian $H \maps \R^V \to \R^V$. Thus one can either think of a Markov process itself as directed, labelled graph $M=(V,E,s,t,r)$ or simply a pair $M=(V,H)$ of a finite set of states $V$ together with an infinitesimal stochastic Hamiltonian $H \maps \R^V \to \R^V$ on $V$. Many different graphs give rise to the same Hamiltonian. For instance, adding self-loops to any vertex leaves the Hamiltonian unchanged. If one wishes to put graphs in one-to-one correspondence with Hamiltonians, some choices must be made. For instance restricting one's attention to directed labelled graphs with no self-loops and at most one edge in each direction between each pair of states. In our treatment we simply allow self-loops and multiple edges in parallel between states.

\section{Flows and affinities}
The concepts introduced in this section can be found in a number of places, for instance see \cite{Kelly, SchnakenRev}. We can define a quantity representing the `flows' of probability in a Markov process. Between any pair of states $i,j \in V$ we define the \define{net flow from $j$ to $i$} $J_{ij}(p)$ as 
\[ J_{ij}(p) = H_{ij}p_j - H_{ji}p_i. \]
In the literature, this quantity is sometimes called a `probability current.' 
For a single state we can define then define the \define{net inflow}
\[ J_i (p) = \sum_j H_{ij}p_j - \sum_j H_{ji}p_i.\]
The single index on the net inflow distinguishes it from a particular flow between a pair of states.
Note that in the second term $p_i$ is independent of $j$ and for the sum in front we have, $\sum_j H_{ji} = -H_{ii}$, from the infinitesimal stochastic property of $H$. The above equation can be written as
\[ J_i(p) = \sum_j H_{ij}p_j \]
which we easily recognize as the $i^{th}$ component of the master equation
\[ \frac{dp_i}{dt} = J_i(p). \]

We also define now the `dual' variable to the flow $J_{ij}(p)$, namely the \define{affinity between state $j$ and state $i$} 
\[ A_{ij}(p) = \ln \left( \frac{ H_{ij}p_j }{H_{ji} p_i } \right). \]
People also use the term thermodynamic force for this quantity. Note that the vanishing of a flow $J_{ij}(p) = 0$ implies the vanishing of the corresponding affinity $A_{ij}(p)$. 

These quantities are essential to quantifying the difference between two distinct ways a probability distribution $q \in \R^V$ can be constant in time
\[ \frac{dq}{dt} = Hq =0\] 
in a closed Markov process, namely those for which $J_{ij}(q)$ vanishes for all pairs of states $i,j \in V$ and those for which it does not. This is the subject of the next section.

\section{Detailed balanced equilibrium versus non-equilibrium steady states}
There are essentially two types of equilibrium or steady states in a Markov process, those which satisfy detailed balance and those which don't. The latter are typically referred to as non-equilibrium steady states. All currents (and therefore affinities as well) or flows vanish in a detailed balanced equilibrium $q$, while a non-equilibrium steady state generally has non-zero currents or flows.

An equilibrium distribution $q \in \R^V$ satisfies \define{detailed balance}, if
\[ H_{ij}q_j = H_{ji}q_i \] 
for all pairs $i,j \in V$. Note that this implies not only that
\[  \frac{dq_i}{dt} = \sum_j (H_{ij}q_j -  H_{ji}q_i) = 0   \ \ \text{for all} \ \ i \in V  \]
but also that each individual term in the sum on the right hand side vanishes. A \define{non-equilibrium steady state} $q$ still satisfies
\[ \frac{dq_i}{dt} = \sum_j ( H_{ij}q_j - H_{ji} q_i )=0 \ \ \text{for all} \ \ i \in V, \]
but each individual term in the sum on the right hand side need not vanish.

This terminology can be confusing as often the use of the term equilibrium is meant to imply detailed balance. In addition, non-equilibrium steady states are often just called steady states despite the fact that both detailed balanced equilibria and non-equilibrium steady states are constant in time.

For a detailed balanced equilibrium $q$ we have that $J_{ij}(q) = 0$ for all pairs $i,j \in V$, i.e. the flows vanish along all edges in the underlying graph. Similarly, all affinities vanish for a detailed balanced equilibrium.  If $q$ is a non-equilibrium steady state we have $J_{ij}(q) \neq 0 $ for some $i,j \in V$.  Similarly a non-equilibrium steady state implies at least one non-vanishing affinity. Later we shall see how these properties imply zero entropy production in a detailed balanced equilibrium as well as non-zero entropy production in a non-equilibrium steady state.

The existence of a detailed balanced equilibrium for a given Markov process amounts to a condition on the transition rates of the Markov process. A necessary and sufficient condition for the existence of a detailed balanced equilibrium distribution is \define{Kolmogorov's Criterion}, which says that the transition rates of the Markov process satisfy
\[ H_{12}H_{23} \cdots H_{n-1 n} H_{n1}  = H_{1n} H_{n n-1}\cdots H_{32}H_{21}  \]
for all finite sequences of states $1,2,\dots,n \in V$. This condition says that the product of the transition rates around any cycle is equal to the product of the rates along the reversed cycle. 

\section{Some simple Markov processes }
Consider the following simple example of a Markov process with two states:
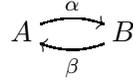
\[ \xymatrix{  A \ar@/^/[r]^{\alpha} & B \ar@/^/[l]^{\beta} } \]
The flows between the states are given by $J_{AB}(p) = -J_{BA}(p) = -\alpha p_A + \beta p_B$. This process admits a detailed balanced equilibrium distribution
\[ q = \left( \begin{array}{c} \beta \\ \alpha \end{array} \right), \]
where one can easily check that $H_{AB}q_B = H_{BA} q_A$. This equilibrium is unique up to normalization. We can see that this process trivially satisfies Kolmogorov's criterion in that $H_{AB}H_{BA} = H_{BA}H_{AB}$. In order for a Markov process to admit a non-equilibrium steady state, it must have a cycle for which Kolmogorov's criterion is violated. To illustrate such a case, consider the graph
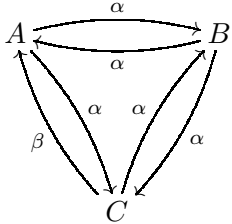
\[ \xymatrix{  A \ar@/^/[rr]^{\alpha} \ar@/^/[ddr]^{\alpha}& &  B \ar@/^/[ll]^{\alpha} \ar@/^/[ddl]^{\alpha} \\ && \\  &  C \ar@/^/[uul]^{\beta} \ar@/^/[uur]^{\alpha} &} \]
Notice that the products of the rates around the cycle in the clockwise direction is $\alpha^2 \beta$ and $\alpha^3$ in the counter-clockwise direction. Detailed balance would require that $p_A = p_B$ since the transition rate between $A$ and $B$ are equal, similarly detailed balance between $B$ and $C$ would require $p_ B = p_C$, which implies $p_A = p_C$. So long as $\alpha \neq \beta$ there will be some net flow between $A$ and $C$, indicating that there is no detailed balanced equilibrium for such a process. This system does admit a non-equilibrium steady state found by solving the set of equations $Hp^*=0$ where $H$ is the Hamiltonian for this system. Doing so yields the non-equilibrium steady state:
\[ p^* = \left( \begin{array}{c}  \alpha^2 + 2 \alpha \beta  \\  2 \alpha^2 + \alpha \beta \\ 3 \alpha^2 \end{array} \right). \]
Calculating the steady state currents $J_{ij}(p^*)$ results in
\[  J_{AB}(p^*) = J_{BC}(p^*) = J_{CA}(p^*) = \alpha^3 - \alpha^2 \beta  \]
which are indeed non-zero so long as $ \alpha \neq \beta$, in which case they all vanish. 

\section{Entropy production}
In his seminal paper on networked master equation systems \cite{SchnakenRev}, Schnakenberg defines a quantity which he calls the `rate of entropy production' to be one-half the product of the flow between each pair of states times the corresponding affinity, summed over all pairs of states
\[ \frac{1}{2} \sum_{i,j \in V} J_{ij}(p)A_{ij}(p). \]
By definition this rate of entropy production is non-negative.

For a detailed balanced equilibrium $q \in \R^V$, all currents vanish $J_{ij}(q)= 0$ for all $i,j \in V$ as do all affinities. Therefore the rate of entropy production as defined by Schnakenberg is zero in a detailed balanced equilibrium while in a non-equilibrium steady state it will generally be non-zero. One should note that this quantity which Schnakenberg calls the rate of entropy production is not (the negative of) the time derivative of the Shannon entropy
\[ S(p) = -\sum_i p_i \ln p_i, \]
rather it is in fact the time derivative of the entropy of the distribution $p$ relative to a detailed balanced equilibrium $q$. Clearly this characterization only applies to processes admitting detailed balanced equilibrium. In Section \ref{sec:relentropy} we introduce this `relative entropy' or `Kullback-Leibler divergence' and explain its relation to Schnakenberg's rate of entropy production and its role as a non-equilibrium free energy.  

Consider the quantity $A_{ij}(p)$. It is the entropy production per unit flow from $j$ to $i$. If $J_{ij}(p)>0$, i.e. if there is a positive net flow of probability from $j$ to $i$, then $A_{ij}(p)>0$. In the realm of statistical mechanics, we can understand $A_{ij}(p)$ as the force resulting from a difference in chemical potential. Let us elaborate on this point to clarify the relation of our framework to the language of chemical potentials used in non-equilibrium thermodynamics. Suppose that we are dealing with only a single type of molecule or chemical species. The states could correspond to different locations of the molecule, for example in membrane transport. Another possibility is that each state correspond to a different internal configuration of the molecule. In this setting the chemical potential $\mu_i$ is related to the concentration of that chemical species in the following way:
\[ \mu_i = \mu_i^o + T \ln( c_i), \]
where $T$ is the temperature of the system in units where Boltzmann's constant is equal to one and $\mu_i^o$ is the standard chemical potential. The difference in chemical potential between two states gives the force associated with the flow of probability which seeks to reduce this difference in chemical potential 
\[ \mu_j - \mu_i = \mu_j^o - \mu_i^o + T \ln\left( \frac{c_j}{c_i} \right). \] 
In general the concentration of the $i^{\text{th}}$ state is proportional to the probability of that chemical species divided by the volume of the system $c_i = \frac{p_i}{V}$. In this case, the volumes cancel out in the ratio of concentrations and we have this relation between chemical potential differences and probability differences:
\[ \mu_j - \mu_i = \mu_j^o - \mu_i^o + T \ln \left( \frac{p_j}{p_i} \right). \]
This potential difference vanishes when $p_i$ and $p_j$ are in equilibrium and we have
\[ 0 = \mu_j^o - \mu_i^o + T \ln \left( \frac{q_j}{q_i} \right), \]
or that
\[ \frac{q_j}{q_i} = e^{-\frac{\mu_j^o-\mu_i^o}{T}}. \]
If $q$ is a detailed balanced equilibrium, then this also gives an expression for the ratio of the transition rates $\frac{H_{ji}}{H_{ij}}$ in terms of the standard chemical potentials. Thus we can translate between differences in chemical potential and ratios of probabilities via the relation
\[ \mu_j - \mu_i = T \ln \left( \frac{p_j q_i}{q_j p_i} \right), \]
which if $q$ is a detailed balanced equilibrium gives
\[ \mu_j - \mu_i = T \ln \left( \frac{H_{ij}p_j}{H_{ji}p_i} \right). \]
We recognize the right hand side as the force $A_{ij}(p)$ times the temperature of the system $T$:
\[ \frac{\mu_j - \mu_i}{T} = A_{ij}(p) .\]

Returning to the simple three-state example whose rates violate Kolmogorov's criterion
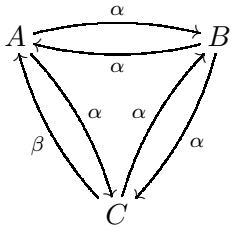
\[ \xymatrix{  A \ar@/^/[rr]^{\alpha} \ar@/^/[ddr]^{\alpha}& &  B \ar@/^/[ll]^{\alpha} \ar@/^/[ddl]^{\alpha} \\ && \\  &  C \ar@/^/[uul]^{\beta} \ar@/^/[uur]^{\alpha} &} \]
This process has a non-equilibrium steady state 
\[ p^* = \left( \begin{array}{c}  \alpha^2 + 2 \alpha \beta  \\  2 \alpha^2 + \alpha \beta \\ 3 \alpha^2 \end{array} \right) \]
which induces the steady state flows
\[  J_{AB}(p^*) = J_{BC}(p^*) = J_{CA}(p^*) = \alpha^3 - \alpha^2 \beta  \]
and affinities
\[\begin{array}{ccc} A_{AB}(p^*) = \ln \left( \frac{2 \alpha^3 + \alpha^2 \beta}{\alpha^3 + 2 \alpha^2 \beta } \right), & A_{BC}(p^*) = \ln \left( \frac{ 3\alpha^3}{2 \alpha^3 + \alpha^2 \beta} \right), & A_{CA}(p^*) = \ln \left( \frac{ \alpha^3 + 2 \alpha^2 \beta}{3 \alpha^2 \beta} \right). \end{array} \] 

With these quantities in hand we can calculate Schnakenberg's rate of entropy production in the non-equilibrium steady state $p^*$ as (for notational ease, suppressing the dependence on $p^*$)
\[ \begin{array}{ccc} \displaystyle{ \frac{1}{2} \sum_{i,j \in V} J_{ij}A_{ij} } &=& J_{AB}A_{AB}+J_{BC}A_{BC}+ J_{CA}A_{CA} \\
&=& \displaystyle{ \left( \alpha^3 - \alpha^2 \beta \right) \ln \left( \frac{\alpha}{\beta} \right) } \\
&=& \displaystyle{ \alpha^2 \left( \alpha - \beta \right) \ln \left( \frac{\alpha}{\beta} \right) }. \end{array} \]
From this expression it is easy to see that this quantity is positive unless $\alpha = \beta$ in which case it vanishes.

\section{Spanning trees and partition functions}
There is an explicit formula for steady state distributions of the Markov process associated to a graph in terms of a sum over the directed spanning trees of the graph, stemming from Kirchhoff's matrix tree theorem due to Tutte \cite{Tutte}. For a detailed description see Hill, 1966 \cite{HillTree}. Similar techniques are also utilized in Bott and Mayberry (1954), King and Altman (1956), and Schnakenberg (1976) \cite{Bott, King, SchnakenRev}.  Consider the underlying graph of a Markov process $G$. If $G$ is a tree with at most a single edge between any pair of states, then to each vertex $ i \in V$, we get a directed spanning tree $T_i$, where one has to make a choice whether all branches of the tree are directed towards or away from $i$. Since the outflow of a state in a Markov process is proportional to the probability of the state, the inflow necessary to keep the probability of the state constant in time depends on the rate of inflow from other states, hence we choose the convention that the directed tree $T_i$ have all edges directed \emph{towards} the $i^{th}$ state. The $i^{th}$ component of an equilibrium distribution of the Markov process is then given by the product of the rates along the edges of the directed tree 
\[ q_i = \prod_{e \in T_i} r_e. \]

Now let us relax the condition that $G$ be a tree with at most a single edge between vertices, but keep it connected. The graph $G$ will still have some set of spanning trees. For each tree in that set, each vertex gives a directed tree, again with the convention taken towards the vertex. Let us denote the set of spanning trees directed to the $i^{\text{th}}$ vertex as $\mathcal{T}_i$. The equilibrium distribution is given by
\[ q_i = \sum_{T \in \mathcal{T}_i} \prod_{e \in T} r_e. \]
Normally one normalizes such a distribution so that $\sum_i q_i = 1$, yielding
\[ q_i  = \frac{ \displaystyle{ \sum_{T \in \mathcal{T}_i } \prod_{e \in T} r_e }}{ \displaystyle{ \sum_i \sum_{T \in \mathcal{T}_i } \prod_{e \in T} r_e } }. \]
If we write the equilibrium distribution as a kind of Gibbs state 
\[ q_i = \frac{ \displaystyle{ \sum_{T \in \mathcal{T}_i } \prod_{e \in T} r_e } }{\mathcal{Z}} \]
then the normalizing factor 
\[ \mathcal{Z} = \sum_i \sum_{T \in \mathcal{T}_i} \prod_{e \in T} r_e \]
plays the role of a partition function \define{partition function}.

If all of the rates of a Markov process are equal to one we see that this partition function simply counts the number of spanning trees $|\mathcal{T}|$ of the underlying graph of the Markov process. 

Notice that we can think of this in terms of paths. A directed tree $T_i$ is a simple path in which `all roads lead to the $i^{th}$ state.' The rate associated to path is the product of all the rates in the path. Recall that the rate of outflow of any state is the probability associated to that state times the sum of the rates of all edges leaving that state. In equilibrium this rate of of outflow must be balanced by the rate of inflow at all states. The rate of inflow into a state is given by the sum of all incoming flows. Each incoming flow is given by the rate along that edge times the probability at its source. A similar argument holds for the probability at this source and so on. This in some sense motivates the above expression. 

When cycles are present in the graph the number of paths leading to some vertices can become infinite. Also one cannot orient a cycle `towards' a vertex along that cycle. Instead we consider only the paths along the spanning trees of the graph. A simple cycle involving $n$ vertices will have $n$ edges and $n$ spanning trees. In general the cycles of a graph are not independent. 

Let us see this how this formula works for the three-state process from our previous example
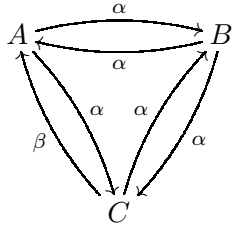
\[ \xymatrix{  A \ar@/^/[rr]^{\alpha} \ar@/^/[ddr]^{\alpha}& &  B \ar@/^/[ll]^{\alpha} \ar@/^/[ddl]^{\alpha} \\ && \\  &  C \ar@/^/[uul]^{\beta} \ar@/^/[uur]^{\alpha} &} \]
The undirected graph underlying this process has three spanning trees
\[ T = \left \{ \begin{array}{ccc} \xymatrix@=0.5em{ A \ar@{-}[rr] \ar@{-}[ddr] & & B \\ & & \\ & C ,} & \xymatrix@=0.5em{ A \ar@{-}[rr]  & & B \ar@{-}[ddl] \\ & & \\ & C ,}& \xymatrix@=0.5em{ A  \ar@{-}[ddr] & & B \ar@{-}[ddl] \\ & & \\ & C }   \end{array} \right \} \]
For each vertex $i \in V$ we get the set $\mathcal{T}_i $ of trees directed towards the $i^{\text{th}}$ state. For instance, the set of spanning trees directed towards $A$ is
\[ \mathcal{T}_A = \left \{ \begin{array}{ccc} \xymatrix@=0.5em{ A   & & B \ar[ll]_{\alpha} \\ & & \\ & C \ar[uul]^{\beta},} & \xymatrix@=0.5em{ A   & & B \ar[ll]_{\alpha} \\ & & \\ & C \ar[uur]_{\alpha} ,}& \xymatrix@=0.5em{ A  & & B \ar[ddl]^{\alpha} \\ & & \\ & C \ar[uul]^{\beta} }   \end{array} \right \} \]

We can then calculate the $A^{\text{th}}$ component of the steady state distribution from the formula 
\[ p^*_A = \frac{ \displaystyle{ \sum_{T \in \mathcal{T}_A } \prod_{e \in T} r(e) } }{ \mathcal{Z} } \]

Choosing the normalization $\mathcal{Z} = 1$ gives
\[ p^*_A = \alpha \beta + \alpha^2 + \alpha \beta = \alpha^2 + 2 \alpha \beta. \]
Following a similar procedure for states $B$ and $C$, one can check that this agrees with our previous calculation solving $\frac{dp^*}{dt} = Hp^* = 0$. 

\section{Time reversal}
Recall that detailed balance, while a property of an equilibrium distribution, is connected through the Kolmogorov criterion to the transition rates of a process. It is their behavior under time-reversal which brings out the fundamental difference between processes admitting a detailed balanced equilibrium and those with non-equilibrium steady state. Since all probability currents vanish in a detailed balanced equilibrium, time-reversal leaves the both the distribution and the currents unchanged. On the other hand, reversing time for a process in a non-equilibrium steady state will leave the distribution unchanged while reversing the sign of all non-zero probability currents. Thus we see that there is a fundamental connection between time-reversal invariance of a process and Kolmogorov's criterion. 

We saw that one can take as a definition of the rate of entropy production in a Markov process the product of the currents times the affinities, appropriately summed \cite{SchnakenRev}. We also saw that this rate of entropy production vanishes in a detailed balanced equilibrium as all currents vanish. Entropy production in a non-equilibrium steady state will generally be non-zero. Entropy production serves as a measure of the ability of a system to perform information processing or the ability of a system to consume or create free energy. Thus non-equilibrium steady states are of interest in biology and chemistry where systems organize and process information, feeding on free energy from their environment. At the same time, it seems unphysical at a microscopic level to abandon the condition of microscopic reversibility. In the next section we demonstrate that by coupling certain `boundary' states of a Markov process to the environment or to external reservoirs one can induce non-equilibrium steady states in systems which are themselves microscopically reversible, i.e. whose rates still satisfy Kolmogorov's criterion. 

\chapter{Second Laws for open Markov processes}\label{ch:entropymark}

In this chapter we define the notion of an `open Markov process' in which probability can flow in and out of the process through certain boundary states. As we saw earlier, one quantity central to the study of non-equilibrium systems is the rate of entropy production \cite{GP, DeGrootM, LJ, Lindblad, Prigogine, SchnakenRev}. For open Markov processes, we prove that rate of change of the `relative entropy' between two distributions is bounded by the flow of relative entropy through the boundary of the process. Certain boundary flows through an open Markov process induce non-equilibrium steady states. In such a non-equilibrium steady state, the rate of change of relative entropy with respect to the underlying equilibrium state is the rate at which the system must consume free energy from its environment to maintain such a steady state.

Prigogine's principle of minimum entropy production \cite{PrigogineEnt} asserts that for non-equilibrium steady states that are near equilibrium, entropy production is minimized. This is an approximate principle that is obtained by linearizing the relevant equations about an equilibrium state. In fact, for open Markov processes, non-equilibrium steady states are governed by a \emph{different} minimum principle that holds \emph{exactly}, arbitrarily far from equilibrium \cite{BMN, Landauer, Landauer2}. We show that for fixed boundary conditions, non-equilibrium steady states minimize a quantity we call `dissipation'. If the probabilities of the non-equilibrium steady state are close to the probability of the underlying detailed balanced equilibrium, one can show that dissipation is close to the rate of change of relative entropy plus a boundary term. Dissipation is in fact related to the Glansdorff--Prigogine criterion which states that a non-equilibrium steady state is stable if the second order variation of the entropy production is non-negative \cite{GP, SchnakenRev}.  

Starting in the 1990's, the Qians and their collaborators developed a school studying non-equilibrium steady states, publishing a number of articles and books on the topic \cite{Qians}. More recently, results concerning fluctuations have been extended to master equation systems \cite{AndrieuxGaspard}. In the past two decades, Hong Qian of the University of Washington and collaborators have published numerous results on non-equilibrium thermodynamics, biology and related topics \cite{Qian1, Qian2, Qian3}.

\section{Open Markov processes}

In this section we introduce the concept of an open Markov process. An open Markov process is a Markov process $(V,H)$ together with a specified subset $B \subseteq V$ of `boundary states.' Probability can flow in and out of the process at boundary states via some coupling to other systems or an environment. 

\begin{defn}
A \define{Markov process with boundary} is a triple $(V,H,B)$ where
\begin{itemize}
\item $V$ is a finite set of \define{states}
\item $H \maps \R^V \to \R^V$ is an infinitesimal stochastic \define{Hamiltonian}
\item $B \subseteq V$ is a subset of \define{boundary states.}
\end{itemize}
States in the subset $I = V-B$ are called \define{internal states}.
\end{defn}

Often, the state space of a system interacting with its environment is given by the product of the state spaces of the system and the environment $S \times E$. Specifying a particular state corresponds to specifying the state of the system and the state of the environment. In the context of this article we consider a different viewpoint, where the state space of the composite system is given by the union of the internal and boundary states $S=I \cup B$. Thus a particle in the composite system can be in either an internal state or a boundary state. The interaction of the system with its environment is captured by the system's behavior at boundary states. 

One can visualize an open Markov process as a graph where the edges are labelled by positive real numbers. Each vertex is a `state' and the numbers attached to the edges are transition rates. Internal states are white and boundary states are shaded:
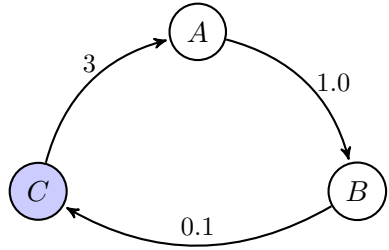
\[
\begin{tikzpicture}[->,>=stealth',shorten >=1pt,auto,node distance=3cm,
  thick,main node/.style={circle,fill=white!20,draw,font=\sffamily\bfseries},terminal/.style={circle,fill=blue!20,draw,font=\sffamily}]]
  \node[main node](1) {$A$};
  \node[main node](2) [below right of=1] {$B$};
  \node[terminal](3) [below left of=1]  {$C$};
  \path[every node/.style={font=\sffamily\small}]
    (3) edge [bend left] node[above] {$3$} (1)
    (2) edge [bend left] node[above] {$0.1$} (3)
    (1) edge [bend left] node[right] {$1.0$} (2);  
\end{tikzpicture} \]

The master equation is modified for an open Markov process. The time evolution of probabilities associated to boundary states is prescribed as the `boundary conditions' of the process. The internal states still evolve according to their regular master equation. 
\begin{defn} The \define{open master equation} is given by
\[ \begin{array}{ccll}\displaystyle{ \frac{d}{dt}p_i(t) } &=& \displaystyle{
\sum_j H_{ij} p_j(t)}, &  i \in V-B \\  \\
 p_i(t) &=& b_i(t), & i \in B.  
\end{array}\]
where the time-dependent $b_i(t)$ are specified.
\end{defn}

\begin{defn}
An \define{ordinary Markov process} $(V,H)$ is an open Markov process whose boundary is empty. 
\end{defn}

Since probability can flow in and out through the boundary of an open Markov process, the total probability $\sum_i p_i$ need not be constant. Thus, even if this sum is initially normalized to 1, it will typically not remain so. If an open Markov process is a subsystem of an ordinary Markov process, one can normalize the probabilities of the ordinary Markov process to unity and interpret the $p_i$ as probabilities. In this case the probabilities restricted to the open Markov process will be sub-normalized, obeying $\sum_i p_i \leq 1$. However, in some applications it is useful to interpret the quantity $p_i$ as a measure of the number of `entities' in the $i$th state, with the number being so large that it is convenient to treat it as varying continuously rather than discretely \cite{Kingman}. In these applications we do not have $\sum_i p _i \leq 1$ and therefore do not assume so throughout this thesis, instead working with un-normalized probabilities $p_i \in [0,\infty)$ for which the sum $\sum_i p_i$ converges. 

We can write down the open master equation for the open Markov process depicted above as 
\[ \frac{dp_A}{dt} = -1.0p_A + 3p_c \]
\[ \frac{dp_B}{dt} = -0.1p_B + 1.0p_A \]
\[ p_C = b(t) \]
where $b(t)$ is the specified boundary probability for the single boundary state $C$.

We will be especially interested in `steady state' solutions of the open master equation:

\begin{defn}
A \define{steady state} solution of the open master equation is a solution $p(t) \maps V \to [0,\infty)$ such that $\frac{dp}{dt} = 0$. 
\end{defn}

Let us take the boundary probability $p_C = b(t)$ to be constant in time at some fixed value $p_C = C_0$. Then we can seek steady state solutions of the open master equation compatible with this boundary probability. Solving,
\[ \frac{dp_A}{dt} = -1.0p_A + 3C_0 = 0 \]
\[ \frac{dp_B}{dt} = -0.1p_B + 1.0p_A = 0 \]
yields
\[ p_A = 3C_0 \]
and 
\[ p_B = 30C_0.\]

Externally fixing the boundary probabilities of an open Markov process whose underlying Markov process satisfies Kolmogorov's criterion will induce a steady-state distribution in which the inflow of probability is balanced by its outflow at all internal states. Hence it is often the case that we restrict our attention to open Markov processes $(V,H,B)$ whose underlying Markov process $(V,H)$ satisfies Kolmogorov's criterion and therefore admitting a detailed balanced equilibrium $q \in \R^V$.

\begin{defn}
An open detailed balanced Markov process $(V,H,B,q)$ is a open Markov process $(V,H,B)$ equipped with a particular detailed balanced equilibrium $q \in \R^V$. Note that the existence of a detailed balanced equilibrium requires that the rates in the Hamiltonian $H \maps \R^V \to \R^V$ satisfy Kolmogorov's criterion. 
\end{defn}

\section{Relative entropy}

We now introduce a divergence between probability distributions known in various circles as the relative entropy, relative information or the Kullback-Leibler divergence. The relative entropy is not symmetric and violates the triangle inequality, which is why it is called a `divergence' rather than a metric, or distance function. 

Given two probability distributions $p,q \in [0,\infty)^V$ the entropy of $p$ relative to $q$ or the \define{relative entropy} is given by:
\[ \displaystyle{ I(p,q) = \sum_{i \in V} p_i \ln \left( \frac{p_i}{q_i} \right). } \]
The relative entropy is sometimes referred to as the information gain or the Kullback--Leibler divergence \cite{kullback1951information}. Moran, Morimoto, and Csiszar proved that, in an ordinary Markov process, the entropy of any distribution relative to the equilibrium distribution is non-increasing \cite{csiszar1963,moran1961entropy,morimoto1963markov}. Dupuis and Fischer proved that the relative entropy between any two distributions satisfying the master equation is non-increasing \cite{dupuis2012construction}. 

For normalized distributions $\sum_i p_i = \sum_i q_i = 1$ the relative entropy $I(p,q)$ enjoys the property that $I(p,q) > 0$ unless $p=q$ where it vanishes. Since we are interested in studying entropy production for open Markov processes in which the distributions are un-normalized we use a generalized version of the relative entropy
\[ I(p,q) = \sum_i \left[ p_i \ln \left( \frac{p_i}{q_i} \right) - (p_i - q_i) \right] \]
for which $I(p,q) \geq 0$ unless $p=q$ for any pair of un-normalized distributions $p,q$. 

First we show that the results regarding the non-increase of relative entropy still hold for this generalized relative entropy and for un-normalized distributions. Following Dupuis and Fischer \cite{dupuis2012construction}, we can see that relative entropy is non-increasing for ordinary Markov processes:
\[ \begin{array}{ccl} \displaystyle{ \frac{d I( p(t), q(t) ) } {dt} } &=& \displaystyle{ \frac{d}{dt} \sum_i \left[ p_i \ln \left( \frac{p_i}{q_i} \right) -(p_i - q_i) \right]  } \\ \\
&=& \displaystyle{  \sum_i \left[ \frac{dp_i}{dt} \ln \left( \frac{p_i}{q_i} \right) -  \frac{dq_i}{dt} \left( \frac{p_i}{q_i} -1 \right) \right] } \\ \\
&=& \displaystyle{  \sum_{i,j} \left[ H_{ij} p_j \ln \left( \frac{p_i}{q_i} \right) - H_{ij} q_j \left( 
\frac{p_i}{q_i} - 1 \right) \right] } \\ \\
\end{array} \]
The last term is zero for an ordinary Markov process. Separating out the $i=j$ term, we have
\[ \begin{array}{ccl} \displaystyle{ \frac{d I( p(t), q(t) ) } {dt} } 
&=& \displaystyle{  \sum_{i, j \neq i} \left[ H_{ij} p_j  \ln \left( \frac{p_i}{q_i} \right) - H_{ij}q_j \frac{p_i}{q_i} \right] + \sum_i \left[ H_{ii}p_i \ln \left( \frac{p_i}{q_i} \right) - H_{ii}p_i  \right] } \\ \\ 
&=& \displaystyle{  \sum_{i, j \neq i} \left[ H_{ij} p_j  \ln \left( \frac{p_i}{q_i} \right) - H_{ij}q_j \frac{p_i}{q_i} \right] + \sum_j \left[ H_{jj}p_j \ln \left( \frac{p_j}{q_j} \right) - H_{jj}p_j  \right] } \\ 
\end{array} \]
\[ \begin{array}{ccl}
&=& \displaystyle{  \sum_{i, j \neq i} \left[ H_{ij} p_j  \ln \left( \frac{p_i}{q_i} \right) - H_{ij}q_j \frac{p_i}{q_i} \right] - \sum_{j, i \neq j} \left[ H_{ij}p_j \ln \left( \frac{p_j}{q_j} \right) - H_{ij}p_j  \right] }  \\ \\
&=& \displaystyle{  \sum_{i,j} H_{ij} p_j \left[ \ln \left( \frac{p_i q_j}{q_i p_j} \right) - \frac{p_i q_j}{q_i p_j} +1 \right] } \\ \\
&\leq& 0. \end{array} \]
The last line follows from the fact that $H_{ij} \geq 0$ for $i \neq j$ along with the fact that the term in the brackets $\ln(x)-x+1$ 
is everywhere negative except at $x=1$ where it is zero. As $q_i \rightarrow 0$ for some $i \in V$, the rate of change of relative entropy tends towards negative infinity. One has to allow infinity as a possible value for relative entropy and negative infinity as a possible value for its first time derivative, in which case the above inequality still holds.  Thus, we conclude that for any ordinary Markov process,
\[ \displaystyle{ \frac{d}{dt} I(p(t),q(t)) \leq 0}.  \]
This inequality is the continuous-time analog of the generalized data processing lemma \cite{cohen1993relative,cohenmajorization}. It holds for any two, un-normalized distributions $p$ and $q$ which obey the same master equation.

Merhav argues that the Second Law of thermodynamics can be viewed as a special case of the monotonicity in time of the relative entropy in Markov processes \cite{merhav2011data}. There is an unfortunate sign convention in the definition of relative entropy: while entropy is typically increasing, relative entropy typically decreases. The reason for using relative entropy instead of the usual Shannon entropy $ S(p) = -\sum_i p_i \ln(p_i) $ is that the usual entropy is not necessarily a monotonic function of time in Markov processes. If a Markov process has the uniform distribution $q_i = c$ for all $i$ and for some constant $ c \geq 0$ as its equilibrium distribution, then the usual entropy will increase \cite{moran1961entropy}. In this case, the relative entropy becomes
\[ \displaystyle{ I(p,q) = \sum_i p_i \ln(p_i) - \sum_i p_i \ln(c) -\sum_i (p_i - c). } \]
If $\sum_i p_i $ is constant, then for $q$ uniform, the relative entropy equals the negative of the usual entropy plus or minus some constant. Thus the above calculation for $ \frac{d I(p,q) }{dt}$ gives the usual Second Law.  

A Markov process has the uniform distribution as its equilibrium distribution if and only if its Hamiltonian is \define{infinitesimal doubly stochastic}, meaning that both the columns and the rows sum to zero. Relative entropy is non-increasing even for Markov processes whose equilibrium distribution is not uniform \cite{cover1994processes}. This suggests the importance of a deeper underlying idea, that of the Markov ordering on the probability distributions themselves; see \cite{alberti1982stochasticity, gorban2010entropy} for details. 

\section{Relative entropy change in open Markov processes} \label{sec:relentropy}

We now prove a number of inequalities bounding the rate of change of various relative entropies in open Markov processes. Roughly speaking these inequalities say that relative entropy can increase only along the boundary of an open Markov processes. 

In an open Markov process, the sign of the rate of change of relative entropy is indeterminate. Consider an open Markov process $(V,B,H)$. For any two probability distributions $p(t)$ and $q(t)$ which obey the open master equation let us introduce the quantities
\[ \frac{Dp_i}{Dt} = \frac{d p_i}{dt} - \sum_{j \in V} H_{ij} p_j \]
and
\[ \frac{Dq_i}{Dt} = \frac{d q_i}{dt} - \sum_{j \in V} H_{ij} q_j, \]
which measure how much the time derivatives of $p(t)$ and $q(t)$ fail to obey the master equation. Notice that $\frac{Dp_i}{Dt} = 0 $ for $ i \in V-B$, as the probabilities of internal states evolve according to the master equation. Also note that the relative entropy
\[ I(p,q) = \sum_i \left[ p_i \ln \left( \frac{p_i}{q_i} \right) - (p_i - q_i ) \right] \]  satisfies the following relations:
\[ \displaystyle{ \frac{\partial I(p,q) }{\partial p_i} = \sum_i  \ln \left(\frac{p_i}{q_i} \right)  } \]
and
\[ \displaystyle{ \frac{\partial I(p,q) } { \partial q_i} = \sum_i \left( 1- \frac{p_i}{q_i} \right). } \]
In terms of these quantities, the rate of change of relative entropy for an open Markov process can be written as
\[ \frac{d}{dt} I(p(t),q(t)) = \sum_{i,j \in V} H_{ij}p_j \left( \ln \left(\frac{p_i}{q_i} \right) - \frac{p_i q_j}{q_i p_j} \right) + \sum_{i \in B} \frac{Dp_i}{Dt} \frac{ \partial I}{\partial p_i} + \frac{Dq_i}{Dt} \frac{\partial I}{\partial q_i}. \]
The first term is the rate of change of relative entropy for a closed or ordinary Markov process, which as we saw above is less than or equal to zero.
\begin{thm}
Given distributions $p(t), q(t) \in [0,\infty)^V$, the rate of change of relative entropy $I(p,q)$ in an open Markov process $(H,V,B)$ satisfies 
\[ \frac{d}{dt} I(p(t), q(t)) \leq \sum_{i \in B} \frac{Dp_i}{Dt} \frac{ \partial I}{\partial p_i} + \frac{Dq_i}{Dt} \frac{\partial I}{\partial q_i}.\]
\end{thm}
\begin{proof}
Taking the time derivative of the relative entropy we obtain
\[  \begin{array}{ccl} \displaystyle{ \frac{d}{dt} I( p(t), q(t) )  } &=& \displaystyle{ \sum_{ i \in V} \frac{dp_i}{dt} \ln \left( \frac{p_i}{q_i} \right)  + \sum_{i \in V} \frac{dq_i}{dt} \left( 1 - \frac{p_i}{q_i}  \right) } \\ \\
&=& \displaystyle{ \sum_{i \in V-B} \sum_{j \in V} \left[ H_{ij} p_j  \ln \left( \frac{p_i}{q_i} \right) + H_{ij} q_j \left (1-\frac{p_i}{q_i} \right) \right] } \\ \\
& & \displaystyle{ + \sum_{i \in B} \left[ \frac{dp_i}{dt}  \ln \left( \frac{p_i}{q_i} \right)  +  \frac{dq_i}{dt} \left( 1 - \frac{p_i}{q_i} \right) \right] }. \\ \\ \end{array} \]
In the last step we separated the contributions from internal and boundary states and used the master equation for the internal states. Now let us add and subtract terms so that the first term corresponds to the rate of change of relative entropy for a Markov process with no boundary states:
\[ \begin{array}{ccl} \displaystyle{ \frac{d}{dt} I( p(t), q(t) ) } 
&=& \displaystyle{ \sum_{i \in V} \sum_{j \in V} \left[ H_{ij} p_j  \ln \left( \frac{p_i}{q_i} \right) + H_{ij} q_j \left( 1- \frac{p_i}{q_i} \right)  \right] } \\ \\
& & \displaystyle{ + \sum_{i \in B} \sum_{j \in V} \left( \frac{dp_i}{dt} - H_{ij}p_j \right)  \ln \left( \frac{p_i}{q_i} \right)  } \\ \\
& & \displaystyle{ + \sum_{i \in B} \sum_{j \in V}  \left(\frac{dq_i}{dt} - H_{ij} q_j \right) \left( 1-\frac{p_i}{q_i} \right) }.
\end{array} \]
The first term is the rate of change of relative entropy for an ordinary Markov process, which we saw is less than or equal to zero. Therefore, we have
\[ \begin{array}{ccl} \displaystyle{ \frac{d}{dt} I ( p(t),q(t) ) } 
&\leq& \displaystyle{  \sum_{i \in B} \sum_{j \in V}  \left( \frac{dp_i}{dt} - H_{ij}p_j \right)  \ln \left( \frac{p_i}{q_i} \right) } \\
& & \displaystyle{ - \sum_{i \in B} \sum_{j \in V} \left(\frac{dq_i}{dt} - H_{ij} q_j \right) \left( 1 - \frac{p_i}{q_i} \right)  }. \end{array} \]
We can write this more compactly as
\[ \begin{array}{ccl} \displaystyle{ \frac{d}{dt} I ( p(t),q(t) ) }
&\leq& \displaystyle{  \sum_{ i \in B} \frac{Dp_i}{Dt}\frac{\partial I}{\partial p_i} + \frac{Dq_i}
{Dt}\frac{\partial I}{\partial q_i} } \end{array}. \]
\end{proof}

This gives a version of the Second Law that holds for open Markov processes. This inequality tells us that the rate of change of relative entropy in an open Markov processes is bounded by the rate at which relative entropy flows through its boundary. If $q$ is an equilibrium solution of the master equation 
\[ \frac{dq}{dt} = Hq = 0, \] then the rate of change of relative entropy can be written as
\[ \frac{d}{dt} I(p(t),q) = \sum_{i,j \in V} ( H_{ij}p_j - H_{ji}p_i) \ln\left( \frac{p_i q_j}{q_i p_j} \right
) + \sum_{i \in B} \frac{Dp_i}{Dt} \frac{ \partial I}{\partial p_i} \]
Furthermore, if $q$ satisfies detailed balance we can write this as
\[ \frac{d}{dt} I(p(t),q) = -\frac{1}{2} \sum_{i,j \in V}  (H_{ij} p_j - H_{ji} p_i ) \ln \left( \frac{H_{ij}p_j}{H_{ji} q_i} \right) + \sum_{i \in B} \frac{Dp_i}{Dt} \frac{ \partial I}{\partial p_i}., \]
where we recognize the quantities
\[ J_{ij}(p) = H_{ij}p_j - H_{ji}p_i \]
and 
\[ A_{ij}(p) = \ln \left( \frac{H_{ij}p_j}{H_{ji}p_i} \right) \]
as the \define{flux} from $j$ to $i$ and the conjugate \define{affinity}. For $q$ detailed balanced, we have that
\[ \frac{d}{dt} I(p(t),q) -\frac{1}{2} \sum_{i,j \in V} J_{ij}(p) A_{ij}(p) + \sum_{i \in B} \frac{Dp_i}{Dt} \frac{\partial I}{\partial p_i} \].

Thus we see that for an ordinary Markov process whose boundary is empty, the quantity which Schnakenberg called the rate of entropy production associated to a distribution $p$ is in fact the rate of change of the relative entropy of $p$ with respect to a detailed balanced equilibrium distribution $q$.  We now relate $I(p(t),q)$ for a detailed balanced equilibrium $q$ to a `non-equilibrium free energy.'

\section{Relative entropy as a free energy difference}
The possibility of increasing relative entropy is a generic feature of interacting systems. For a closed system, relative entropy can increase within a particular subsystem, but as was shown in section 3.1 this increase will always be compensated by a decrease elsewhere in the system. This is analogous to the case of entropy in thermodynamics. The generalization of the Second Law to the type of open systems described in this article can be applied to non-equilibrium thermodynamic systems where external forcings at boundary states maintain the system out of equilibrium. 

Consider the case of an ordinary Markov process whose equilibrium distribution $q$ satisfies detailed balance, $H_{ij}q_j = H_{ji}q_i$. If to each state we associate an energy $E_i$, then we can write the $q_i$'s as Gibbs states
\[ q_i = \frac{ e^{-\beta E_i} }{\mathcal{Z}},\]  where $\beta = \frac{1}{T}$ is the inverse temperature in units where Boltzmann's constant is equal to one and $\mathcal{Z}$ is the \define{partition function}. Consider the \define{free energy} $F[q] = \langle E \rangle_q -TS(q)$ where $S(q) = -\sum_i q_i \ln{q_i}$ is the Shannon entropy and $\langle E \rangle_q = \sum_i q_i E_i $ is the expected energy. Plugging the expression for $q_i$ as a Gibbs state into the formula for the free energy we get the relation
\[ -\beta F[q] = \sum_i q_i \ln{\mathcal{Z}} \]
which in the case that $q$ is a probability distribution reduces to the usual relation between the equilibrium free energy and the partition function.

The entropy of a non-equilibrium state $p(t)$ relative to the equilibrium $q$ is
\[ I(p(t), q) = \sum_i \left[ p_i(t) \ln \left( \frac{p_i(t)}{q_i} \right) - (p_i(t) - q_i) \right] \]
which, using the above formulas, can be written as 
\[ I(p(t),q) = -S(p) + \beta \langle E \rangle_{p(t)} - \beta F[q] \frac{\sum_i p_i}{\sum_i q_i} - \sum_i p_i(t) + \sum_i q_i . \] 
If we define the free energy of the non-equilibrium distribution $p$ as $F[p] = \langle E \rangle_p - TS(p)$ we have that 
 \[ I(p(t),q) = \frac{ F[p(t)] - F[q] \frac{\sum_i p_i(t)}{\sum_i q_i} }{T} - \sum_i p_i(t) + \sum_i q_i .\]
If we introduce a time-dependent scaling for $q$ so that $\sum_i p_i(t) = \sum_i q_i$ at all times, the relative entropy $I(p(t),q)$ is simply the amount by which the free energy of $p$ exceeds the equilibrium free energy, divided by the temperature
\[ I(p(t),q) = \frac{ F[p(t)] - F[q] }{T}. \]

\section{Principle of minimum dissipation}\label{sec:dissipation}
We are interested in steady states which arise when the boundary probabilities of an open Markov process are held externally fixed. Here we show that by externally fixing the probabilities at boundary states, one induces steady states which minimize a quadratic form which we call `dissipation.' Dissipation coincides with the rate of relative entropy production only in the vicinity of equilibrium. 

\begin{defn} Given an open detailed balanced Markov process we define the \define{dissipation functional} of a probability distribution $p$ to be
\[ D(p) = \frac{1}{2}\sum_{i,j} H_{ij}q_j \left( \frac{p_j}{q_j} - \frac{p_i}{q_i} \right)^2. \]
\end{defn}

\begin{defn}
We say a probability distribution obeys the \define{principle of minimum dissipation with boundary probability $b$} if $p$ minimizes $D(p)$ subject to the constraint that $p|_B = b$.
\end{defn} 
With this we can state the following theorem:
\begin{thm}
A probability distribution $p \in \R^V$ is a steady state with boundary probability $b \in \R^B$ if and only if $p$ obeys the principle of minimum dissipation with boundary probability $b$.
\end{thm}
\begin{proof}
Given boundary probabilities $b \in R^B$, we can minimize the dissipation functional over all $p$ which agree on the boundary. Differentiating the dissipation functional with respect to an internal probability, we get
\[ \frac{\partial D(p)}{\partial p_n} = -2\sum_j H_{nj}\frac{p_j}{q_n}. \]
Multiplying by $\frac{q_n}{2}$ yields
\[ \frac{q_n}{2} \frac{ \partial D(p)}{\partial p_n} = -\sum_j H_{nj} p_j, \]
where we recognize the right-hand side from the open master equation for internal states. We see that for fixed boundary probabilities, the conditions for $p$ to be a steady state, namely that
\[ \frac{dp_i}{dt} = 0 \ \  \text{for all} \ \ i \in V,\] is equivalent to the condition that 
\[ \frac{\partial D(p)}{\partial p_n} = 0 \ \  \text{for all} \ \ n \in V-B.\]
\end{proof}

Given specified boundary probabilities, one can compute the steady state boundary flows by minimizing the dissipation subject to the boundary conditions. Recall that the flow into the $i^{\text{th}}$ state is given by
\[ J_i(p) = \sum_j \left( H_{ij}p_j - H_{ji}p_i \right). \]
Given a steady state $p$ with boundary probabilities, this quantity will vanish on all internal states and generally be non-zero for boundary states.
\begin{defn}
We call a probability-flow pair a \define{steady state probability-flow pair} if the flows arise from a probability distribution which obeys the principle of minimum dissipation. 
\end{defn}
\begin{defn}
The \define{behavior} of an open detailed balanced Markov process with boundary $B$ is the set of all steady state probability-flow pairs $(p_B, J_B)$ along the boundary.
\end{defn}

Later on we will see that dissipation plays a fundamental role in our functorial characterization of the behaviors of open Markov processes via the `black-box theorem.'

\section{Dissipation and entropy production}

We saw that steady states of an open detailed balanced Markov process with fixed boundary probabilities minimize the dissipation. Here we show that for distributions $p$ close to the equilibrium distribution $q$, the dissipation is approximately equal to the rate of change of relative entropy. We shall see explicitly that Prigogine's principle of minimum entropy production is valid in the vicinity of equilibrium, while the principle of minimum dissipation holds for $p$ arbitrarily far from equilibrium. 
 
Let us return to our expression for $\frac{d}{dt} I(p(t),q)$ where $q$ is an equilibrium distribution:
\[ \frac{d}{dt} I(p(t),q) = -\frac{1}{2}\sum_{i,j \in V} \left( H_{ij}p_j - H_{ji}p_i \right) \ln \left( \frac{q_i p_j}{q_j p_i} \right) + \sum_{i \in B} \frac{Dp_i}{Dt}\frac{\partial I}{\partial p_i}. \]
Consider the situation in which $p$ is near to the equilibrium distribution $q$ in the sense that
\[ \frac{p_i}{q_i} = 1 + \epsilon_i \]
where $\epsilon_i \in \R$ is the deviation in the ratio $\frac{p_i}{q_i}$ from unity. We collect these deviations in a vector denoted by $\epsilon$. 
Expanding the logarithm to first order in $\epsilon$ we have that
\[ \frac{d}{dt} I(p(t),q) = -\frac{1}{2} \sum_{i,j \in V} \left(H_{ij} p_j - H_{ji}p_i \right) \left( \epsilon_j - \epsilon_i \right) + \sum_{i \in B} \frac{Dp_i}{Dt} \frac{\partial I}{\partial p_i} + O(\epsilon^2), \]
which gives
\[ \frac{d}{dt} I(p(t),q) = -\frac{1}{2} \sum_{i,j \in V} \left(H_{ij}p_j - H_{ji}p_i \right) \left(\frac{p_j}{q_j} - \frac{p_i}{q_i} \right) + \sum_{i \in B}  \frac{Dp_i}{Dt}\frac{\partial I}{\partial p_i} + O(\epsilon^2).\]
By $O(\epsilon^2)$ we mean a sum of terms of order $\epsilon_i^2$. 
When $q$ is a detailed balanced equilibrium we can rewrite this quantity as
\[ \frac{d}{dt} I(p(t),q) = -\frac{1}{2} \sum_{i,j} H_{ij}q_j \left( \frac{p_j}{q_j} - \frac{p_i}{q_i} \right)^2 + \sum_{i \in B} \frac{Dp_i}{Dt} \frac{\partial I}{\partial p_i} + O(\epsilon^2). \]
We recognize the first term as the negative of the dissipation $D(p)$ which yields
\[ \frac{d}{dt} I(p(t),q) = - D(p) + \sum_{i \in B} \frac{Dp_i}{Dt}\frac{\partial I}{\partial p_i} + O(\epsilon^2) . \]

We see that for open Markov processes, minimizing the dissipation approximately minimizes the rate of decrease of relative entropy plus a term which depends on the boundary probabilities. In the case that boundary probabilities are held fixed so that $\frac{dp_i}{dt} = 0 , \ i \in B$, we have that 
\[ \frac{Dp_i}{Dt} = -\sum_{j \in V} H_{ij} p_j, \ \ i \in B. \]
In this case, the rate of change of relative entropy can be written as
\[ \frac{d}{dt} I(p(t),q) = \sum_{i \in V-B} \frac{p_i}{q_i} \frac{dp_i}{dt} + 2 \sum_{i \in B} \frac{Dp_i}{Dt} + O(\epsilon^2). \]

Summarizing the results of this section, we have that for $p$ arbitrarily far from the detailed balanced equilibrium equilibrium $q$, the rate of relative entropy reduction can be written as 
\[ \frac{dI(p(t),q)}{dt}  = -\frac{1}{2}\sum_{i,j}  J_{ij}(p) A_{ij}(p)+ \sum_{i \in B} \frac{Dp_i}{Dt} \frac{\partial I}{\partial p_i}. \]
For $p$ in the vicinity of a detailed balanced equilibrium we have that
\[ \frac{dI(p(t),q)}{dt} = -D(p) + \sum_{i \in B} \frac{Dp_i}{Dt}\frac{\partial I}{\partial p_i} + O(\epsilon^2) \]
where $D(p)$ is the dissipation and $\epsilon_i = \frac{p_i}{q_i} - 1$ measures the deviations of the probabilities $p_i$ from their equilibrium values. We have seen that in a non-equilibrium steady state with fixed boundary probabilities, dissipation is minimized. We showed that for steady states near equilibrium, the rate of change of relative entropy is approximately equal to minus the dissipation plus a boundary term. Minimum dissipation coincides with minimum entropy production only in the limit $\epsilon \to 0$.

We now utilize a simple three-state example of an open Markov process to illustrate the difference between probabilities which minimize dissipation and those which minimize entropy production:  
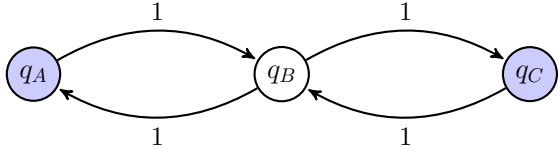
\[
\begin{tikzpicture}[->,>=stealth',shorten >=1pt,thick,scale=1.1,terminal/.style={circle,fill=blue!20,draw,font=\sffamily\Large\bfseries}]
\node[terminal, scale=.75](A) at (-2,0) {$q_A$};
\node[main node, scale=.75](B) at (1,0) {$q_B$};
\node[terminal, scale=.75](C) at (4,0) {$q_C$};

  \path[every node/.style={font=\sffamily\small}, shorten >=1pt]
    (A) edge [bend left] node[above] {$1$} (B)
    (B) edge [bend left] node[below] {$1$} (A) 
    (B) edge [bend left] node[above] {$1$} (C)
    (C) edge [bend left] node[below] {$1$} (B);


\end{tikzpicture}
\]
Here states $A$ and $C$ are boundary states, while state $B$ is internal. For simplicity, we have set all transition rates equal to one. In this case, the detailed balanced equilibrium distribution is uniform. We take $q_A = q_B = q_C = 1$. If the probabilities $p_A$ and $p_C$ are externally fixed, then the probability $p_B$ which minimizes the dissipation is simply the arithmetic mean of the boundary probabilities
\[ p_B = \frac{p_A + p_C}{2}. \]

The rate of change of the relative entropy $I(p(t),q)$ where $q$ is the uniform detailed balanced equilibrium is given by
\begin{multline*} \frac{d}{dt} I(p(t),q) = \\
\underbrace{-(p_A - p_B) \ln  (\frac{p_A}{p_B} ) - (p_B - p_C) \ln ( \frac{p_B}{p_C} )}_{-\frac{1}{2}\sum_{i,j \in V} J_{ij}A_{ij}} 
+\underbrace{(p_A - p_B) ( \ln(p_A) + 1) + (p_C- p_B) ( \ln(p_C) +1 )}_{\sum_{i \in B} \frac{Dp_i}{Dt}\frac{\partial I}{\partial p_i}}. \end{multline*}
Differentiating this quantity with respect to $p_B$ for fixed $p_A$ and $p_C$ yields the condition
\[ \frac{p_A+p_C}{2p_B} - \ln(p_B) - 2 = 0. \]
The solution of this equation gives the probability $p_B$ which extremizes the rate of change of relative entropy, namely
\[ p_B = \frac{p_A+p_C}{2W \left(\frac{(p_A+p_C)}{2} e^2 \right) }, \]
where $W(x)$ is the Lambert $W$-function or the omega function which satisfies the following relation
\[ x = W(x)e^{W(x)}. \]
The Lambert $W$-function is defined for $x \geq \frac{-1}{e}$ and double valued for $x \in [\frac{-1}{e},0)$. This simple example illustrates the difference between distributions which minimize dissipation subject to boundary constraints and those which extremize the rate of change of relative entropy. For fixed boundary probabilities, dissipation is minimized in steady states arbitrarily far from equilibrium. For steady states in the neighborhood of the detailed balanced equilibrium, the rate of change of relative entropy is approximately equal to minus the dissipation plus a boundary term.

\section{Csiszar--Morimoto entropy}
We now show that both relative entropy and the dissipation are special cases of a generalized entropy defined separately by Csiszar and Morimoto \cite{csiszar1963, morimoto1963markov} who also showed the monotonicity of such a quantity for ordinary Marokv processes. To this end we first extend our theorems of the previous sections, proving inequalities bounding the rate of change of this generalized entropy for open Markov processes. Then we show that dissipation is in fact a special case of such a generalized entropy production. This generalization of the standard relative entropy replaces the natural logarithm with any convex function. Suppose $ f \maps [0,\infty) \to \R$ is any convex function. Given two probability distributions $p, q : V \to [0,\infty)$, where $q_i \neq$ for all $i \in V$, we can define the Csiszar--Morimoto entropy or \define{f-divergence} 

\[ \displaystyle{  I_f(p,q) = \sum_i q_i f( \frac{p_i}{q_i} ) } \]

For \emph{closed} Markov processes $I_f(p,q)$ is nonincreasing.  For \emph{open} Markov processes $\frac{d}{dt} I_f(p,q)$ is bounded by the $f$-flow through the boundary.

\begin{thm}\label{thm:fdiv1}  Consider an open Markov process with $V$ as its set of states and $B$ as the set of boundary states.   Suppose $p(t)$ and $q(t)$ obey the open master equation, and let the quantities

\[ \displaystyle{ \frac{Dp_i}{Dt} = \frac{dp_i}{dt} - \sum_{j \in V} H_{ij}p_j } \]

\[ \displaystyle{  \frac{Dq_i}{Dt} = \frac{dq_i}{dt} - \sum_{j \in V} H_{ij}q_j } \]

measure how much the time derivatives of $p_i$ and $q_i$ fail to obey the master equation.   Then we have

\[  \begin{array}{ccl}   \displaystyle{  \frac{d}{dt}  I_f(p(t),q(t)) } &=& \displaystyle{ \sum_{i, j \in V} H_{ij} q_j \left( f( \frac{p_i}{q_i}) - \frac{p_i}{q_i} f'(\frac{p_i}{q_i}) \right) + H_{ij}p_j f'(\frac{p_i}{q_i})  } \\ \\
&& \; + \; \displaystyle{ \sum_{i \in B} \frac{\partial I_f}{\partial p_i} \frac{Dp_i}{Dt} +  \frac{\partial I_f}{\partial q_i} \frac{Dq_i}{Dt}
}  \end{array} \]
\end{thm}

\begin{proof}

For convenience, let us introduce the notation $f_i = f ( \frac{p_i}{q_i} )$, so that the formula for $I_f(p,q)$ can be written as
\[ I_f(p,q) = \sum_i q_i \f. \]
We begin by taking the time derivative of the $f$-divergence:

\[ \begin{array}{ccl} \displaystyle{ \frac{d}{dt}  I_f(p(t),q(t)) } &=& 
\displaystyle{  \sum_{i \in V} \frac{\partial I_f}{\partial p_i} \frac{dp_i}{dt} +  \frac{\partial I_f}{\partial q_i} \frac{dq_i}{dt} }
\end{array} \]

We can separate this into a sum over states $i \in V - B,$ for which the time derivatives of $p_i$ and $q_i$ are given by the master equation, and boundary states $i \in B,$ for which they are not:

\[ \begin{array}{ccl} \displaystyle{ \frac{d}{dt}  I_f(p(t),q(t)) } &=& 
\displaystyle{  \sum_{i \in V-B, \; j \in V} \frac{\partial I_f}{\partial p_i} H_{ij} p_j + 
                                             \frac{\partial I_f}{\partial q_i} H_{ij} q_j }\\  \\  
&& + \; \; \; \displaystyle{  \sum_{i \in B} \frac{\partial I_f}{\partial p_i} \frac{dp_i}{dt} +  \frac{\partial I_f}{\partial q_i} \frac{dq_i}{dt}}  
\end{array} \]
For boundary states we have
\[ \displaystyle{ \frac{dp_i}{dt} = \frac{Dp_i}{Dt} + \sum_{j \in V} H_{ij}p_j } \]
and similarly for the time derivative of $q_i.$    We thus obtain
\[
\begin{array}{ccl}
 \displaystyle{ \frac{d}{dt}  I_f(p(t),q(t)) } &=& 
\displaystyle{  \sum_{i,j \in V} \frac{\partial I_f}{\partial p_i} H_{ij} p_j + \frac{\partial I_f}{\partial q_i} H_{ij} q_j }\\  \\
&& + \; \; \displaystyle{  \sum_{i \in B} \frac{\partial I_f}{\partial p_i} \frac{Dp_i}{Dt} +  \frac{\partial I_f}{\partial q_i} \frac{Dq_i}{Dt}}  
\end{array} \]
To evaluate the first sum, recall that
\[ \displaystyle{   I_f(p,q) = \sum_{i \in V} q_i \f }  \]
so
\[ 
   \displaystyle{\frac{\partial I_f}{\partial p_i}} =\displaystyle{\fp } ,  \qquad
\displaystyle{ \frac{\partial I_f}{\partial q_i}}=  \displaystyle{\f - \frac{p_i}{q_i} \fp  }
\]
Thus, we have
\[    \displaystyle{ \sum_{i,j \in V}  \frac{\partial I_f}{\partial p_i} H_{ij} p_j + \frac{\partial I_f}{\partial q_i} H_{ij} q_j  =  
\sum_{i,j\in V} H_{ij}p_j \fp + H_{ij}q_j ( \f - \frac{p_i}{q_i} \fp) } \] 
\end{proof}

This result separates the change in relative entropy change into two parts: an 'internal' part and a 'boundary' part.  

It turns out the 'internal' part is always less than or equal to zero.  So, from Theorem \ref{thm:fdiv1} we can deduce a version of the Second Law of Thermodynamics for open Markov processes:

\begin{thm}\label{thm:fdiv2}  Given the conditions of Theorem \ref{thm:fdiv1}, we have

\[     \displaystyle{  \frac{d}{dt}  I_f(p(t),q(t)) \; \le \; 
\sum_{i \in B} \frac{\partial I_f}{\partial p_i} \frac{Dp_i}{Dt} +  \frac{\partial I_f}{\partial q_i} \frac{Dq_i}{Dt} 
}   \]
\end{thm}

\begin{proof}
Thanks to Theorem \ref{thm:fdiv1}, to prove
\[     \displaystyle{  \frac{d}{dt}  I(p(t),q(t)) \; \le \; 
\sum_{i \in B} \frac{\partial I}{\partial p_i} \frac{Dp_i}{Dt} +  \frac{\partial I}{\partial q_i} \frac{Dq_i}{Dt} 
}   \]
it suffices to show that
\[  \displaystyle{  \sum_{i,j\in V} H_{ij}p_j \fp + H_{ij}q_j ( \f - \frac{p_i}{q_i} \fp) }. \] 
Separating out the $i=j$ term we get
\[  \displaystyle{   \sum_{i, j \neq i \in V} H_{ij} p_j  \fp + H_{ij} q_j (\f - \frac{p_i}{q_i} \fp  ) + \sum_j H_{jj}p_j \fpj  + H_{jj} q_j (\fj - \frac{p_j}{q_j} \fpj) } \]
Since $H_{ij}$ is infinitesimal stochastic we have $H_{jj}  + \sum_{i \neq j } H_{ij} = 0 .$ Plugging this into the previous formula we have 
\[    \displaystyle{   \sum_{i,j \neq i  \in V} H_{ij} p_j ( \fp - \fpj) + H_{ij} q_j \left( \f - \fj - \frac{p_i}{q_i} \fp + \frac{p_j}{q_j} \fpj \right) } \]
which can be written
\[ \displaystyle{ \sum_{i,j \neq i \in V} H_{ij} q_j \left[ \f - \fj + \fp ( \frac{p_j}{q_j} - \frac{p_i}{q_i} ) \right]  }\].  
Note that $i=j$ contribution to the sum vanishes. The quantity in brackets is of the form $f(y) - f(x) - f'(y)(y-x)$ which is less than or equal to zero for $f$ convex. 
\end{proof}

Intuitively, this says that the $f$-divergence can only increase if it comes in from the boundary.

There is another nice result that holds when $q$ is an equilibrium solution of the master equation:

\begin{thm}\label{thm:fdiv3} Given the conditions of Theorem \ref{thm:fdiv1}, suppose also that $q$ is an equilibrium solution of the master equation.  Then we have
\[     \displaystyle{  \frac{d}{dt}  I_f(p(t),q) =  \sum_{i, j \in V} -J_{ij}(p)A_{ij}(p)  +  \sum_{i \in B} \frac{\partial I_f}{\partial p_i} \frac{Dp_i}{Dt} } \]
where
\[ J_{ij}(p) = H_{ij}p_j - H_{ji}p_i \]
is the \define{net flow} from $j$ to $i$, while  
\[A_{ij}(p) = f'( \frac{p_j}{q_j} ) - f' ( \frac{p_i}{q_i} ) \]
is the conjugate \define{generalized affinity}.
\end{thm}

\begin{proof}
Now suppose also that $q$ is an equilibrium solution of the master equation.  Then $Dq_i/Dt = dq_i/dt = 0$ for all states $i,$ so by Theorem 1 we have
\[  \frac{d}{dt} I_f(p(t),q) = \displaystyle{ \sum_{i, j \in V} H_{ij} p _j \fp + \sum_{i \in B} \frac{Dp_i}{Dt}\frac{\partial I_j}{\partial p_i} }. \]
To prove the theorem it suffices to show that
\[ \sum_{i,j \in V} H_{ij} p_j \fp = \frac{1}{2} \sum_{i,j \in V} J_{ij}(p)A_{ij}(p). \]
Starting with the left side of the above expression and separating out the $i=j$ term we have
\[ \sum_{i,j \in V} H_{ij} p_j \fp = \sum_{i, j\neq i \in V} H_{ij} p_j \fp + \sum_{i \in V} H_{ii}p_i \fp. \]
Since the index $i$ in the last term is summed over all of $V$, we can instead index the sum by $j$
\[ \sum_{i,j \neq i \in V} H_{ij} p_j \fp + \sum_{j \in V} H_{jj} p_j \fpj. \]
Using $H_{jj} + \sum_{ i \neq j} H_{ij} = 0$ from the infinitesimal stochastic property of $H$ we have,
\[ \sum_{i,j \neq i \in V} H_{ij} p_j ( \fp - \fpj ). \]
Note that the term in parenthesis vanishes when $i=j$ so we the result is unchanged if we sum over all $i$ and $j$. Finally we split our sum into two parts with the summation indices exchanged:
\[ \frac{1}{2} \left( \sum_{i,j \in V} H_{ij}p_j (\fp - \fpj) \sum_{j,i \in V} H_{ji}p_i (\fpj - \fp) \right). \] 
We can rewrite this as
\[ \sum_{i,j \in V} H_{ij} p_j \fp = \frac{1}{2} \sum_{i,j} (H_{ij}p_j - H_{ji}p_i) ( \fp - \fpj), \]
from which we see that
\[ \sum_{i,j \in V} H_{ij} p_j \fp = -\frac{1}{2} \sum_{i,j} J_{ij}(p)A_{ij}(p). \]
\end{proof}

\section{Dissipation as an $f$-divergence}
We now show that dissipation is in fact related to the rate of change of an $f$- divergence for a certain choice of convex $f$. In particular, taking $f(x) = -\frac{x}{2}^2$ we have that
\[ I_f(p,q) = \frac{1}{2}\sum_ i q_i \frac{p_i^2}{q_i^2} = \frac{1}{2} \sum_i \frac{p_i^2}{q_i}. \]
Let us denote the generalized entropy with this choice of $f$ by $ I_D(p,q)$. Then we have the following theorem:
\begin{thm}
Consider an open Markov processes with $V$ as its set of states and boundary $B$. Let $p(t)$ be some probability distribution obeying the open Master equation and  $q$ be a detailed balanced equilibrium of the underlying closed Markov process. Let $I_D(p(t),q) = \frac{1}{2} \sum_i \frac{p_i^2}{q_i}$ be a generalized entropy with $f(x) = \frac{x^2}{2}$. Then we have
\[ \frac{d}{dt} I_D(p(t),q) = -D(p) + \sum_{i \in B} \frac{Dp_i}{Dt} \frac{\partial I_D}{\partial p_i} \]
where $D(p)$ is the dissipation. 
\end{thm}

\begin{proof}
Differentiating $I_D(p(t),q)$ with respect to time we have
\[ \frac{d}{dt} I_D(p(t),q) = \sum_ {i \in V} \frac{p_i}{q_i} \frac{dp_i}{dt} \]
Using the master equation for internal states we have
\[ \frac{d}{dt} I_D(p(t),q) = \sum_{i \in V-B,j \in V} H_{ij} p_j \frac{p_i}{q_i} + \sum_{i \in B} \frac{p_i}{q_i} \frac{dp_i}{dt}.\]
For boundary states $\frac{dp_i}{dt} = \frac{Dp_i}{Dt} + \sum_j H_{ij}p_j$. Plugging this into the previous expression yields
\[ \frac{d}{dt} I_D(p(t),q) = \sum_{i, j \in V} H_{ij}p_j \frac{p_i}{q_i} + \sum_{i \in B} \frac{Dp_i}{Dt} \frac{p_i}{q_i}, \]
where the first sum is now over all states, not only internal states. Since $q_i \neq 0 $ for a detailed balanced equilibrium, we can write this slightly differently as
\[ \frac{d}{dt} I_D(p(t),q) = \sum_{i,j \in V} H_{ij}q_j \frac{p_i p_j}{q_i q_j} + \sum_{i \in B} \frac{D p_i}{Dt} \frac{p_i}{q_i}.\]
Therefore it remains to show that
\[ \sum_{i,j \in V} H_{ij}q_j \frac{p_i p_j}{q_i q_j} = -D(p) \]
where $D(p)$ is the dissipation.

Recall that the dissipation is given by
\[ D(p) = \frac{1}{2} \sum_{i,j} H_{ij} q_j ( \frac{p_j}{q_j} - \frac{p_i}{q_i} )^2, \]
where again, $q$ is a detailed balanced equilibrium. Expanding out the squared term we have
\[ D(p) = \frac{1}{2} \sum_{i,j} H_{ij}q_j (\frac{p_j^2}{q_j^2} + \frac{p_i^2}{q_i^2} - 2 \frac{p_i p_j}{q_i q_j} ). \]
The first term vanishes since $H$ is infinitesimal stochastic. The second term vanishes since $q$ is an equilibrium distribution. Thus we are left with
\[ D(p) = -\sum_{i,j} H_{ij} q_j \frac{p_i p_j}{q_i q_j}, \]
which completes the proof.
\end{proof}

\section{The master equation as a gradient flow}
We can understand dissipation in a more geometric way as the rate of change of an inner product on a space whose metric depends in a specific way on the components of the detailed balanced equilibrium distribution $q$. 

We saw that dissipation can be seen as a generalized entropy production of the generalized entropy
\[ I_D(p(t),q) = \sum_i \frac{p_i^2}{q_i} \]
where $q$ satisfies detailed balance.

Let us now introduce a metric $g_{ij} = \frac{ \delta_{i,j} }{q_i} $, where $\delta_{i,j}$ is the Kronecker delta. Note $g_{ij}$ is symmetric and positive definite for detailed balanced $q$. It's inverse is given by $g_{ij}^{-1} = q_i  \delta_{i,j}$.

The inner product on this space is given by 
\[ \begin{array}{ccc} \langle p, p' \rangle_q &=& \displaystyle{ \sum_{i,j} g_{ij} p_i p'_j } \\ 
&=& \displaystyle{ \sum_{i,j} \frac{\delta_{i,j}}{q_i} p_i p'_j }\\
&=& \displaystyle{ \sum_i \frac{p_i p'_i}{q_i} }. \end{array} \]
Taking $p'=p$ we have
\[ \langle p,p \rangle_q = \sum_i \frac{p_i^2}{q_i} = I_D(p(t),q) \]
We saw that for a closed Markov process, the dissipation is in fact the rate of change of $I_D(p(t),q)$
\[ \frac{d}{dt} I_D(p(t),q) = D(p). \]
Putting these expressions together we have that
\[ D(p) = \frac{d}{dt} \langle p,p \rangle_q. \]

We saw that one can write the dissipation as 
\[ D(p) = -\sum_{i,j} H_{ij}q_j \frac{p_i p_j}{q_i q_j}. \]
Differentiating with respect to some $p_n$
\[ \frac{\partial D(p)}{\partial p_n} = -\sum_i H_{in} \frac{p_i}{q_i} + \sum_j H_{nj}\frac{p_j}{q_n} \]
which using the detailed balanced property of $q$ can be written as
\[ \frac{\partial D(p)}{\partial p_n} = -2 \sum_j H_{nj} \frac{p_j}{q_n}.\]
Multiplying through by $-\frac{q_n}{2}$ yields
\[ -\frac{q_n}{2} \frac{ \partial D(p) }{ \partial p_n} = \sum_j H_{nj} p_j \]
where we recognize the right hand side as the master equation for $p_n$, so we have
\[ \frac{dp_n}{dt} = -\frac{q_n}{2} \frac{ \partial D(p) }{ \partial p_n} \]

If we define a gradient on this space is given by
\[ \nabla_i = \sum_j \frac{g_{ij}^{-1} }{2} \frac{\partial}{\partial p_j} \]
or simply
\[ \nabla_i = \frac{q_i}{2} \frac{\partial}{\partial p_i}, \]
then we can write the master equation as a \define{gradient flow}
\[ \frac{dp}{dt} = -\nabla D(p). \]

\chapter{Categorical modeling of open systems}\label{ch:catopen}

The basic idea is to consider open systems, open in the sense that they interact with their environment or with other systems, as morphisms in symmetric monoidal categories. Composition and tensoring in these categories provide methods for building up larger open systems from smaller ones. One can then study the behaviors of such systems in a compositional way using functors, allowing one to compute the behavior of a composite system as the composite of the behaviors of its constituent systems. For an introduction to category theory see for instance \cite{MacLane}.

The focus in this thesis is on two types of interacting systems which admit graphical syntax: probabilistic systems and reaction networks. In the first vein we construct a category whose morphisms are open Markov processes. This is accomplished by splitting up the boundary of an open process into a a set of inputs and a set of outputs. If the outputs of one open Markov process match the inputs of another open Markov process, the two can be composed by gluing the processes together along their common overlap, resulting in a new open Markov process.

In this chapter we lay the necessary category-theoretic groundwork for constructing categories of open systems and explain the idea of studying the behaviors of such systems in a compositional way using functors. Much of this work utilizes the decorated cospan approach due to Fong \cite{Fong} and we recall some relevant results from his work. The decorated cospan approach and much more are explained in detail in Fong's thesis \cite{FongThesis}. The previous work of Katis, Sabadini, and Walters also uses cospans to model various networked systems such as electrical circuits, automata, Petri nets, and Markov processes \cite{PSWalters}. The decorated cospan approach differs in that rather than working with cospans in a particular network cateogry, we use cospans to keep track of the interconnection of systems while the decoration carries the data of the system itself. 

\begin{defn}\label{defn:category}
A category $\CC$ consists of
\begin{itemize}
\item a class of \define{objects} $X,Y \in \CC$
\item a class of \define{morphisms} $ f \maps X \to Y \in \CC$ between objects. Each morphism $f \maps X \to Y $ has a \define{source} object $X$ and a \define{target} object $Y$. For each pair of objects $X,Y$ we have a \define{hom-set} $\hom_{\CC}(X,Y)$ of morphisms in $\CC$ with source $X$ and target $Y$.
\item For all triples of objects $X,Y,$ and $Z$, a \define{composition} operation
\[ -;- \maps \hom_{\CC}(X,Y) \times \hom_{\CC}(Y,Z) \to \hom_{\CC}(X,Z) \] 
which, sends the pair morphisms $f \maps X \to Y$ and $g \maps Y \to Z$ to their composite
$f;g \maps X \to Z$. Note that usually composition is written using a circle $g \circ f$ and the opposite order. Composition is associative meaning that for all quadruples of objects $X,Y,Z,$ and $A$ the following diagram commutes  
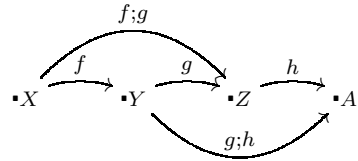
\[ \xymatrix{ \centerdot_X \ar@/^.4pc/[r]^f \ar@/^2pc/[rr]^{f;g}  & \centerdot_Y \ar@/^.4pc/[r]^g \ar@/_2pc/[rr]^{g;h} & \centerdot_Z \ar@/^.4pc/[r]^h & \centerdot_A} \]

\item For all objects $X$, \define{identity morphisms} $ 1_X \maps X \to X$ satisfying the \define{left/right identity laws}, which say that the following diagram commutes for all objects $X$
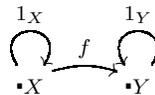
\[  \xymatrix{ \centerdot_X \ar@(ul,ur)^{1_X} \ar@/^/[r]^f & \centerdot_Y \ar@(ur,ul)_{1_Y} }  \]
\end{itemize}

\end{defn}

We can think of open systems as morphisms in a category. 
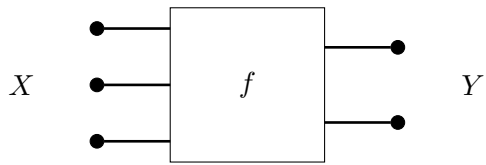
\[
\begin{tikzpicture}
	\begin{pgfonlayer}{nodelayer}
		\node [style=none] (0) at (-1, 0.75) {};
		\node [style=none] (1) at (1, -0.5) {};
		\node [style=none] (2) at (-1, -0) {};
		\node [style=circle, fill=black,draw=black,scale=0.5] (3) at (2, 0.5) {};
		\node [style=rectangle,draw=black, scale=8.0] (4) at (0, -0) {};
		\node [style=none] (5) at (-3, -0) {$X$};
		\node [style=none] (6) at (3, -0) {$Y$};
		\node [style=none] (7) at (1, 0.5) {};
		\node [style=circle, fill=black,draw=black,scale=0.5] (8) at (2, -0.5) {};
		\node [style=circle, fill=black,draw=black,scale=0.5] (9) at (-2, 0.75) {};
		\node [style=none] (10) at (0, -0) {$f$};
		\node [style=circle, fill=black,draw=black,scale=0.5] (11) at (-2, -0) {};
		\node [style=none] (12) at (-1, -0.75) {};
		\node [style=circle, fill=black,draw=black,scale=0.5] (13) at (-2, -0.75) {};
	\end{pgfonlayer}
	\begin{pgfonlayer}{edgelayer}
		\draw [style=simple] (9) to (0);
		\draw [style=simple] (11) to (2);
		\draw [style=simple] (13) to (12);
		\draw [style=simple] (8) to (1);
		\draw [style=simple] (3) to (7);
	\end{pgfonlayer}
\end{tikzpicture}
\]

Here we are imagining that inside the big box is some system with inputs $X$ and outputs $Y$. Given two open systems $f \maps X \to Y$ and $g \maps Y \to Z$ where the outputs of $f$ match the inputs of $g$:
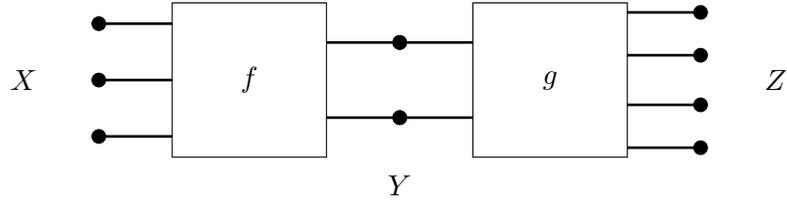
\[
\begin{tikzpicture}
	\begin{pgfonlayer}{nodelayer}
		\node [style=none] (0) at (-3, 0.75) {};
		\node [style=none] (1) at (-1, -0.5) {};
		\node [style=none] (2) at (-3, -0) {};
		\node [style=circle, fill=black,draw=black,scale=0.5] (3) at (0, 0.5) {};
		\node [style=rectangle,draw=black, scale=8.0] (4) at (-2, -0) {};
		\node [style=none] (5) at (-5, -0) {$X$};
		\node [style=none] (6) at (1, -0) {};
		\node [style=none] (7) at (-1, 0.5) {};
		\node [style=circle, fill=black,draw=black,scale=0.5] (8) at (0, -0.5) {};
		\node [style=circle, fill=black,draw=black,scale=0.5] (9) at (-4, 0.75) {};
		\node [style=none] (10) at (-2, -0) {$f$};
		\node [style=circle, fill=black,draw=black,scale=0.5] (11) at (-4, -0) {};
		\node [style=none] (12) at (-3, -0.75) {};
		\node [style=circle, fill=black,draw=black,scale=0.5] (13) at (-4, -0.75) {};
	\end{pgfonlayer}
	\begin{pgfonlayer}{edgelayer}
		\draw [style=simple] (9) to (0);
		\draw [style=simple] (11) to (2);
		\draw [style=simple] (13) to (12);
		\draw [style=simple] (8) to (1);
		\draw [style=simple] (3) to (7);
	\end{pgfonlayer}

	\begin{pgfonlayer}{nodelayer}
	
		\node [style=none] (0) at (1, 0.5) {};
		\node [style=none] (1) at (3, .9) {};
		\node [style=none] (2) at (1, -0.5) {};
		
		\node [style=none] (1') at (3, .33) {};
		\node [style=none] (1'') at (3, -0.33) {};
		\node [style=circle, fill=black,draw=black,scale=0.5] (3) at (4, .9) {};
		\node [style=rectangle,draw=black, scale=8.0] (4) at (2, -0) {};
		\node [style=none] (5) at (-0, -1.4) {$Y$};
		\node [style=none] (6) at (5, -0) {$Z$};
		\node [style=none] (7) at (3, -.9) {};
		\node [style=circle, fill=black,draw=black,scale=0.5] (8) at (4, -.9) {};
		\node [style=circle, fill=black,draw=black,scale=0.5] (14) at (4, .33) {};
		\node [style=circle, fill=black,draw=black,scale=0.5] (13) at (4, -.33) {};
		\node [style=circle, fill=black,draw=black,scale=0.5] (9) at (0, 0.5) {};
		\node [style=none] (10) at (2, -0) {$g$};
		\node [style=circle, fill=black,draw=black,scale=0.5] (11) at (0, -0.5) {};
		\node [style=none] (12) at (-1, -0.75) {};
	\end{pgfonlayer}
	\begin{pgfonlayer}{edgelayer}
		\draw [style=simple] (9) to (0);
		\draw [style=simple] (11) to (2);
		\draw [style=simple] (13) to (1'');
		\draw [style=simple] (8) to (7);
		\draw [style=simple] (3) to (1);
		\draw [style=simple] (14) to (1');

	\end{pgfonlayer}
\end{tikzpicture}
\]
we can form the composite open system $f;g \maps X \to Z$
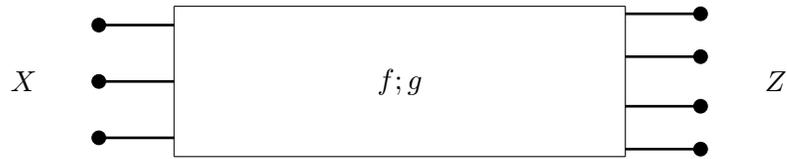
\[
\begin{tikzpicture}
	\begin{pgfonlayer}{nodelayer}
		\coordinate (0) at (-3, 0.75) {};
		\coordinate (2) at (-3, -0) {};
		\node [style=none] (5) at (-5, -0) {$X$};
		\node [style=none] (6) at (0, 0) {$f;g$};
		\node [style=circle, fill=black,draw=black,scale=0.5] (9) at (-4, 0.75) {};
		\node [style=circle, fill=black,draw=black,scale=0.5] (11) at (-4, -0) {};
		\coordinate (12) at (-3, -0.75) {};
		\node [style=circle, fill=black,draw=black,scale=0.5] (13) at (-4, -0.75) {};
		
		\coordinate (a) at (3,1) {};
		\coordinate(b) at (3,-1) {};
		\coordinate(c) at (-3,-1){};
		\coordinate(d) at (-3,1){};
		
	\end{pgfonlayer}
	\begin{pgfonlayer}{edgelayer}
		\draw [style=simple] (9) to (0);
		\draw [style=simple] (11) to (2);
		\draw [style=simple] (13) to (12);

		\draw[-] (a) to (b);
		\draw[-] (b) to (c);
		\draw[-] (c) to (d);
		\draw[-] (d) to (a);
	\end{pgfonlayer}

	\begin{pgfonlayer}{nodelayer}
	
		\coordinate (1) at (3, .9) {};
		
		\coordinate (1') at (3, .33) {};
		\coordinate (1'') at (3, -0.33) {};
		\node [style=circle, fill=black,draw=black,scale=0.5] (3) at (4, .9) {};
		\node [style=none] (5) at (-1, -0) {};
		\node [style=none] (6) at (5, -0) {$Z$};
		\coordinate (7) at (3, -.9) {};
		\node [style=circle, fill=black,draw=black,scale=0.5] (8) at (4, -.9) {};
		\node [style=circle, fill=black,draw=black,scale=0.5] (14) at (4, .33) {};
		\node [style=circle, fill=black,draw=black,scale=0.5] (13) at (4, -.33) {};
		\node [style=none] (12) at (-1, -0.75) {};
	\end{pgfonlayer}
	\begin{pgfonlayer}{edgelayer}
		\draw [style=simple] (13) to (1'');
		\draw [style=simple] (8) to (7);
		\draw [style=simple] (3) to (1);
		\draw [style=simple] (14) to (1');

	\end{pgfonlayer}
\end{tikzpicture}
\]
This requires some method of taking two systems $f \maps X \to Y$ and $g \maps Y \to Z$ and forming a composite system $f;g \maps X \to Z$. We study systems which admit a graphical syntax and composition corresponds to the gluing of systems together along a common boundary. This will be made precise later as we study open Markov processes and open reaction networks from a categorical perspective. 

\section{Compositional behaviors of open systems via functors}
Representing open systems as morphisms not only allows us to build up larger open systems via the composition of open systems, it also enables the study of such systems in a compositional way using `functors.' 
\begin{defn}\label{defn:functor}
Given categories $\CC$ and $\D$, a \define{functor} $F \maps \CC \to \D$ is a mapping between categories which preserves identities and respects composition,
\[ F(f ; g) = F(f)F(g) \]
\[ F(1_X) = 1_{F(X)} \]
for all composable $f,g \in \CC$ and all objects $X \in \CC$.
\end{defn}

The basic idea behind `functorial semantics' is to use functors to study the `behaviors' of open systems
\[ F \maps \OpenSys \to \Behave \]
where here $\OpenSys$ is some category whose objects are finite sets corresponding to inputs or outputs and whose morphisms are open systems and $\Behave$ is some category capable of capturing the behaviors of the open systems. We will be using categories where the morphisms are relations, either linear for Markov processes or more general for reaction networks, for the category $\Behave$. There will be some notion of composition of open systems in $\OpenSys$ as well as some notion of composition of behaviors in $\Behave$. The fact that a functor $F \maps \OpenSys \to \Behave$ preserves composition 
\[ F \left( \ \begin{tikzpicture}[baseline=-6, scale=0.5]
	\begin{pgfonlayer}{nodelayer}
		\node [style=none] (0) at (-1, 0.75) {};
		\node [style=none] (1) at (1, -0.5) {};
		\node [style=none] (2) at (-1, -0) {};
		\node [style=circle, fill=black,draw=black,scale=0.5] (3) at (2, 0.5) {};
		\node [style=rectangle,draw=black, scale=4.0] (4) at (0, -0) {};

		\node [style=none] (7) at (1, 0.5) {};
		\node [style=circle, fill=black,draw=black,scale=0.5] (8) at (2, -0.5) {};
		\node [style=circle, fill=black,draw=black,scale=0.5] (9) at (-2, 0.75) {};
		\node [style=none] (10) at (0, -0) {$f$};
		\node [style=circle, fill=black,draw=black,scale=0.5] (11) at (-2, -0) {};
		\node [style=none] (12) at (-1, -0.75) {};
		\node [style=circle, fill=black,draw=black,scale=0.5] (13) at (-2, -0.75) {};
	\end{pgfonlayer}
	\begin{pgfonlayer}{edgelayer}
		\draw [style=simple] (9) to (0);
		\draw [style=simple] (11) to (2);
		\draw [style=simple] (13) to (12);
		\draw [style=simple] (8) to (1);
		\draw [style=simple] (3) to (7);
	\end{pgfonlayer}
\end{tikzpicture}
\right) \ ; \ F \left( \ \begin{tikzpicture}[baseline=-6, scale=0.5]
		\begin{pgfonlayer}{nodelayer}
		\node [style=none] (0) at (1, 0.5) {};
		\node [style=none] (1) at (3, .9) {};
		\node [style=none] (2) at (1, -0.5) {};
		
		\node [style=none] (1') at (3, .33) {};
		\node [style=none] (1'') at (3, -0.33) {};
		\node [style=circle, fill=black,draw=black,scale=0.5] (3) at (4, .9) {};
		\node [style=rectangle,draw=black, scale=4.0] (4) at (2, -0) {};
		
		\node [style=none] (7) at (3, -.9) {};
		\node [style=circle, fill=black,draw=black,scale=0.5] (8) at (4, -.9) {};
		\node [style=circle, fill=black,draw=black,scale=0.5] (14) at (4, .33) {};
		\node [style=circle, fill=black,draw=black,scale=0.5] (13) at (4, -.33) {};
		\node [style=circle, fill=black,draw=black,scale=0.5] (9) at (0, 0.5) {};
		\node [style=none] (10) at (2, -0) {$g$};
		\node [style=circle, fill=black,draw=black,scale=0.5] (11) at (0, -0.5) {};

	\end{pgfonlayer}
	\begin{pgfonlayer}{edgelayer}
		\draw [style=simple] (9) to (0);
		\draw [style=simple] (11) to (2);
		\draw [style=simple] (13) to (1'');
		\draw [style=simple] (8) to (7);
		\draw [style=simple] (3) to (1);
		\draw [style=simple] (14) to (1');

	\end{pgfonlayer}
\end{tikzpicture}
\ \right) \ = \ F \left( \ 
\begin{tikzpicture}[baseline=-5, scale=0.5]
	\begin{pgfonlayer}{nodelayer}
		\coordinate (0) at (-3, 0.75) {};
		\coordinate (2) at (-3, -0) {};

		\node [style=none] (6) at (0, 0) {$f;g$};
		\node [style=circle, fill=black,draw=black,scale=0.5] (9) at (-4, 0.75) {};
		\node [style=circle, fill=black,draw=black,scale=0.5] (11) at (-4, -0) {};
		\coordinate (12) at (-3, -0.75) {};
		\node [style=circle, fill=black,draw=black,scale=0.5] (13) at (-4, -0.75) {};
		
		\coordinate (a) at (3,1) {};
		\coordinate(b) at (3,-1) {};
		\coordinate(c) at (-3,-1){};
		\coordinate(d) at (-3,1){};
		
	\end{pgfonlayer}
	\begin{pgfonlayer}{edgelayer}
		\draw [style=simple] (9) to (0);
		\draw [style=simple] (11) to (2);
		\draw [style=simple] (13) to (12);

		\draw[-] (a) to (b);
		\draw[-] (b) to (c);
		\draw[-] (c) to (d);
		\draw[-] (d) to (a);
	\end{pgfonlayer}

	\begin{pgfonlayer}{nodelayer}
	
		\coordinate (1) at (3, .9) {};
		
		\coordinate (1') at (3, .33) {};
		\coordinate (1'') at (3, -0.33) {};
		\node [style=circle, fill=black,draw=black,scale=0.5] (3) at (4, .9) {};
		\node [style=none] (5) at (-1, -0) {};

		\coordinate (7) at (3, -.9) {};
		\node [style=circle, fill=black,draw=black,scale=0.5] (8) at (4, -.9) {};
		\node [style=circle, fill=black,draw=black,scale=0.5] (14) at (4, .33) {};
		\node [style=circle, fill=black,draw=black,scale=0.5] (13) at (4, -.33) {};
		\node [style=none] (12) at (-1, -0.75) {};
	\end{pgfonlayer}
	\begin{pgfonlayer}{edgelayer}
		\draw [style=simple] (13) to (1'');
		\draw [style=simple] (8) to (7);
		\draw [style=simple] (3) to (1);
		\draw [style=simple] (14) to (1');

	\end{pgfonlayer}
\end{tikzpicture} \ \right)
\]
means that the behavior of a composite open system can be computed as the composite of the behaviors of the constituent open systems.

Just as a functor is a certain type of `nice' mapping between categories, there exist certain classes of `nice' maps between functors. These are called natural transformations. 

\begin{defn} \label{defn:natural}
Given functors $F,G \maps \CC \to \D$,  a \define{natural transformation} $\alpha \maps F \Rightarrow G$ consists of morphisms $\alpha_X \maps F(X) \to G(X)$ for all objects $ X \in \CC$ and $\alpha_f \maps F(f) \to G(f)$ for all morphisms $f \maps X \to Y$ such that the following diagram commutes 
\[ \xymatrix{  F(X) \ar[r]^{F(f)} \ar[d]_{\alpha_X} &  F(Y) \ar[d]^{\alpha_Y} \\ G(X) \ar[r]_{G(f)} & G(Y) } \]
for all pairs of objects and all morphisms between them in $\CC$. If the morphisms $\alpha_X \maps F(X) \to G(X)$ and $\alpha_f \maps F(f) \to G(f)$ are all isomorphisms, then $\alpha \maps F \Rightarrow G$ is a \define{natural isomorphism}. 
\end{defn}

In the context of the functorial semantics of open systems, natural transformations between different behavioral functors give methods of interpolating between the corresponding notions of behavior. Natural transformations play a number of other important roles, for instance in the definition of a monoidal category to which we now turn our attention.

\section{Monoidal categories}
In fact almost all the categories we study in this thesis are `symmetric monoidal categories' with symmetric monoidal functors between them. Given two open systems $f \maps X \to Y$:
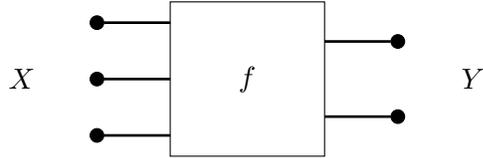
\[
\begin{tikzpicture}
	\begin{pgfonlayer}{nodelayer}
		\node [style=none] (0) at (-1, 0.75) {};
		\node [style=none] (1) at (1, -0.5) {};
		\node [style=none] (2) at (-1, -0) {};
		\node [style=circle, fill=black,draw=black,scale=0.5] (3) at (2, 0.5) {};
		\node [style=rectangle,draw=black, scale=8.0] (4) at (0, -0) {};
		\node [style=none] (5) at (-3, -0) {$X$};
		\node [style=none] (6) at (3, -0) {$Y$};
		\node [style=none] (7) at (1, 0.5) {};
		\node [style=circle, fill=black,draw=black,scale=0.5] (8) at (2, -0.5) {};
		\node [style=circle, fill=black,draw=black,scale=0.5] (9) at (-2, 0.75) {};
		\node [style=none] (10) at (0, -0) {$f$};
		\node [style=circle, fill=black,draw=black,scale=0.5] (11) at (-2, -0) {};
		\node [style=none] (12) at (-1, -0.75) {};
		\node [style=circle, fill=black,draw=black,scale=0.5] (13) at (-2, -0.75) {};
	\end{pgfonlayer}
	\begin{pgfonlayer}{edgelayer}
		\draw [style=simple] (9) to (0);
		\draw [style=simple] (11) to (2);
		\draw [style=simple] (13) to (12);
		\draw [style=simple] (8) to (1);
		\draw [style=simple] (3) to (7);
	\end{pgfonlayer}
\end{tikzpicture}
\]
and $f' \maps X' \to Y'$:
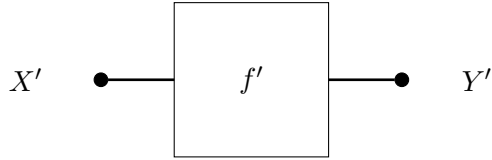
\[
\begin{tikzpicture}
	\begin{pgfonlayer}{nodelayer}
		\node [style=none] (0) at (-1, 0.75) {};
		\node [style=none] (1) at (1, 0) {};
		\node [style=none] (2) at (-1, -0) {};

		\node [style=rectangle,draw=black, scale=8.0] (4) at (0, -0) {};
		\node [style=none] (5) at (-3, -0) {$X'$};
		\node [style=none] (6) at (3, -0) {$Y'$};
		\node [style=none] (7) at (1, 0.5) {};
		\node [style=circle, fill=black,draw=black,scale=0.5] (8) at (2, 0) {};
		\node [style=none] (10) at (0, -0) {$f'$};
		\node [style=circle, fill=black,draw=black,scale=0.5] (11) at (-2, -0) {};
		\node [style=none] (12) at (-1, -0.75) {};
	\end{pgfonlayer}
	\begin{pgfonlayer}{edgelayer}
		\draw [style=simple] (11) to (2);
		\draw [style=simple] (8) to (1);
	\end{pgfonlayer}
\end{tikzpicture}
\] 
the tensor product of the two systems $f \otimes f' \maps X \otimes X' \to Y \otimes Y'$ corresponds to considering the two systems  placed `side by side' as a single system:
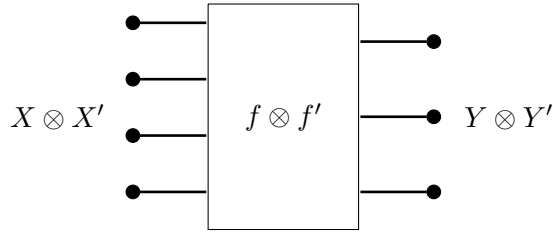
\[
\begin{tikzpicture}
	\begin{pgfonlayer}{nodelayer}
		\node [style=none] (0) at (-1, 0.75) {};
		\node [style=none] (1) at (1, -0.5) {};
		\node [style=none] (2) at (-1, -0) {};
		\node[style=none] (2') at (-1,-1.5) {};
		\node[style=none] (3') at (1, -1.5) {};
		\node [style=circle, fill=black,draw=black,scale=0.5] (3) at (2, 0.5) {};
		\node [style=none] (5) at (-3, -0.5) {$X \otimes X'$};
		\node [style=none] (6) at (3, -0.5) {$Y \otimes Y'$};
		\node [style=none] (7) at (1, 0.5) {};
		\node [style=circle, fill=black,draw=black,scale=0.5] (8) at (2, -0.5) {};
		\node [style=circle, fill=black,draw=black,scale=0.5] (9) at (-2, 0.75) {};
		\node [style=circle, fill=black,draw=black,scale=0.5] (9') at (-2, -1.5) {};
		\node [style=circle, fill=black,draw=black,scale=0.5] (10') at (2, -1.5) {};
		\node [style=none] (10) at (0, -0.5) {$f \otimes f'$};
		\node [style=circle, fill=black,draw=black,scale=0.5] (11) at (-2, -0) {};
		\node [style=none] (12) at (-1, -0.75) {};
		\node [style=circle, fill=black,draw=black,scale=0.5] (13) at (-2, -0.75) {};
	\end{pgfonlayer}
	\begin{pgfonlayer}{edgelayer}
		\draw [style=simple] (9) to (0);
		\draw [style=simple] (9') to (2');
		\draw [style=simple] (11) to (2);
		\draw [style=simple] (13) to (12);
		\draw [style=simple] (8) to (1);
		\draw [style=simple] (3) to (7);
		\draw[style=simple] (3') to (10');
		\draw (-1,1) -- (1,1) --(1,-2) -- (-1,-2)--(-1,1);
	\end{pgfonlayer}
\end{tikzpicture}
\]

\begin{defn} \label{defn:mc}
A \define{monoidal category} is a category $\CC$ together with 
\begin{itemize} \item a \define{tensor product} $ \otimes \maps \CC \times \CC \to \CC $ which sends a pair of objects $X,Y \in \CC$ to the object $X \otimes Y \in \CC$ and sends a pair of morphisms $f \maps X \to Y$ and $f' \maps X' \to Y'$ in $\CC$ to a morphism $f \otimes f' \maps X \otimes X' \to Y \otimes Y' \in \CC$.
\item a \define{unit object} $I \in \CC$ for the tensor product which together with 
\item natural isomorphisms $\alpha_{X,Y,Z} \maps (X \otimes Y) \otimes Z \to X \otimes (Y \otimes Z)$, $\lambda_X \maps I \otimes X \to X$ and $\rho_X \maps X \otimes I \to X$, called the \define{associator} and the \define{left and right unitors} respectively, making the following triangle and pentagon diagrams commute
\end{itemize}
\[ \xymatrix{ ( X \otimes I ) \otimes Y \ar[rr]^{\alpha_{X,I,Y} } \ar[rd]_{\rho_X \otimes 1_Y} & & X \otimes (I \otimes Y) \ar[ld]^{ 1_X \otimes \lambda_Y} \\ & X \otimes Y & } \]
The pentagon identity:
\[ \xymatrix{ & (( W \otimes X ) \otimes Y ) \otimes Z  \ar[rd]^{\alpha_{W,X,Y} \otimes 1_Z} \ar[ld]_{\alpha_{W\otimes X,Y,Z}} & \\ ( W \otimes X) \otimes (Y \otimes Z) \ar[d]_{\alpha_{W,X,Y\otimes Z} } &  & (W \otimes (X \otimes Y) ) \otimes Z \ar[d]^{\alpha_{W,X \otimes Y,Z}} \\ W \otimes (X \otimes (Y \otimes Z)) \ar[rr]_{1_W \otimes \alpha_{X,Y,Z}} & & W \otimes ((X \otimes Y) \otimes Z ) } \]
\end{defn}

There is an operation called the `braiding' which switches the order of a tensor product. For instance, one can `swap' the order of the inputs of $f \otimes f' \maps X \otimes X' \to Y \otimes Y'$ by precomposing with the braiding morphism $\beta_{X',X} \maps X' \otimes X \to X \otimes X' $ giving 
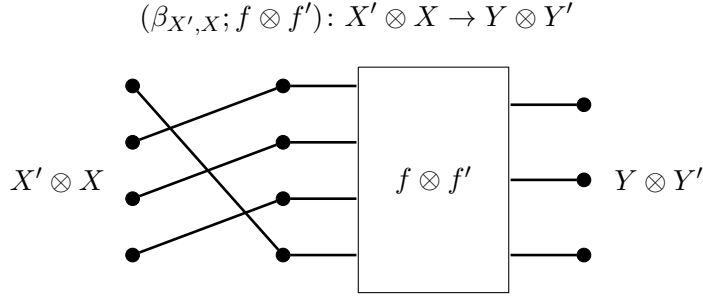
\[ ( \beta_{X',X} ; f \otimes f' ) \maps X' \otimes X \to Y \otimes Y' \]
\[
\begin{tikzpicture}
	\begin{pgfonlayer}{nodelayer}
	
		\node [style=none] (0) at (-1, 0.75) {};
		\node [style=none] (1) at (1, -0.5) {};
		\node [style=none] (2) at (-1, -0) {};
		\node[style=none] (2') at (-1,-1.5) {};
		\node[style=none] (3') at (1, -1.5) {};
		\node [style=circle, fill=black,draw=black,scale=0.5] (3) at (2, 0.5) {};
		\node [style=none] (5) at (-5, -0.5) {$X' \otimes X$};
		\node [style=none] (6) at (3, -0.5) {$Y \otimes Y'$};
		\node [style=none] (7) at (1, 0.5) {};
		\node [style=circle, fill=black,draw=black,scale=0.5] (8) at (2, -0.5) {};
		\node [style=circle, fill=black,draw=black,scale=0.5] (9) at (-2, 0.75) {};
		\node [style=circle, fill=black,draw=black,scale=0.5] (9') at (-2, -1.5) {};
		\node [style=circle, fill=black,draw=black,scale=0.5] (10') at (2, -1.5) {};
		\node [style=none] (10) at (0, -0.5) {$f \otimes f'$};
		\node [style=circle, fill=black,draw=black,scale=0.5] (11) at (-2, -0) {};
		\node [style=none] (12) at (-1, -0.75) {};
		\node [style=circle, fill=black,draw=black,scale=0.5] (13) at (-2, -0.75) {};
	
		\node [style=circle, fill=black,draw=black,scale=0.5] (9'') at (-4, 0.75) {};
		\node [style=circle, fill=black,draw=black,scale=0.5] (10'') at (-4, -1.5) {};
	\node [style=circle, fill=black,draw=black,scale=0.5] (13'') at (-4, -0.75) {};	
		\node [style=circle, fill=black,draw=black,scale=0.5] (11'') at (-4, -0) {};

	\end{pgfonlayer}
	\begin{pgfonlayer}{edgelayer}
	
		\draw [style=simple] (9'') to (9');
		\draw [style=simple] (10'') to (13);
		\draw [style=simple] (11'') to (9);
		\draw [style=simple] (13'') to (11);

		\draw [style=simple] (9) to (0);
		\draw [style=simple] (9') to (2');
		\draw [style=simple] (11) to (2);
		\draw [style=simple] (13) to (12);
		\draw [style=simple] (8) to (1);
		\draw [style=simple] (3) to (7);
		\draw[style=simple] (3') to (10');
		\draw (-1,1) -- (1,1) --(1,-2) -- (-1,-2)--(-1,1);
	\end{pgfonlayer}
\end{tikzpicture}
\]

\begin{defn} \label{defn:bmc}
A \define{braided} monoidal category is a monoidal category equipped with a \define{braiding natural isomorphism}
\[ \beta_{X,Y} \maps X \otimes Y \to Y \otimes X \]
which is compatible with the associator so as to make the following hexagons commute for all objects
\[ \xymatrix{ (X \otimes Y ) \otimes Z \ar[rr]^{\alpha_{X,Y,Z}}  \ar[d]_{ \beta_{X,Y} \otimes 1_Z} & & X \otimes ( Y \otimes Z) \ar[rr]^{\beta_{X,Y \otimes Z} }& & (Y \otimes Z) \otimes X \ar[d]^{\alpha_{Y,Z,X} } \\  ( Y \otimes X) \otimes Z \ar[rr]_{\alpha_{X,Y,Z} } & & Y \otimes (X \otimes Z) \ar[rr]_{1_Y \otimes \beta_{X,Z} } & & Y \otimes (Z \otimes X) } \]  
and 
\[ \xymatrix{ X \otimes (Y \otimes Z) \ar[rr]^{\alpha^{-1}_{X,Y,Z} } \ar[d]_{1_X \otimes \beta_{Y,Z} } && ( X \otimes Y) \otimes Z \ar[rr]^{\beta_{X \otimes Y, Z} }& & Z \otimes ( X \otimes Y) \ar[d]^{\alpha^{-1}_{Z,X,Y} } \\ X \otimes (Z \otimes Y) \ar[rr]_{\alpha^{-1}_{X,Z,Y} } & & (X \otimes Z ) \otimes Y \ar[rr]_{\beta_{X,Z} \otimes 1_Y} & & (Z \otimes X) \otimes Y } \]
\end{defn}

\begin{defn} \label{defn:smc}
A \define{symmetric} monoidal category is a braided monoidal category for which the braiding natural isomorphism is its own inverse
\[ \beta_{X,Y} = \beta_{Y,X}^{-1} \]
for all objects $X,Y \in Ob(\CC)$. 
\end{defn}

The braiding natural isomorphism $\beta_{X,X'} \maps X \otimes X' \to X' \otimes X$ corresponds to switching the order of the inputs $X$ and $X'$. For symmetric monoidal categories, the fact that this braiding natural isomorphism is its own inverse means that switching the order of the inputs $X$ and $X'$ and then switching $X'$ with $X$ is the same as doing nothing to the original original system, i.e. acting with the identity $1_{X \otimes X'}$. In terms of pictures, this means that the morphism
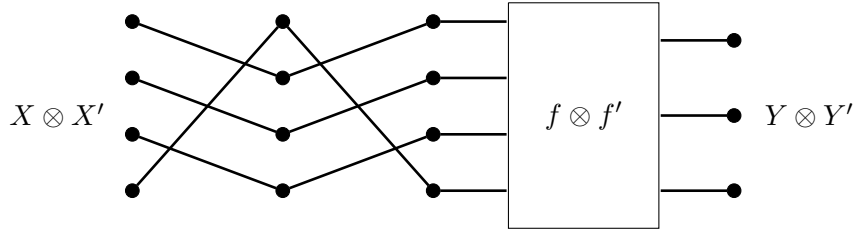
\[
\begin{tikzpicture}
	\begin{pgfonlayer}{nodelayer}
	
		\node [style=none] (0) at (-1, 0.75) {};
		\node [style=none] (1) at (1, -0.5) {};
		\node [style=none] (2) at (-1, -0) {};
		\node[style=none] (2') at (-1,-1.5) {};
		\node[style=none] (3') at (1, -1.5) {};
		\node [style=circle, fill=black,draw=black,scale=0.5] (3) at (2, 0.5) {};
		\node [style=none] (5) at (-7, -0.5) {$X \otimes X'$};
		\node [style=none] (6) at (3, -0.5) {$Y \otimes Y'$};
		\node [style=none] (7) at (1, 0.5) {};
		\node [style=circle, fill=black,draw=black,scale=0.5] (8) at (2, -0.5) {};
		\node [style=circle, fill=black,draw=black,scale=0.5] (9) at (-2, 0.75) {};
		\node [style=circle, fill=black,draw=black,scale=0.5] (9') at (-2, -1.5) {};
		\node [style=circle, fill=black,draw=black,scale=0.5] (10') at (2, -1.5) {};
		\node [style=none] (10) at (0, -0.5) {$f \otimes f'$};
		\node [style=circle, fill=black,draw=black,scale=0.5] (11) at (-2, -0) {};
		\node [style=none] (12) at (-1, -0.75) {};
		\node [style=circle, fill=black,draw=black,scale=0.5] (13) at (-2, -0.75) {};
	
		\node [style=circle, fill=black,draw=black,scale=0.5] (9'') at (-4, 0.75) {};
		\node [style=circle, fill=black,draw=black,scale=0.5] (10'') at (-4, -1.5) {};
	\node [style=circle, fill=black,draw=black,scale=0.5] (13'') at (-4, -0.75) {};	
		\node [style=circle, fill=black,draw=black,scale=0.5] (11'') at (-4, -0) {};

		\node [style=circle, fill=black,draw=black,scale=0.5] (9''') at (-6, 0.75) {};
		\node [style=circle, fill=black,draw=black,scale=0.5] (10''') at (-6, -1.5) {};
	\node [style=circle, fill=black,draw=black,scale=0.5] (13''') at (-6, -0.75) {};	
		\node [style=circle, fill=black,draw=black,scale=0.5] (11''') at (-6, -0) {};

	\end{pgfonlayer}
	\begin{pgfonlayer}{edgelayer}
	
		\draw [style=simple] (9'') to (9');
		\draw [style=simple] (10'') to (13);
		\draw [style=simple] (11'') to (9);
		\draw [style=simple] (13'') to (11);

		\draw [style=simple] (11'') to (9''');
		\draw [style=simple] (9'') to (10''');
		\draw [style=simple] (10'') to (13''');
		\draw [style=simple] (13'') to (11''');

		\draw [style=simple] (9) to (0);
		\draw [style=simple] (9') to (2');
		\draw [style=simple] (11) to (2);
		\draw [style=simple] (13) to (12);
		\draw [style=simple] (8) to (1);
		\draw [style=simple] (3) to (7);
		\draw[style=simple] (3') to (10');
		\draw (-1,1) -- (1,1) --(1,-2) -- (-1,-2)--(-1,1);
	\end{pgfonlayer}
\end{tikzpicture}
\]
is equal to the identity $1_{X \otimes X'}$ composed with $ f \otimes f' \maps X \otimes X' \to Y \otimes Y'$.

A monoidal, braided monoidal or symmetric monoidal category is \define{strict} if $\alpha_{X,Y,Z}$, $\lambda_X$, $\rho_X$ are all identities so that
\[ (X \otimes Y) \otimes Z = X \otimes (Y \otimes Z) \]
\[ I \otimes X = X = X \otimes I \]
One should note that these equations are not sufficient to imply that a category is strict. They are just consequences of strictness.

\begin{defn} \label{defn.smf}
Given symmetric monoidal categories $(\CC, \otimes, I_{\otimes})$ and $(\D, \boxtimes, I_{\boxtimes})$, a \define{symmetric monoidal functor} $F \maps (\CC, \otimes, I_{\otimes}) \to (\D, \boxtimes, I_{\boxtimes} )$ is a functor $F \maps \CC \to \D$ respecting the tensor product, so that
\[ F( X \otimes Y) = F(X) \boxtimes F(Y), \]
\[ F(f \otimes f') = F(f) \boxtimes F(f') \]
and
\[ F(I_{\otimes}) = I_{\boxtimes}. \]
\end{defn}

Studying the behaviors of such systems using symmetric monoidal functors
\[ F \maps (\OpenSys, \otimes) \to (\Behave, \boxtimes) \]
means that we need both a way of sticking systems side-by-side, $\otimes$, as well as a way of sticking behaviors side-by-side, $\boxtimes$. The fact that a symmetric monoidal functor respects these tensor products 
\[ F(f \otimes f') = F(f) \boxtimes F(f') \]
says that the behavior of $f \otimes f'$ should correspond to the tensor product of the behaviors of $f$ and $f'$.

\section{Decorated cospan categories}
We described the idea of considering open systems as morphisms in a category. We restrict our attention to systems built on some finite set of states. We utilize the decorated cospan approach, due to Fong \cite{Fong, FongThesis}. At the heart of this approach is the use of lax monoidal functors to describe the various structures one can put on a finite set in order to study the systems of interest. For instance one might use lax monoidal functors to assign to some finite set $S$ the set of all directed, labelled graphs with $S$ as its vertex set. Another example is using a lax monoidal functor to assign to a finite set $S$ the set of all vector fields $v \maps \R^S \to \R^S$ obeying some conditions. We will see that the key ingredient in the definition of a lax monoidal functor is in fact a natural transformation. 

\begin{defn} \label{defn.lmf}
  Let $(\CC,\boxtimes)$, $(\D,\otimes)$ be symmetric monoidal categories. A \define{lax symmetric
  monoidal functor} 
  \[
    (F,\Phi): (\CC,\boxtimes) \to (\D,\otimes)
  \]
  consists of a functor $F: \CC \to D$ together with natural transformations 
  \[
    \Phi_{-,-}: F(-)\otimes F(-) \Rightarrow F(-\boxtimes -),
  \]
  \[
    \Phi_I: I_{\otimes} \Rightarrow F(I_{\boxtimes})
  \]
  such that three so-called coherence diagrams commute. These diagrams are 
 \[
   \xymatrix{ F(X) \otimes (F(Y) \otimes F(Z)) \ar[d]_{\mathrm{1_{F(X)} } \otimes \Phi_{Y,Z}} 
                    \ar[rr]^{\sim} &&
                    (F(X) \otimes F(Y)) \otimes F(Z) \ar[d]^{\Phi_{X,Y}
		    \otimes \mathrm{1_{F(Z)} }}  \\
                    F(X) \otimes F(Y \boxtimes Z) \ar[d]_{\Phi_{X,Y\boxtimes Z}} &&
                     F(X \boxtimes Y) \otimes F(Z) \ar[d]^{\Phi_{X\boxtimes Y,Z}} \\
                     F(X \boxtimes (Y\boxtimes Z)) \ar[rr]^{\sim} && 
                     F((X \boxtimes Y) \boxtimes Z)                    
    }
  \]
where the horizontal arrows come from the associators for $\otimes$ and
$\boxtimes$, and
\[ 
  \xymatrix{
    I_{\otimes} \otimes F(X) \ar[rr]^{\Phi_I \otimes \mathrm{1_{ F(X)} } } \ar[d]_{\sim} &&
    F(I_{\boxtimes}) \otimes F(X) \ar[d]^{\Phi_{I_\boxtimes,X}} \\
    F(X) && F(I_\boxtimes \boxtimes X) \ar[ll]_{\sim}
  }
\]
and
\[ 
  \xymatrix{
    F(X) \otimes I_\otimes \ar[rr]^{\mathrm{1_{F(X)} } \otimes \Phi_I} \ar[d]_{\sim} &&
    F(X) \otimes F(I_{\boxtimes}) \ar[d]^{\Phi_{X,I_\boxtimes}} \\
    F(X) && F(X \boxtimes I_\boxtimes) \ar[ll]_{\sim}
  }
\]
where the isomorphisms come from the unitors for $\otimes$ and $\boxtimes$.
\end{defn}

Let $(\Set, \times)$ denote the category of sets and functors made into a symmetric monoidal category with the cartesian product as its tensor product.  We write $1$ for a chosen one-element set serving as the unit for this tensor product.  Given a category $\CC$ with finite colimits, we let $(\CC,+)$ denote this category made into a symmetric monoidal category with the coproduct as its tensor product, and write $0$ for a chosen initial object serving as the unit of this tensor product.

The decorated cospan construction is then as follows:

\begin{lem}[\textbf{Fong}] \label{lemma:fcospans}
  Suppose the $\CC$ is a category with finite colimits and 
  \[
    (F,\Phi)\maps (\CC,+) \longrightarrow (\Set, \times)
  \]
   is lax symmetric monoidal functor. 
   We may define a category $F\Cospan$, the category of
   \define{$F$-decorated cospans}, whose objects are those of $\CC$ and whose
   morphisms are equivalence classes of pairs 
  \[
    (X \stackrel{i}\longrightarrow V \stackrel{o}\longleftarrow Y, d)
  \]
  consisting of a cospan $X \stackrel{i}\rightarrow V \stackrel{o}\leftarrow Y$ in
  $\CC$ together with an element $d \in F(V)$ called a \define{decoration}.
   The equivalence relation arises from isomorphism of cospans; an 
  isomorphism of cospans induces a one-to-one correspondence between their sets of decorations.

Given two decorated cospans
 \[
    (X \stackrel{i}\longrightarrow V \stackrel{o}\longleftarrow Y, d) 
    \quad \textrm{ and } \quad
    (Y \stackrel{i'}\longrightarrow V' \stackrel{o'}\longleftarrow Z, d'), 
  \]
their composite consists of the equivalence class of this cospan constructed via a pushout:
  \[
    \xymatrix{
      && V +_Y V' \\
      & V \ar[ur]^{j} && V' \ar[ul]_{j'} \\
      \quad X\quad \ar[ur]^{i} && Y \ar[ul]_{o} \ar[ur]^{i'} &&\quad Z \quad \ar[ul]_{o'}
    }
  \]
together with the decoration obtained by applying the map
\[      
\xymatrix{      F(V) \times F(V') \ar[rr]^-{\Phi_{V,V'}} && 
                     F(V + V') \ar[rr]^-{F([j,j'])} && F(V +_Y V') } \]
to the pair $(d,d') \in F(V) \times F(V')$.  Here $[j,j'] \maps V + V' \to V +_Y V'$ is the
canonical morphism from the coproduct to the pushout.
\end{lem}

\begin{proof}
This is Proposition 3.2 in Fong's paper on decorated cospans \cite{Fong}.
\end{proof}

In fact $F\Cospan$ is a monoidal category, where the tensor product of objects is the disjoint union of sets and the tensor product of morphisms arising from two decorated cospans 
\[
    (X \stackrel{i}{\longrightarrow} V \stackrel{o}{\longleftarrow} Y, d) 
    \quad \textrm{ and } \quad
    (X' \stackrel{i'}{\longrightarrow} V' \stackrel{o'}{\longleftarrow} Y', d')
  \]
is the morphism associated to this decorated cospan:
\[  ( X + X' \stackrel{i+i'}{\longrightarrow} V + V'  \stackrel{o + o'}{\longleftarrow} Y + Y', \;
\Phi_{V,V'}(d,d') ) .\]
Fong proves that $F\Cospan$ is a specially nice sort of 
symmetric monoidal category called a `hypergraph category', and thus in particular
a dagger compact category, which first arose in the study of categorical quantum mechanics \cite{AC, Se}. The proof, and the definition and significance of these concepts, are clearly explained in his thesis \cite{FongThesis}.  

Decorated cospan categories help us understand diagrams of open systems, and
how to manipulate them.  We also need the ability to take such diagrams and read off 
their `meaning', for example via the rate equation.  For this we need functors between
decorated cospan categories.   Fong showed that we construct these from monoidal natural transformations.

\begin{defn}\label{def:monnattran}
  A \define{monoidal natural transformation} $\alpha$ from a lax symmetric monoidal
  functor
  \[    (F,\Phi)\maps (\CC,\otimes) \longrightarrow (\Set,\times)
  \]
  to a lax symmetric monoidal functor
  \[
    (G,\gamma)\maps (\CC,\otimes) \longrightarrow (\Set,\times)
  \]
  is a natural transformation $\alpha\maps F \Rightarrow G$ such that
  \[
    \xymatrix{
      F(A) \times F(B) \ar[r]^-{\Phi_{A,B}} \ar[d]_{\alpha_A \times
      \alpha_B} & F(A \otimes B) \ar[d]^{\alpha_{A \otimes B}} \\
      G(A) \times G(B) \ar[r]^-{\gamma_{A,B}} & G(A \otimes B)
    }
  \]
  commutes.
\end{defn}

We then construct functors between decorated cospan categories as follows:

\begin{lem}[\textbf{Fong}] \label{lemma:decoratedfunctors}
  Let $\CC$ be a category with finite colimits and let
  \[
    (F,\Phi) \maps (\CC,+) \longrightarrow (\Set,\times)
  \]
  and
  \[
    (G,\gamma) \maps (\CC,+) \longrightarrow (\Set,\times)
  \]
  be lax symmetric monoidal functors. This gives rise to decorated cospan
  categories $F\Cospan$ and $G\Cospan$. 

  Suppose then that we have a monoidal natural transformation $\theta_V \maps F(v) \to G(v)$.   
  Then there is a symmetric monoidal functor 
  \[
    T \maps F\Cospan \longrightarrow G\Cospan
  \]
  mapping each object $X \in F\Cospan$ to $X \in G\Cospan$, and
  each decorated cospan
  \[
    (X \stackrel{i}\longrightarrow V \stackrel{o}\longleftarrow Y, d)
  \]
  to
  \[
      (X \stackrel{i}\longrightarrow V \stackrel{o}\longleftarrow Y, \theta_V (d) ).
  \]
\end{lem}
\begin{proof}
This is a special case of Theorem 4.1 in Fong's paper on decorated cospans \cite{Fong}.  He proves that the functor $T$ is actually a `hypergraph functor', the
sort that preserves the structure of a hypergraph category.  As a corollary it is a symmetric monoidal dagger functor.
\end{proof}

Throughout this paper we typically take the category $\CC$ to be $\FinSet$. Thus, given lax symmetric moniodal functors $F,G \maps (\FinSet, +) \to (\Set, \times)$, Lemma \label{lemma:decoratedfunctors} gives a symmetric monoidal functor
$T \maps F\Cospan \to G\Cospan$ from a monoidal natural transformation $\theta \maps F \Rightarrow G$.

\chapter{Open Markov processes}\label{ch:openmark}
\label{sec:open_markov}
In Chapter \ref{ch:entropymark}, we saw that open Markov processes are Markov processes with specified boundary states through which probability can flow in and out of the process.  By externally fixing the probabilities at the boundary states of an open Markov process one generally induces a non-equilibrium steady state for which non-zero probability currents flow through the boundary of the process in order to maintain such a steady state. In the present chapter we show that open Markov processes can be seen as morphisms in a symmetric monoidal category, in particular a decorated cospan category as described in the previous chapter.  Composition in this category corresponds to gluing together two open Markov processes along a shared subset of their boundary states. This gives a way to build up complicated Markov processes out of smaller open Markov processes, or alternatively to break down a complicated Markov process into a number of simpler interacting open Markov processes.  

We thus divide the boundary of an open Markov process into `input' states and `output' states and consider an open Markov process as a morphism from the input states to the output states. We should emphasize that this distinction between input and output states is completely arbitrary. Probability can flow either in or out of an input or output. Inputs and outputs need not be disjoint, i.e.\ a state can appear as both an input and an output. Composition of two open Markov processes corresponds to gluing the outputs of the first process onto the inputs of the second. The resulting open Markov process will have the inputs of the first process as its inputs and the outputs of the second process as its outputs. Thus the whole point of breaking the boundary of an open Markov process into specified inputs and outputs is simply to make more transparent the fact that composition of two open Markov processes corresponds to gluing the process along certain compatible subsets of their boundaries.

In Chapter \ref{ch:blackbox}, we go on to prove a `black-boxing' theorem for open Markov processes. The basic idea is to characterize an open Markov process in terms of the space of allowed steady state probabilities and probability currents along the boundary of the process, forgetting the details regarding the internal states of the process; hence the term black boxing. We call this space of possible steady state probability, current pairs the `behavior' of an open Markov process. This effectively provides an equivalence relation on the set of all open Markov processes whereby two open Markov processes are considered equivalent if they give rise to the same space of possible steady state boundary probabilities and currents. In particular, we show that there is a `black-box functor' from the category of open Markov process to the category of relations. The fact that `black-boxing' is accomplished via a functor means that the behavior of a composite open Markov process can be computed as the composite of the behaviors of its constituents. 

\section{The category of open Markov processes}

Recall that, mathematically, a \define{cospan} of sets consists of a set $V$ together with functions
$i \maps X \to V$ and $o \maps Y \to V$.  We draw a cospan as follows:
\[ \xymatrix{ & V & \\ 
X \ar[ur]^{i} & &  Y \ar[ul]_{o} \\ } \]
The set $V$ describes the system, $X$ describes its inputs, and $Y$ its outputs.  

For open Markov processes $V$ is the set of states of some Markov process, and the maps $i \maps X \to V$ and $o \maps Y \to V$ specify how the input and output states are included in $V$.   We do not require these maps to be one-to-one.

We then `decorate' the cospan with a complete description of the system. To decorate the above cospan with a Markov process, we attach to it the extra data 
\[ \xymatrix{ (0,\infty) & E \ar[l]_-r \ar[r]<-.5ex>_t \ar[r] <.5ex>^s & V  } \] 
describing a Markov process on $V$. Thus, we make the following definition:
\begin{defn}
Given finite sets $X$ and $Y$, an \define{open Markov process from} $X$ \define{to} $Y$ is a cospan of finite sets
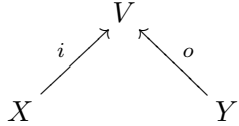
\[ \xymatrix{ & V & \\
X \ar[ur]^{i} && Y \ar[ul]_{o} \\ } \]
together with a Markov process $M$ on $V$.  We often abbreviate such an open Markov process simply as $M \maps X \to Y$.  We say $X$ is the set of \define{inputs} of the open Markov process and $Y$ is its set of \define{outputs}. We define a \define{boundary node} to 
be a node in $B = i(X) \cup o(Y)$, and call a node 
\define{internal} if it is not a boundary node.
\end{defn}

We draw an open Markov process $M \maps X \to Y$ in the following way:
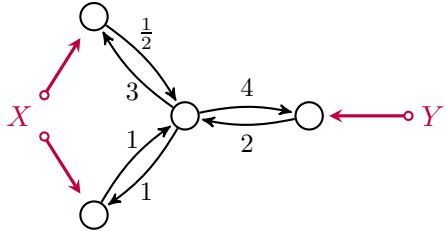
\[
\begin{tikzpicture}[->,>=stealth',shorten >=1pt,thick,scale=1.1]
  \node[main node] (1) at (-8.4,2.2) {};
  \node[main node](2) at (-8.4,-.2) {};
  \node[main node](3) at (-7.3,1)  {};
  \node[main node](4) at (-5.8,1) {};
\node(input)[color=purple] at (-9.3,1) {$X$};
\node[terminal, scale=.3] (A) at (-9,1.25) {};
\node[terminal,scale=.3] (B) at (-9,.75) {};
\node(output)[color=purple] at (-4.3,1) {$Y$};
\node[terminal, scale=.3](D) at (-4.6,1) {};
  \path[every node/.style={font=\sffamily\small}, shorten >=1pt]
    (3) edge [bend left=12] node[above] {$4$} (4)
    (4) edge [bend left=12] node[below] {$2$} (3)
    (2) edge [bend left=12] node[above] {$1$} (3) 
    (3) edge [bend left=12] node[below] {$1$} (2)
    (1) edge [bend left=12] node[above] {$\frac{1}{2}$}(3) 
    (3) edge [bend left=12] node[below] {$3$} (1);
    
\path[color=purple, very thick, shorten >=2pt, ->, >=stealth] (D) edge (4);
\path[color=purple, very thick, shorten >=4pt, ->, >=stealth] (A) edge (1);
\path[color=purple, very thick, shorten >=4pt, ->, >=stealth]
(B) edge (2);
\end{tikzpicture}
\]

As we mentioned in Chapter \ref{ch:entropymark}, since probability can flow in and out of an open Markov process, we work with non-normalized probability distributions. We can understand this in the following way. A closed Markov process is a special case of an open Markov process that has no inputs, no outputs, and therefore no boundary. Imagine then that we wish to consider a subsystem of a closed system as an open system:
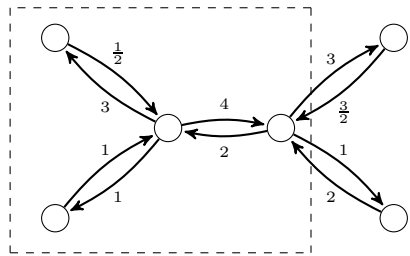
\[
\begin{tikzpicture}
  \node[main node] (1) at (-5.8,2.2) {};
  \node[main node](2) at (-5.8,-.2) {};
  \node[main node](3) at (-4.3,1)  {};
  \node[main node](4) at (-2.8,1) {};
  \node[main node] (5) at (-1.3,2.2) {};
  \node[main node](6) at (-1.3,-.2) {};
  \coordinate (7) at (-2.4,2.6) {};
  \coordinate (8) at (-2.4,-.66) {};
  \coordinate (9) at (-6.4,2.6) {};
  \coordinate (10) at (-6.4,-.66) {};

\draw[dashed] (7) -- (8) ;
\draw[dashed] (10) -- (8) ; 
\draw[dashed] (9) -- (10) ; 
\draw[dashed] (7) -- (9) ; 

  \path[every node/.style={font=\sffamily\tiny},->,>=stealth',shorten >=1pt,thick,scale=1.1]
    (3) edge [bend left=12] node[above] {$4$} (4)
    (4) edge [bend left=12] node[below] {$2$} (3)
    (2) edge [bend left=12] node[above] {$1$} (3) 
    (3) edge [bend left=12] node[below] {$1$} (2)
    (1) edge [bend left=12] node[above] {$\frac{1}{2}$}(3) 
    (3) edge [bend left=12] node[below] {$3$} (1)
    (6) edge [bend left=12] node[below] {$2$} (4) 
    (4) edge [bend left=12] node[above] {$1$} (6)
    (5) edge [bend left=12] node[below] {$\frac{3}{2}$}(4) 
    (4) edge [bend left=12] node[above] {$3$} (5);
    
\end{tikzpicture}
\]
Even if probability is conserved and normalized to unity in the closed process, it will in general be sub-normalized in the open system. We go even further and work with non-normalized probabilities throughout. These non-normalized probabilities can take arbitrary non-negative values.

Given open Markov processes $M \maps X \to Y$ and $M' \maps Y \to Z$ we can compose them

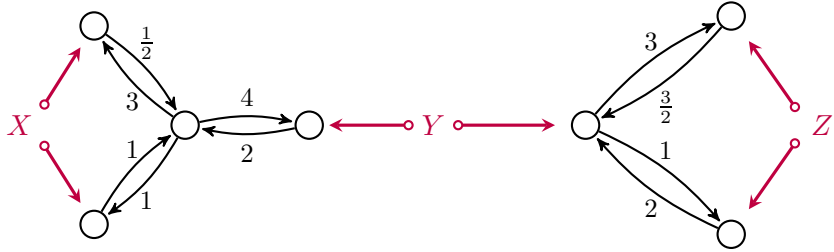
\[
\begin{tikzpicture}[->,>=stealth',shorten >=1pt,thick,scale=1.1,]
  \node[main node] (1) at (-8.4,2.2) {};
  \node[main node](2) at (-8.4,-.2) {};
  \node[main node](3) at (-7.3,1)  {};
  \node[main node](4) at (-5.8,1) {};
\node(input)[color=purple] at (-9.3,1) {$X$};
\node[terminal, scale=.3] (A) at (-9,1.25) {};
\node[terminal,scale=.3] (B) at (-9,.75) {};
\node(output)[color=purple] at (-4.3,1) {$Y$};
\node[terminal, scale=.3](D) at (-4.6,1) {};
  \path[every node/.style={font=\sffamily\small}, shorten >=1pt]
    (3) edge [bend left=12] node[above] {$4$} (4)
    (4) edge [bend left=12] node[below] {$2$} (3)
    (2) edge [bend left=12] node[above] {$1$} (3) 
    (3) edge [bend left=12] node[below] {$1$} (2)
    (1) edge [bend left=12] node[above] {$\frac{1}{2}$}(3) 
    (3) edge [bend left=12] node[below] {$3$} (1);
    
\path[color=purple, very thick, shorten >=2pt, ->, >=stealth] (D) edge (4);
\path[color=purple, very thick, shorten >=4pt, ->, >=stealth] (A) edge (1);
\path[color=purple, very thick, shorten >=4pt, ->, >=stealth]
(B) edge (2);
\begin{scope}[shift={(-4,-0.1)},->,>=stealth',shorten >=1pt,thick,scale=1.1]
  \node[main node] (1) at (3,2.2) {};
  \node[main node](2) at (3,-.2) {};
  \node[main node](3) at (1.4,1)  {};
\node(input)[color=purple] at (-.3,1) {$$};
\node[terminal,scale=.3](F) at (0,1) {};
\node(output)[color=purple] at (4,1) {$Z$};
\node[terminal,scale=.3](E) at (3.7,1.2) {};
\node[terminal,scale=.3](G) at (3.7,.8) {};
  \path[every node/.style={font=\sffamily\small}, shorten >=1pt]
    (2) edge [bend left=12] node[below] {$2$} (3) 
    (3) edge [bend left=12] node[above] {$1$} (2)
    (1) edge [bend left=12] node[below] {$\frac{3}{2}$}(3) 
    (3) edge [bend left=12] node[above] {$3$} (1);
    
\path[color=purple, very thick, shorten >=6pt, ->, >=stealth] (F) edge (3);
\path[color=purple, very thick, shorten >=6pt, ->, >=stealth] (E) edge (1);
\path[color=purple, very thick, shorten >=6pt, ->, >=stealth]
(G) edge (2);
\end{scope}
\end{tikzpicture}
\]
yielding a new open Markov process $MM' \maps X \to Z$
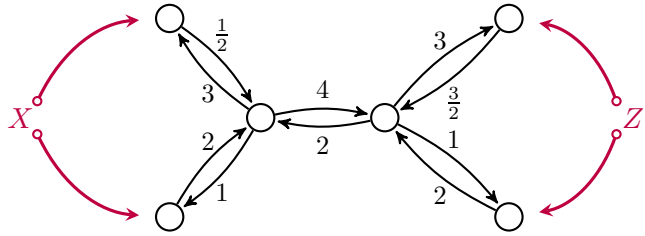
\[
\begin{tikzpicture}[->,>=stealth',shorten >=1pt,thick,scale=1.1]
  \node[main node] (1) at (-5.4,2.2) {};
  \node[main node](2) at (-5.4,-.2) {};
  \node[main node](3) at (-4.3,1)  {};
  \node[main node](4) at (-2.8,1) {};
  \node[main node] (5) at (-1.3,2.2) {};
  \node[main node](6) at (-1.3,-.2) {};
\node(input)[color=purple] at (-7.2,1) {$X$};
\node[terminal,scale=.3](A) at (-7,1.2) {};
\node[terminal, scale=.3] (B) at (-7,.8) {};
\node(output)[color=purple] at (.2,1) {$Z$};
\node[terminal, scale=.3](E) at (0,1.2) {};
\node[terminal, scale=.3](G) at (0,.8){};
  \path[every node/.style={font=\sffamily\small}, shorten >=1pt]
    (3) edge [bend left=12] node[above] {$4$} (4)
    (4) edge [bend left=12] node[below] {$2$} (3)
    (2) edge [bend left=12] node[above] {$2$} (3) 
    (3) edge [bend left=12] node[below] {$1$} (2)
    (1) edge [bend left=12] node[above] {$\frac{1}{2}$}(3) 
    (3) edge [bend left=12] node[below] {$3$} (1)
    (6) edge [bend left=12] node[below] {$2$} (4) 
    (4) edge [bend left=12] node[above] {$1$} (6)
    (5) edge [bend left=12] node[below] {$\frac{3}{2}$}(4) 
    (4) edge [bend left=12] node[above] {$3$} (5);
    
\path[color=purple, very thick, shorten >=6pt, ->, >=stealth, bend left] (A) edge (1);
\path[color=purple, very thick, shorten >=6pt, ->, >=stealth, bend right]
(B) edge (2);
\path[color=purple, very thick, shorten >=6pt, ->, >=stealth, bend right] (E) edge (5);
\path[color=purple, very thick, shorten >=6pt, ->, >=stealth, bend left]
(G) edge (6);
\end{tikzpicture}
\]

With this notion of composition, open Markov processes form a category. There is however one caveat. The definition of a category requires an \emph{associative} composition operation. An open Markov process is a cospan of finite sets, together with a Markov process on the apex of the cospan. We saw in the previous chapter that composition of cospans is accomplished via the pushout. As it turns out, composition of cospans via pushout is associative only up to isomorphism. What is really going on is that we are actually dealing with the structure of a \emph{bicategory} \cite{Benabou, Kenny, Leinster}. Since telling the story of bicategories would take us quite far afield and because for our purposes we need only the structure of a category, we can simply take the morphisms in our category to be \emph{isomorphism classes} of open Markov processes. 

We wish to show that there is a decorated cospan category whose morphisms correspond to open Markov processes. Recall that the key ingredient to form a decorated cospan cateogry is a lax symmetric monoidal functor $F \maps \FinSet \to \Set$ sending any finite set $V$ to $F(V)$, the set of all Markov processes on $V$. For this we need to introduce a bit of machinery. Given a finite set $V$ and a function $f \maps V \to V'$ between finite sets, then for any Markov process $(V,E,s,t,r)$ on $V$, we'd like a way to cook up a Markov process on $V'$. Looking at the diagram
\[ \xymatrix{ (0,\infty) & E \ar[l]_-r \ar[r]<-.5ex>_t \ar[r] <.5ex>^s & V \ar[r]^f & V' } \] 
the answer becomes clear. We can simply define a new Markov process $(V',E,s',t',r)$ with $s' = f \circ s$ and $t' = f \circ s$. 

\begin{lem}
\label{lemma:MarkFunctor}
There is a functor $F \maps \FinSet \to \Set$ such that:
\begin{itemize} 
\item For any finite set $V$, F(V) is the set of all Markov processes on $V$.
\item For any function $f \maps V \to V'$ and any Markov process $(V,E,s,t,r) \in F(V)$, we have
\[ F(f)(V,E,s,t,r) = (V',E, f \circ s, f \circ t, r)  \]
a Markov process on $V'$.
\end{itemize}
\end{lem}
\begin{proof} We need to check that $F$ preserves composition and sends identities to identities. Both are straightforward calculations.
\end{proof}

To get a decorated cospan category whose morphisms are open Markov processes we need that $F \maps \FinSet \to \Set$ be a lax monoidal functor.
\begin{lem}
\label{lemma:MarkLax}
For any pair of finite sets $V$ and $V'$, there is a map $\varphi_{V,V'} \maps F(V) \times F(V') \to F(V+V') $ which gives us a way to think of a pair of open Markov processes, one on $V$ and the other on $V'$, as a single open Markov process on $V+V'$. This map makes $F \maps \FinSet \to \Set$ into a lax monoidal functor. 
\end{lem}
\begin{proof}
We define the natural transformation
\[ \phi_{V,V'} \maps F(V) \times F(V') \to F(V+V') \]
via the assignment
\[ \phi_{V,V'} \maps (V,E,s,t,r) \times (V',E',s',t',r') \mapsto (V+V', E+E', s+s', t+t', [r,r']) \]
where $+$ denotes the coproduct in $\FinSet$ and where the \define{copairing} of functions $[r,r'] \maps E+E' \to (0,\infty)$ sends any edge in $E$ to its transition rate $r(e)$ and similarly for $E'$. 

This together with the unit map
\[ \varphi \maps 1 \to F(\emptyset) \]
assinging the unique Markov process with no states and no edges to the one-element set $1$. Commutativity of the hexagon and left/right unitor squares required for laxness of $F$ follows from the universal property of the coproduct in $\FinSet$.
\end{proof}

Thus we arrive at the following definition:
\begin{thm}
\label{thm:OpenMark}
There is a category $\OpenMark$ where:
\begin{itemize}
\item an object is a finite set and
\item a morphism from $X$ to $Y$ is an equivalence class of open Markov processes from $X$ to $Y$.
\item Given morphisms represented by an open Markov process from $X$ to $Y$ and one from $Y$ to $Z$:
 \[
    (X \stackrel{i}\longrightarrow V \stackrel{o}\longleftarrow Y, M) 
    \quad \textrm{ and } \quad
    (Y \stackrel{i'}\longrightarrow V' \stackrel{o'}\longleftarrow Z, M'), 
  \]
their composite consists of the equivalence class of this cospan:
  \[
    \xymatrix{
      & M +_Y M' \\
      \quad X\quad \ar[ur]^{ji} && \quad Z \quad \ar[ul]_{j'o'}
    }
  \]
together with an open Markov process on $V +_Y V'$ 
obtained by applying the map
\[      
\xymatrix{      F(V) \times F(V') \ar[rr]^-{\varphi_{V,V'}} && 
                     F(V + V') \ar[rr]^-{F([j,j'])} && F(V +_Y V') } \]
to the pair $(M,M') \in F(V) \times F(V')$.  
\end{itemize}
The category $\OpenMark$ is a symmetric monoidal category where the tensor product of objects $X$ and $Y$ is their disjoint union $X + Y$, while the tensor product of the morphisms
\[
    (X \stackrel{i}{\longrightarrow} V \stackrel{o}{\longleftarrow} Y, M) 
    \quad \textrm{ and } \quad
    (X' \stackrel{i'}{\longrightarrow} V' \stackrel{o'}{\longleftarrow} Y', M') 
  \]
is defined to be
\[  ( X + X' \stackrel{i+i'}{\longrightarrow} V + V'  \stackrel{o + o'}{\longleftarrow} Y + Y', \; \varphi_{V,V'}(M,M') ) .\]
In fact $\OpenMark$ is a hypergraph category.
\end{thm}

\begin{proof}
This follows from Lemmas \ref{lemma:fcospans} and \ref{lemma:MarkLax}, 
where we explain the equivalence relation in detail.  
\end{proof}

Composition in $\OpenMark$ corresponds to gluing processes together along their overlap. Since input and output maps need not be one-to-one, this procedure can result in the identification of a number of states at once, not simply pairs of states. 

\section{The open master equation}
We now introduce a version of the open master equation in which the master equation is modified by additional flows at its inputs and outputs. Consider an open Markov process $M \maps X \to Y$ consisting of a cospan of finite sets \[ \xymatrix{ & V & \\
X \ar[ur]^{i} && Y \ar[ul]_{o} \\ } \]
together with a Markov process
\[ \xymatrix{ (0,\infty) & E \ar[l]_-r \ar[r]<-.5ex>_t \ar[r] <.5ex>^s & V }  \]
on $V$. Given specified \define{inflows} $I(t) \in \R^X$ and \define{outflows} $O(t) \in \R^Y$ together with input and output maps $i \maps X \to V$ and $o \maps Y \to V$, we can pushforward the inflow and outflows to get vectors valued in $\R^V$. That is, given $I(t) \in \R^X$ and $i \maps X \to V$ we can define
\[ i_*(I(t))_v = \sum_{ \{x \in X | i(x) = v \} } I(t)_x \]
yielding a vector $i_*(I(t)) \in \R^V$.
Similarly we can define
\[ o_*(O(t)_v = \sum_{ \{ x \in X| o(x) =v \} } O(t)_x\]
for the outputs. With these pushforwards in hand we can define the open master equation.

\begin{defn}
Given an open Markov process $M \maps X \to Y$, consisting of a cospan of finite sets
\[ \xymatrix{ & V & \\
X \ar[ur]^{i} && Y \ar[ul]_{o} \\ } \]
together with a Markov process
\[ \xymatrix{ (0,\infty) & E \ar[l]_-r \ar[r]<-.5ex>_t \ar[r] <.5ex>^s & V }  \]
on $V$, together with specified \define{inflows} $I(t)\in \R^X$ and $O(t) \in \R^Y$,  the \define{open master equation} is given by
\[ \frac{dp}{dt} = Hp + i_*(I) - o_*(O). \]
\end{defn}

We will be especially interested in `steady state' solutions of the open master equation:

\begin{defn}
A \define{steady state} solution of the open master equation is a solution $p(t) \maps V \to [0,\infty)$ such that $\frac{dp}{dt} = 0$. 
\end{defn}

In Section \ref{sec:balance} we turn to open Markov processes which admit a detailed balanced equilibrium.  We saw that steady state solutions of the open master equation minimize the dissipation.  This is analogous to the minimization of power dissipation by circuits made of linear resistors.   In Section \ref{sec:reduction_2}, this analogy lets us reduce the black boxing problem for detailed balanced Markov processes to the analogous, and already solved, problem for circuits of resistors.

\section{Markov processes with energies} 
\label{sec:balance}

We are especially interested in Markov processes with a chosen detailed balanced equilibrium $q$.  This means that with this particular choice
of $q$, the flow of probability from $i$ to $j$ equals the flow from $j$ to $i$.  This ensures that $q$ is an equilibrium, but it is a significantly stronger condition. If the underlying graph of a Markov process is connected, then this distribution is unique. For Markov processes with multiple connected components, one can separately scale the components of the equilibrium distribution on each component. 

As we saw, the existence of a detailed balanced equilibrium is in fact a condition on the rates of a Markov process. This condition, known as Kolomogorov's criterion, states that the product of the rates around any cycle in a Markov process is the same in the forward as in the reverse direction and is a necessary and sufficient condition for the existence of a detailed balanced equilibrium distribution. We show that this is equivalent to an assignment of `energies' to each state in the Markov process with the equilibrium probabilities proportional to the Boltzmann factor
\[ q_i \propto e^{-\beta \epsilon_i}, \]
where $\beta$ plays the role of an inverse temperature $\frac{1}{T}$, here using units where Boltzmann's constant is equal to one along with the requirement that the ratio of the forward and reverse transition rates between any pair of states satisfies
\[ \frac{H_{ij}}{H_{ji}} = e^{-\beta (\epsilon_i - \epsilon_j) }. \] To see this we write
\[ \frac{H_{ij} q_j}{H_{ji}q_i }  = \frac{ H_{ij} } {H_{ji} } e^{-\beta (\epsilon_j - \epsilon_i) }. \]
For detailed balanced equilibrium distributions, the left-hand-side is equal to one and we have that
\[ \frac{H_{ji} } {H_{ij} } = e^{-\beta (\epsilon_j - \epsilon_i ). } \]
Recall that $H_{ji}$ is the transition rate from $i$ to $j$. This is large when $\epsilon_i >> \epsilon_j$, reflecting the intuition that the lower energy state is preferred.   

Kolmogorov's criterion can be written
\[ H_{12}H_{23} \cdots H_{n-1 n} H_{n1}  = H_{1n} H_{n n-1}\cdots H_{32}H_{21}  \]
for any sequence of states $1,2,\dots,n,1$. 
For finite energies and temperatures, the detailed balance condition implies that if $H_{ij} = 0$, then so is $H_{ji}$. Assuming all the transition rates are non-zero in the above expression we can write
\[ \frac{ H_{12} H_{23} \cdots H_{n-1 n } H_{n1} } { H_{21} H_{32} \cdots H_{n n-1} H_{1n} } = e^{-\beta (E_1 - E_2) } e^{ -\beta (E_2 - E_3) } \cdots e^{-\beta (E_{n-1} - E_{n} ) } e^{ -\beta ( E_n - E_1) } = 1. \]
So we see that the detailed balance condition for equilibrium distributions which are Gibbs states implies Kolmogorov's criterion on the rates of a Markov process. Alternatively suppose that Kolmogorov's criterion holds, then one can always assign energies to the states so that 
\[ q_i \propto e^{-\beta \epsilon_i } \]
is a detailed balanced equilibrium. Here we reproduce the argument from Kelly's book \cite{Kelly}. Let us assume that the Markov process is irreducible and that its rates satisfy Kolmogorov's criterion. Then define
\[ q_j = c \frac{H_{10} H_{21} \cdots H_{jn} } { H_{nj} \cdots H_{12} H_{01} } \]
for some positive constant $c$ and for any sequence of states $0,1,\dots n,j$. Note that $q_j$ does not depend on the particular sequence of states from $0$ to $j$ since the rates satisfy Kolmogorov's criterion. Let us see that such an expression indeed gives a detailed balanced equilibrium. Consider another state $k$ and suppose that $H_{jk} \neq 0$ so that the detailed balance condition between $j$ and $k$ is non-trivial, then we can write
\[ q_k = c \frac{H_{10} H_{21} \cdots H_{jn} H_{kj} }{H_{jk} H_{nj} \cdots H_{01} }, \]
from which we see that 
\[ H_{kj}q_j = H_{jk} q_k.\]

We can then define $\epsilon_j$ via the equation $q_j = c e^{-\beta \epsilon_j}$ which via detailed balance determines the energies of the other states. Note that changing the proportionality factor $c$ corresponds to a uniform translation of energies.  

Intuitively we can understand this as an analog of the existence of a single-valued potential function giving rise to a conservative force-field. Later we will make this more apparent when write the master equation as the gradient of a certain `potential function' on a space of probabilities whose inner-product is weighted by the components of the detailed balanced equilibrium distribution. Non-trivial probability currents can flow around cycles in Markov processes which violate the Kolmogorov criterion. In terms of energies this corresponds to cycles which when transversed result in an overall change in energy, implying some non-trivial aspects in terms of the underlying energy landscape. 

Recall that given a Markov process $(V,E,s,t,r)$, the off-diagonal entries of the Hamiltonian $H_{ij}$ for $i \neq j$ are determined by the sum of the rates along edges going from $j$ to $i$  
\[ H_{ij} = \mathop{ \sum r_e }_{ e \maps j \to i }  . \ \  \]
For $i \neq j$ the ratio $\frac{H_{ij}}{H_{ji}}$ is given in terms of the rates via 
\[ \frac{H_{ij} }{H_{ji} } = \frac{ \displaystyle{ \mathop{ \sum r_e }_{e \maps j \to i } } } { \displaystyle{ \mathop{ \sum r_e}_{ e \maps i \to j } . } } \]
For an assignment of energies to states $ \epsilon \maps V \to \R$ to be compatible with a detailed balanced equilibrium of the form $q_i \propto e^{-\beta \epsilon_i}$ means that the rates of the Markov process must satisfy
\[ \frac{H_{ij}}{H_{ji}} =  \frac{ \displaystyle{ \mathop{ \sum r_e }_{ e \maps j \to i } } } { \displaystyle{ \mathop{ \sum r_e}_{ e \maps i \to j }  } }    = e^{-\beta( \epsilon_i - \epsilon_j). } \]
We call this condition the \define{detailed balance condition}.

\begin{defn}
A \define{Markov process with energies} 
\[ \xymatrix{  (0,\infty) & E \ar[l]_-r \ar[r]<-.5ex>_t  \ar[r] <.5ex>^s & V 
\ar[r]^-{\epsilon} & \R }  \]
is a Markov process
\[ \xymatrix{   (0,\infty) & E \ar[l]_-r \ar[r]<-.5ex>_t  \ar[r] <.5ex>^s & V 
  }  \]
together with a map $\epsilon \maps V \to \R$ assigning an \define{energy} to each state satisfying
\[   \frac{ \displaystyle{ \mathop{ \sum r_e }_{ e \maps j \to i } } } { \displaystyle{ \mathop{ \sum r_e}_{ e \maps i \to j }  } }    = e^{-\beta( \epsilon_i - \epsilon_j). } \]
We sometimes write the collection $(V,E,s,t,r,\epsilon)$ to denote the data of a Markov process with energies.
\end{defn}

A map of Markov processes with energies is similar to a map of Markov processes. Given two Markov processes with energies $M=(V,E,s,t,r,\epsilon)$ and $M' = (V',E',s',t',r',\epsilon')$, a \define{map of Markov processes with energies} $f \maps M \to M'$ is a pair of functions $f_E \maps E \to E'$ and $f_V \maps V \to V'$ such that the following diagrams commute 

\[ \xymatrix{   & \ar[ld]_r E \ar[dd]_{f_E}   \ar[rr]^s & &  V \ar[dd]^{f_V} \ar[rd]^{\epsilon} & \\
(0,\infty) &  & & & \R \\     
   & \ar[ul]^{r'}  E'  \ar[rr]_{s'} & & V'  \ar[ru]_{\epsilon'} & }  \ \ \ \  \xymatrix{    & \ar[dl]^r E \ar[dd]_{f_E} \ar[rr]^t    & & V \ar[dd]^{f_V} \ar[rd]^{\epsilon} & \\ 
   (0,\infty) & & & & \R \\    
   & \ar[ul]^{r'} E'  \ar[rr]_{t'} &  & V' \ar[ru]_{\epsilon'} & }  \]
Notice that commutativity of the right most triangles requires that states which are mapped onto one another must have the same energies. Later, we will see how this requirement enters into the composition of open Markov processes with energies, namely that only states with equal energies can be identified.

One should note that in what follows we consider a Markov process whose states $V$ are labelled by energies $\epsilon \maps V \to \R$ to be equivalent to a Markov process whose states $V$ are labelled by a detailed balanced equilibrium $q \maps V \to (0,\infty)$ via the assignment $q_i = \frac{e^{-\beta E_i}}{\mathcal{Z}}$. A uniform translation of energies $\epsilon \mapsto \epsilon + c$ for some constant vector $c \in \R^V$, corresponds to a uniform scaling of $q$ via the factor $q \mapsto  e^{-\beta c}q$.

\section{Open Markov processes with energies}
We now describe a decorated cospan category where the morphisms are open detailed balanced Markov processes, or rather open Markov processes with energies. So far the objects in our decorated cospan categories have simply been finite sets. Composition corresponded to gluing together open processes along some finite set of shared boundary states. We saw in the previous section that for Markov processes whose states are labelled with energies, states can only be identified if they have the same energy. In terms of the decorated cospan approach, this corresponds to considering cospans not in $\FinSet$, but in $\FinSet_{\epsilon}$ the category of finite sets with energies and maps between them. 

\begin{defn}
A \define{finite set with energies} is a finite set $X$ together with a map $\epsilon \maps X \to \R$. We write $\epsilon_i$ for the energy of some point $i \in X$.
\end{defn}

\begin{defn} 
Given two finite sets with energies $(X,\epsilon)$ and $(X', \epsilon')$ we define a \define{map} $f \maps (X,\epsilon) \to (X', \epsilon')$ to be a function $f \maps X \to X'$ such that the following diagram commutes
\[ \xymatrix{  X \ar[dr]_{\epsilon} \ar[rr]^{f} & & X' \ar[dl]^{\epsilon'} \\ & \R & }. \] There is a category $\FinSet_{\epsilon}$ of finite sets with energies and maps between them.
\end{defn}

A cospan in $\FinSet_{\epsilon}$ is a diagram
\[ \xymatrix{ & (V,\epsilon) & \\
(X,\chi) \ar[ur]^{i} && (Y,\upsilon) \ar[ul]_{o} \\ } \]
where $i \maps (X,\chi) \to (V,\epsilon)$ and $o \maps (Y,\upsilon) \to (V, \epsilon)$ are maps of finite sets with energies meaning that
\[  \epsilon(i(x)) = \chi(x) \] and
\[ \epsilon(o(y)) = \upsilon(y) \]
for all $x \in X$ and $y \in Y$.

Note that any Markov process with energies $M=(V,E,s,t,r,\epsilon)$ has an underlying set with energies $(V,\epsilon)$ and well as an underlying graph $(V,E,s,t)$. 

\begin{defn}
Given sets with energies $(X, \chi)$ and $(Y, \upsilon)$, an \define{open Markov process with energies} consists of a cospan of finite sets with energies
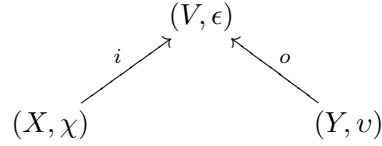
\[ \xymatrix{ & (V,\epsilon) & \\
(X,\chi) \ar[ur]^{i} && (Y,\upsilon) \ar[ul]_{o} \\ } \]
together with a Markov process with energies $M$ on $(V,\epsilon)$.   We often abbreviate this as $M \maps (X,\chi) \to (Y,\upsilon)$.
\end{defn}

We can draw an open Markov process with energies $M \maps (X, \chi) \to (Y, \upsilon)$ as
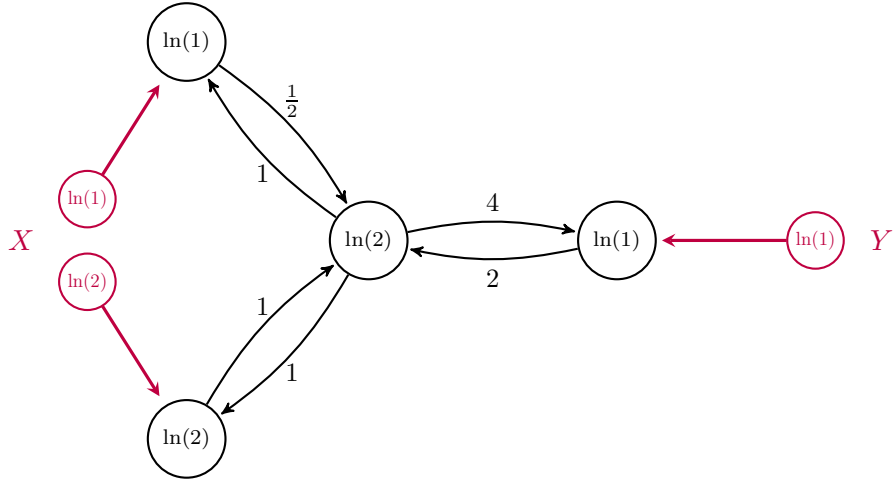
\[
\begin{tikzpicture}[->,>=stealth',shorten >=1pt,thick,scale=2.2]
  \node[main node] (1) at (-8.4,2.2) {\scriptsize$ \ln(1)$};
  \node[main node](2) at (-8.4,-.2) {\scriptsize$\ln(2)$};
  \node[main node](3) at (-7.3,1)  {\scriptsize$\ln(2)$};
  \node[main node](4) at (-5.8,1) {\scriptsize$\ln(1)$};
\node(input)[color=purple] at (-9.4,1) {$X$};
\node[terminal, scale=.5] (A) at (-9,1.25) {$\ln(1)$};
\node[terminal,scale=.5] (B) at (-9,.75) {$\ln(2)$};
\node(output)[color=purple] at (-4.2,1) {$Y$};
\node[terminal, scale=.5](D) at (-4.6,1) {$\ln(1)$};
  \path[every node/.style={font=\sffamily\small}, shorten >=1pt]
    (3) edge [bend left=12] node[above] {$4$} (4)
    (4) edge [bend left=12] node[below] {$2$} (3)
    (2) edge [bend left=12] node[above] {$1$} (3) 
    (3) edge [bend left=12] node[below] {$1$} (2)
    (1) edge [bend left=12] node[above] {$\frac{1}{2}$}(3) 
    (3) edge [bend left=12] node[below] {$1$} (1);
    
\path[color=purple, very thick, shorten >=2pt, ->, >=stealth] (D) edge (4);
\path[color=purple, very thick, shorten >=4pt, ->, >=stealth] (A) edge (1);
\path[color=purple, very thick, shorten >=4pt, ->, >=stealth]
(B) edge (2);
\end{tikzpicture}
\]

Recall that we require the rates of an open Markov process with energies to satisfy
\[ \frac{H_{ij}}{H_{ji}} = e^{-\beta ( \epsilon_i - \epsilon_j). } \]
Setting $\beta = 1$, this requires that the rates and energies satisfy the relation
\[ \ln \left( \frac{H_{ij}}{H_{ji}} \right) = \epsilon_j - \epsilon_i \].

Whenever we depict an open Markov process with energies as a labelled graph, we will take $\beta = 1$. 
%
%

To show that these form a decorated cospan category we have to show that there is a lax monoidal functor $G \maps (\FinSet_{\epsilon},+)  \to (\Set, \times)$ sending a finite set with energies $(V,\epsilon)$ to the set of Markov processes with energies on $(V,\epsilon)$. Laxness requires the existence of a natural transformation
\[ \Phi_{V,V'} \maps G(V,\epsilon) \times G(V',\epsilon') \to G(V+V',[\epsilon,\epsilon']) \]
taking a pair of Markov processes with energies to Markov process with energies on the disjoint union of the sets of states of the pair of Markov processes with energies, together with a map
\[ \Phi_{1} \maps 1 \to G(\varnothing, ! ) \] 
sending the one element set $1$, i.e. the unit for the tensor product of $(\Set, \times)$ to the unique Markov processes with energies on the empty set.   
 
The key to this construction is the mapping 
\[ \Phi_{V,V'} \maps G(V,\epsilon) \times G(V',\epsilon') \to G(V+V', [\epsilon, \epsilon']) \]
which provides a method of taking two Markov processes with energies, one on $V$ and one on $V'$, and turning them into a Markov process with energies on the disjoint union $V+V'$. This is accomplished via the assignment
\[ \Phi_{V,V'} \maps (V,E,s,t,r,\epsilon) \times (V',E',s',t',r',\epsilon') \mapsto (V+V', E+E', s+s', t+t', [r,r'], [\epsilon, \epsilon'] ).\]
It is easy to see that if $r \maps E \to (0,\infty)$ and $\epsilon \maps V \to \R$ satisfy the detailed balance condition and $r' \maps E' \to (0,\infty)$ and $\epsilon' \maps V' \to \R$ satisfy the detailed balance condition, then $[r,r'] \maps E+E' \to (0,\infty)$ and $[\epsilon,\epsilon'] \maps V+V' \to \R$ satisfy the detailed balance condition. 
 
Thus by Lemma \ref{lemma:fcospans} in \cite{Fong} we arrive at a decorated cospan category of open Markov process with energies.

\begin{defn}
The category $\OpenMark_{\epsilon}$ is the decorated cospan category where an object is a finite set with energies and a morphism is an isomorphism class of open Markov processes with energies $M \maps (X,\chi) \to (Y,\upsilon)$. 
\end{defn}

As is the case for all decorated cospan categories, we have:

\begin{prop} 
The category $\OpenMark_{\epsilon}$ is a dagger compact category.
\end{prop}

Given open Markov processes with energies $M \maps (X,\chi)  \to (Y,\upsilon)$ and $M' \maps (Y, \upsilon) \to (Z, \zeta)$ we can compose them
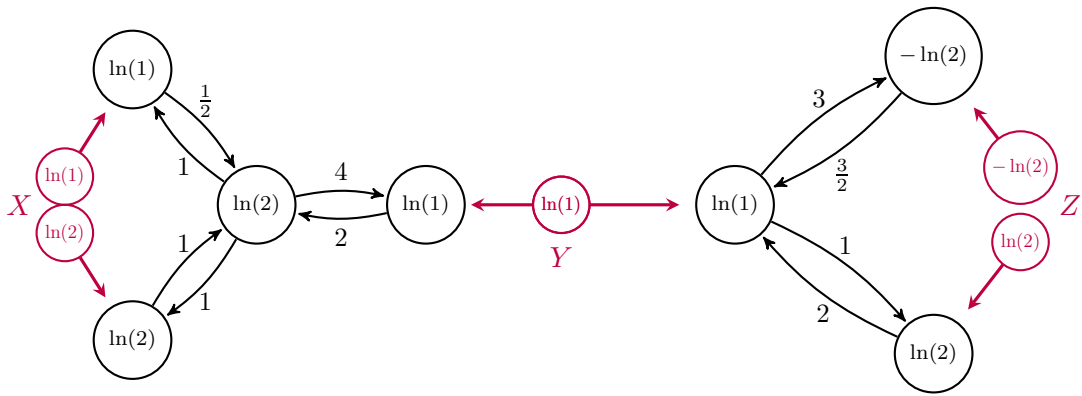
\[
\begin{tikzpicture}[->,>=stealth',shorten >=1pt,thick,scale=1.5]
  \node[main node] (1) at (-8.4,2.2) {\scriptsize$ \ln(1)$};
  \node[main node](2) at (-8.4,-.2) {\scriptsize$\ln(2)$};
  \node[main node](3) at (-7.3,1)  {\scriptsize$\ln(2)$};
  \node[main node](4) at (-5.8,1) {\scriptsize$\ln(1)$};
\node(input)[color=purple] at (-9.4,1) {$X$};
\node[terminal, scale=.5] (A) at (-9,1.25) {$\ln(1)$};
\node[terminal,scale=.5] (B) at (-9,.75) {$\ln(2)$};
\node(output)[color=purple] at (-4.6,.54) {$Y$};
\node[terminal, scale=.5](D) at (-4.6,1) {$\ln(1)$};
  \path[every node/.style={font=\sffamily\small}, shorten >=1pt]
    (3) edge [bend left=12] node[above] {$4$} (4)
    (4) edge [bend left=12] node[below] {$2$} (3)
    (2) edge [bend left=12] node[above] {$1$} (3) 
    (3) edge [bend left=12] node[below] {$1$} (2)
    (1) edge [bend left=12] node[above] {$\frac{1}{2}$}(3) 
    (3) edge [bend left=12] node[below] {$1$} (1);
    
\path[color=purple, very thick, shorten >=2pt, ->, >=stealth] (D) edge (4);
\path[color=purple, very thick, shorten >=4pt, ->, >=stealth] (A) edge (1);
\path[color=purple, very thick, shorten >=4pt, ->, >=stealth]
(B) edge (2);
\begin{scope}[shift={(-4.6,-0.1)},->,>=stealth',shorten >=1pt,thick,scale=1.1]
  \node[main node] (1) at (3,2.2) {\scriptsize$-\ln(2)$};
  \node[main node](2) at (3,-.2) {\scriptsize$\ln(2)$};
  \node[main node](3) at (1.4,1)  {\scriptsize$\ln(1)$};
\node(input)[color=purple] at (-.4,1) {$$};
\node[terminal,scale=.5](F) at (0,1) {$\ln(1)$};
\node(output)[color=purple] at (4.1,1) {$Z$};
\node[terminal,scale=.5](E) at (3.7,1.3) {$-\ln(2)$};
\node[terminal,scale=.5](G) at (3.7,.7) {$\ln(2)$};
  \path[every node/.style={font=\sffamily\small}, shorten >=1pt]
    (2) edge [bend left=12] node[below] {$2$} (3) 
    (3) edge [bend left=12] node[above] {$1$} (2)
    (1) edge [bend left=12] node[below] {$\frac{3}{2}$}(3) 
    (3) edge [bend left=12] node[above] {$3$} (1);
    
\path[color=purple, very thick, shorten >=6pt, ->, >=stealth] (F) edge (3);
\path[color=purple, very thick, shorten >=6pt, ->, >=stealth] (E) edge (1);
\path[color=purple, very thick, shorten >=6pt, ->, >=stealth]
(G) edge (2);
\end{scope}
\end{tikzpicture}
\]
yielding a new open Markov process with energies $MM' \maps X \to Z$
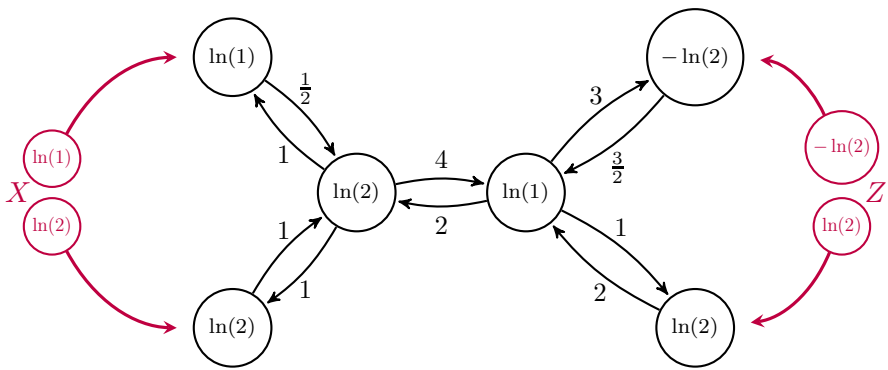
\[
\begin{tikzpicture}[->,>=stealth',shorten >=1pt,thick,scale=1.5]
  \node[main node] (1) at (-5.4,2.2) {\scriptsize $\ln(1)$};
  \node[main node](2) at (-5.4,-.2) {\scriptsize $\ln(2)$};
  \node[main node](3) at (-4.3,1)  {\scriptsize $\ln(2)$};
  \node[main node](4) at (-2.8,1) {\scriptsize $\ln(1)$};
  \node[main node] (5) at (-1.3,2.2) {\scriptsize $-\ln(2)$};
  \node[main node](6) at (-1.3,-.2) {\scriptsize $\ln(2)$};
\node(input)[color=purple] at (-7.3,1) {$X$};
\node[terminal,scale=.5](A) at (-7,1.3) {$\ln(1)$};
\node[terminal, scale=.5] (B) at (-7,.7) {$\ln(2)$};
\node(output)[color=purple] at (.3,1) {$Z$};
\node[terminal, scale=.5](E) at (0,1.4) {$-\ln(2)$};
\node[terminal, scale=.5](G) at (0,.7){$ \ln(2)$ };
  \path[every node/.style={font=\sffamily\small}, shorten >=1pt]
    (3) edge [bend left=12] node[above] {$4$} (4)
    (4) edge [bend left=12] node[below] {$2$} (3)
    (2) edge [bend left=12] node[above] {$1$} (3) 
    (3) edge [bend left=12] node[below] {$1$} (2)
    (1) edge [bend left=12] node[above] {$\frac{1}{2}$}(3) 
    (3) edge [bend left=12] node[below] {$1$} (1)
    (6) edge [bend left=12] node[below] {$2$} (4) 
    (4) edge [bend left=12] node[above] {$1$} (6)
    (5) edge [bend left=12] node[below] {$\frac{3}{2}$}(4) 
    (4) edge [bend left=12] node[above] {$3$} (5);
    
\path[color=purple, very thick, shorten >=6pt, ->, >=stealth, bend left] (A) edge (1);
\path[color=purple, very thick, shorten >=6pt, ->, >=stealth, bend right]
(B) edge (2);
\path[color=purple, very thick, shorten >=6pt, ->, >=stealth, bend right] (E) edge (5);
\path[color=purple, very thick, shorten >=6pt, ->, >=stealth, bend left]
(G) edge (6);
\end{tikzpicture}
\]
Note that in order to compose open Markov processes with energies, their energies must agree on their overlap. This ensures that the category $\DetBalMark$ is closed under composition in the sense that  the rates of any composite process will obey Kolmogorov's criterion.

Let's see an example using open Markov processes to model passive diffusion across a membrane. 

\section{Membrane diffusion as an open Markov process}\label{sec:Membrain}

To illustrate these ideas, we consider a simple model of the diffusion of neutral particles across a membrane as an open detailed balanced Markov process with three states $V=\{A,B,C\}$, input $A$ and output $C$. The states $A$ and $C$ correspond to the each side of the membrane, while $B$ corresponds within the membrane itself
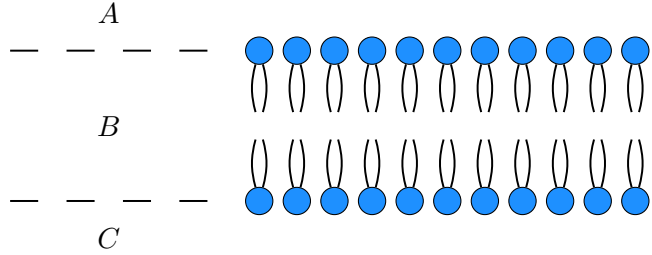
\[ \begin{tikzpicture}
	\begin{pgfonlayer}{nodelayer}
		\node [style=empty] (0) at (-0.5, -0) {};
		\node [style=shadecircle] (1) at (-2, -1) {};
		\node [style=shadecircle] (2) at (2.5, -1) {};
		\node [style=empty] (3) at (0.5, -0) {};
		\node [style=shadecircle] (4) at (-2, 1) {};
		\node [style=shadecircle] (5) at (-0.5, -1) {};
		\node [style=empty] (6) at (-3, 1) {};
		\node [style=empty] (7) at (-3.75, 1) {};
		\node [style=empty] (8) at (1, -0) {};
		\node [style=shadecircle] (10) at (-1.5, -1) {};
		\node [style=shadecircle] (11) at (0.5, -1) {};
		\node [style=shadecircle] (12) at (1.5, 1) {};
		\node [style=empty] (13) at (0, -0) {};
		\node [style=empty] (15) at (-3.75, -1) {};
		\node [style=shadecircle] (16) at (0.5, 1) {};
		\node [style=shadecircle] (17) at (-2.5, 1) {};
		\node [style=empty] (18) at (2, -0) {};
		\node [style=empty] (19) at (-5.25, 1) {};
		\node [style=shadecircle] (20) at (0, 1) {};
		\node [style=empty] (21) at (-4.5, 1) {};
		\node [style=empty] (22) at (-5.25, -1) {};
		\node [style=shadecircle] (23) at (-1.5, 1) {};
		\node [style=shadecircle] (24) at (1.5, -1) {};
		\node [style=shadecircle] (25) at (2, 1) {};
		\node [style=shadecircle] (26) at (1, 1) {};
		\node [style=shadecircle] (27) at (1, -1) {};
		\node [style=empty] (28) at (-3, -1) {};
		\node [style=empty] (29) at (1.5, -0) {};
		\node [style=shadecircle] (30) at (2.5, 1) {};
		\node [style=empty] (31) at (-2.5, -0) {};
		\node [style=empty] (32) at (-4.5, -1) {};
		\node [style=empty] (33) at (-1, -0) {};
		\node [style=shadecircle] (34) at (-1, 1) {};
		\node [style=shadecircle] (35) at (-0.5, 1) {};
		\node [style=empty] (36) at (-6, 1) {};
		\node [style=empty] (37) at (2.5, -0) {};
		\node [style=shadecircle] (38) at (0, -1) {};
		\node [style=empty] (39) at (-1.5, -0) {};
		\node [style=empty] (40) at (-6, -1) {};
		\node [style=shadecircle] (41) at (2, -1) {};
		\node [style=shadecircle] (42) at (-1, -1) {};
		\node [style=shadecircle] (44) at (-2.5, -1) {};
		\node [style=empty] (45) at (-2, -0) {};
		\node [style=empty] (9) at (-4.5, 1.5) {$A$};
		\node [style=empty] (14) at (-4.5, -0) {$B$};
		\node [style=empty] (43) at (-4.5, -1.5) {$C$};
	\end{pgfonlayer}
	\begin{pgfonlayer}{edgelayer}
		\draw [style=lipid, bend right=15, looseness=1.00] (17) to (31);
		\draw [style=lipid, bend left=15, looseness=1.00] (17) to (31);
		\draw [style=lipid, bend left=15, looseness=1.00] (44) to (31);
		\draw [style=lipid, bend right=15, looseness=1.00] (44) to (31);
		\draw [style=lipid, bend right=15, looseness=1.00] (4) to (45);
		\draw [style=lipid, bend left=15, looseness=1.00] (4) to (45);
		\draw [style=lipid, bend left=15, looseness=1.00] (1) to (45);
		\draw [style=lipid, bend right=15, looseness=1.00] (1) to (45);
		\draw [style=lipid, bend right=15, looseness=1.00] (23) to (39);
		\draw [style=lipid, bend left=15, looseness=1.00] (23) to (39);
		\draw [style=lipid, bend left=15, looseness=1.00] (10) to (39);
		\draw [style=lipid, bend right=15, looseness=1.00] (10) to (39);
		\draw [style=lipid, bend right=15, looseness=1.00] (34) to (33);
		\draw [style=lipid, bend left=15, looseness=1.00] (34) to (33);
		\draw [style=lipid, bend left=15, looseness=1.00] (42) to (33);
		\draw [style=lipid, bend right=15, looseness=1.00] (42) to (33);
		\draw [style=lipid, bend right=15, looseness=1.00] (35) to (0);
		\draw [style=lipid, bend left=15, looseness=1.00] (35) to (0);
		\draw [style=lipid, bend left=15, looseness=1.00] (5) to (0);
		\draw [style=lipid, bend right=15, looseness=1.00] (5) to (0);
		\draw [style=lipid, bend right=15, looseness=1.00] (20) to (13);
		\draw [style=lipid, bend left=15, looseness=1.00] (20) to (13);
		\draw [style=lipid, bend left=15, looseness=1.00] (38) to (13);
		\draw [style=lipid, bend right=15, looseness=1.00] (38) to (13);
		\draw [style=lipid, bend right=15, looseness=1.00] (16) to (3);
		\draw [style=lipid, bend left=15, looseness=1.00] (16) to (3);
		\draw [style=lipid, bend left=15, looseness=1.00] (11) to (3);
		\draw [style=lipid, bend right=15, looseness=1.00] (11) to (3);
		\draw [style=lipid, bend right=15, looseness=1.00] (26) to (8);
		\draw [style=lipid, bend left=15, looseness=1.00] (26) to (8);
		\draw [style=lipid, bend left=15, looseness=1.00] (27) to (8);
		\draw [style=lipid, bend right=15, looseness=1.00] (27) to (8);
		\draw [style=lipid, bend right=15, looseness=1.00] (12) to (29);
		\draw [style=lipid, bend left=15, looseness=1.00] (12) to (29);
		\draw [style=lipid, bend left=15, looseness=1.00] (24) to (29);
		\draw [style=lipid, bend right=15, looseness=1.00] (24) to (29);
		\draw [style=lipid, bend right=15, looseness=1.00] (25) to (18);
		\draw [style=lipid, bend left=15, looseness=1.00] (25) to (18);
		\draw [style=lipid, bend left=15, looseness=1.00] (41) to (18);
		\draw [style=lipid, bend right=15, looseness=1.00] (41) to (18);
		\draw [style=lipid, bend right=15, looseness=1.00] (30) to (37);
		\draw [style=lipid, bend left=15, looseness=1.00] (30) to (37);
		\draw [style=lipid, bend left=15, looseness=1.00] (2) to (37);
		\draw [style=lipid, bend right=15, looseness=1.00] (2) to (37);
		\draw [style=lipid] (6) to (7);
		\draw [style=lipid] (7) to (21);
		\draw [style=lipid] (21) to (19);
		\draw [style=lipid] (19) to (36);
		\draw [style=lipid] (28) to (15);
		\draw [style=lipid] (15) to (32);
		\draw [style=lipid] (32) to (22);
		\draw [style=lipid] (22) to (40);
	\end{pgfonlayer}
\end{tikzpicture}
\]
In this model, $p_A$ is the number of particles on one side of the membrane, $p_B$ the number of particles within the membrane and $p_C$ the number of particles on the other side of the membrane. The off-diagonal entires in the Hamiltonian $H_{ij}, i \neq j$ are the rates at which probability flows from $j$ to $i$. For example $H_{AB}$ is the rate at which probability flows from $B$ to $A$, or from inside the membrane to the top of the membrane. Let us assume that the membrane is symmetric in the sense that the rate at which particles hop from outside of the membrane to the interior is the same on either side, i.e. $H_{BA} = H_{BC} = H_{\text{in}}$ and $H_{AB} = H_{CB} = H_{\text{out}}$. We can draw such an open Markov process as a labeled graph:
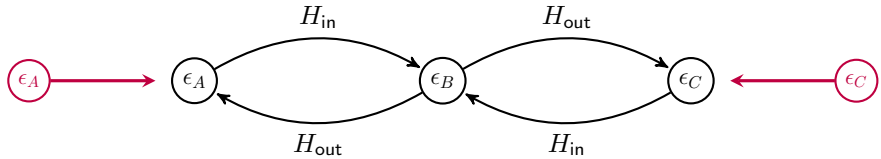
\[
\begin{tikzpicture}[->,>=stealth',shorten >=1pt,thick,scale=1.1]
\node[main node, scale=.65](A) at (-2,0) {$\epsilon_A$};
\node[main node, scale=.65](B) at (1,0) {$\epsilon_B$};
\node[main node, scale=.65](C) at (4,0) {$\epsilon_C$};
\node[terminal,scale=.6](X) at (-4,0) {$\epsilon_A$};
\node[terminal, scale=.6](Y) at (6,0) {$\epsilon_C$};

  \path[every node/.style={font=\sffamily\small}, shorten >=1pt]
    (A) edge [bend left] node[above] {$H_{\text{in}}$} (B)
    (B) edge [bend left] node[below] {$H_{\text{out}}$} (A) 
    (B) edge [bend left] node[above] {$H_{\text{out}}$} (C)
    (C) edge [bend left] node[below] {$H_{\text{in}}$} (B);

\path[color=purple, very thick, shorten >=6pt, ->, >=stealth] (X) edge (A);
\path[color=purple, very thick, shorten >=6pt, ->, >=stealth] (Y) edge (C);

\end{tikzpicture}
\]
The labels on the edges are the corresponding transition rates. The states are labeled by their energies (defined up to an additive constant), which determine their equilibrium probabilities (up to a multiplicative scaling) which up to an overall scaling, are given by $q_A = q_C = H_{\text{in}} H_{\text{out}}$ and $q_B = H_{\text{in}}^2$.
Suppose the values of $p_A$ and $p_C$ are externally maintained at constant values, i.e. whenever a particle diffuses from outside the cell into the membrane, the environment around the cell provides another particle and similarly when particles move from inside the membrane to the outside. We call $(p_A,p_C)$ the \define{boundary probabilities}. Given the values of $p_A$ and $p_C$, the steady state probability $p_B$ compatible with these values is
\[ p_B = \frac{H_{\text{in}}p_A + H_{\text{in}}p_C}{-H_{BB} } = \frac{H_{\text{in}}}{H_{\text{out}}}\frac{p_A + p_C}{2}. \]
In Section \ref{sec:dissipation} we show that this steady state probability minimizes the dissipation, subject to the constraints on $p_A$ and $p_C$.

We thus have a non-equilibrium steady state $p = (p_A, p_B, p_C)$ with $p_B$ given in terms of the boundary probabilities above. From these values we can compute the boundary flows, $J_A, J_C$ as 
\[ J_A = \sum_j J_{Aj}(p) = H_{\text{out}}p_B -H_{\text{in}}p_A \]
and 
\[ J_C = \sum_j J_{Cj}(p) = H_{\text{out}}p_B - H_{\text{in}}p_C. \]
Written in terms of the boundary probabilities this gives
\[ J_A = \frac{ H_{\text{in}}(p_C-p_A)}{2} \]
and
\[ J_C = \frac{H_{\text{in}}(p_A-p_C)}{2}. \]
Note that $J_A = - J_C$ implying that there is a constant net flow through the open Markov process. As one would expect, if $p_A > p_C$ there is a positive flow from $A$ to $C$ and vice-versa. Of course, in actual membranes there exist much more complex transport mechanisms than the simple diffusion model presented here. A number of authors have modeled more complicated transport phenomena using the framework of networked master equation systems \cite{OPKBio, SchnakenBook}.

In our framework, we call the collection of all boundary probability-flows pairs the steady state `behavior' of the open Markov process. The main theorem of \cite{BaezFongP} constructs a functor from the category of open detailed balanced Markov process to the category of linear relations. Applied to an open detailed balanced Markov process, this functor yields the set of allowed steady state boundary probability-flow pairs. One can imagine a situation in which only the probabilities and flows of boundary states are observable, thus characterizing a process in terms of its behavior. This provides an effective `black-boxing' of open detailed balanced Markov processes. 

As morphisms in a category, open detailed balanced Markov processes can be composed, thereby building up more complex processes from these open building blocks. The fact that `black-boxing'  is accomplished via a functor means that the behavior of a composite Markov process can be built up from the composite behaviors of the open Markov processes from which it is built. 

Suppose we had another such membrane 
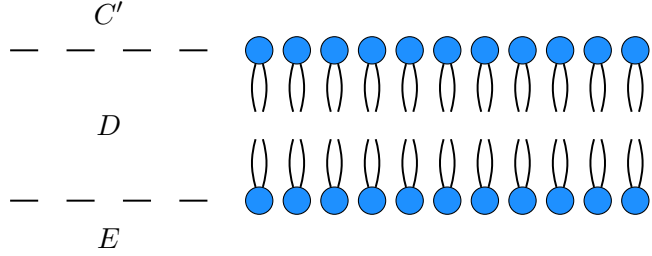
\[ \begin{tikzpicture}
	\begin{pgfonlayer}{nodelayer}
		\node [style=empty] (0) at (-0.5, -0) {};
		\node [style=shadecircle] (1) at (-2, -1) {};
		\node [style=shadecircle] (2) at (2.5, -1) {};
		\node [style=empty] (3) at (0.5, -0) {};
		\node [style=shadecircle] (4) at (-2, 1) {};
		\node [style=shadecircle] (5) at (-0.5, -1) {};
		\node [style=empty] (6) at (-3, 1) {};
		\node [style=empty] (7) at (-3.75, 1) {};
		\node [style=empty] (8) at (1, -0) {};
		\node [style=shadecircle] (10) at (-1.5, -1) {};
		\node [style=shadecircle] (11) at (0.5, -1) {};
		\node [style=shadecircle] (12) at (1.5, 1) {};
		\node [style=empty] (13) at (0, -0) {};
		\node [style=empty] (15) at (-3.75, -1) {};
		\node [style=shadecircle] (16) at (0.5, 1) {};
		\node [style=shadecircle] (17) at (-2.5, 1) {};
		\node [style=empty] (18) at (2, -0) {};
		\node [style=empty] (19) at (-5.25, 1) {};
		\node [style=shadecircle] (20) at (0, 1) {};
		\node [style=empty] (21) at (-4.5, 1) {};
		\node [style=empty] (22) at (-5.25, -1) {};
		\node [style=shadecircle] (23) at (-1.5, 1) {};
		\node [style=shadecircle] (24) at (1.5, -1) {};
		\node [style=shadecircle] (25) at (2, 1) {};
		\node [style=shadecircle] (26) at (1, 1) {};
		\node [style=shadecircle] (27) at (1, -1) {};
		\node [style=empty] (28) at (-3, -1) {};
		\node [style=empty] (29) at (1.5, -0) {};
		\node [style=shadecircle] (30) at (2.5, 1) {};
		\node [style=empty] (31) at (-2.5, -0) {};
		\node [style=empty] (32) at (-4.5, -1) {};
		\node [style=empty] (33) at (-1, -0) {};
		\node [style=shadecircle] (34) at (-1, 1) {};
		\node [style=shadecircle] (35) at (-0.5, 1) {};
		\node [style=empty] (36) at (-6, 1) {};
		\node [style=empty] (37) at (2.5, -0) {};
		\node [style=shadecircle] (38) at (0, -1) {};
		\node [style=empty] (39) at (-1.5, -0) {};
		\node [style=empty] (40) at (-6, -1) {};
		\node [style=shadecircle] (41) at (2, -1) {};
		\node [style=shadecircle] (42) at (-1, -1) {};
		\node [style=shadecircle] (44) at (-2.5, -1) {};
		\node [style=empty] (45) at (-2, -0) {};
		\node [style=empty] (9) at (-4.5, 1.5) {$C'$};
		\node [style=empty] (14) at (-4.5, -0) {$D$};
		\node [style=empty] (43) at (-4.5, -1.5) {$E$};
	\end{pgfonlayer}
	\begin{pgfonlayer}{edgelayer}
		\draw [style=lipid, bend right=15, looseness=1.00] (17) to (31);
		\draw [style=lipid, bend left=15, looseness=1.00] (17) to (31);
		\draw [style=lipid, bend left=15, looseness=1.00] (44) to (31);
		\draw [style=lipid, bend right=15, looseness=1.00] (44) to (31);
		\draw [style=lipid, bend right=15, looseness=1.00] (4) to (45);
		\draw [style=lipid, bend left=15, looseness=1.00] (4) to (45);
		\draw [style=lipid, bend left=15, looseness=1.00] (1) to (45);
		\draw [style=lipid, bend right=15, looseness=1.00] (1) to (45);
		\draw [style=lipid, bend right=15, looseness=1.00] (23) to (39);
		\draw [style=lipid, bend left=15, looseness=1.00] (23) to (39);
		\draw [style=lipid, bend left=15, looseness=1.00] (10) to (39);
		\draw [style=lipid, bend right=15, looseness=1.00] (10) to (39);
		\draw [style=lipid, bend right=15, looseness=1.00] (34) to (33);
		\draw [style=lipid, bend left=15, looseness=1.00] (34) to (33);
		\draw [style=lipid, bend left=15, looseness=1.00] (42) to (33);
		\draw [style=lipid, bend right=15, looseness=1.00] (42) to (33);
		\draw [style=lipid, bend right=15, looseness=1.00] (35) to (0);
		\draw [style=lipid, bend left=15, looseness=1.00] (35) to (0);
		\draw [style=lipid, bend left=15, looseness=1.00] (5) to (0);
		\draw [style=lipid, bend right=15, looseness=1.00] (5) to (0);
		\draw [style=lipid, bend right=15, looseness=1.00] (20) to (13);
		\draw [style=lipid, bend left=15, looseness=1.00] (20) to (13);
		\draw [style=lipid, bend left=15, looseness=1.00] (38) to (13);
		\draw [style=lipid, bend right=15, looseness=1.00] (38) to (13);
		\draw [style=lipid, bend right=15, looseness=1.00] (16) to (3);
		\draw [style=lipid, bend left=15, looseness=1.00] (16) to (3);
		\draw [style=lipid, bend left=15, looseness=1.00] (11) to (3);
		\draw [style=lipid, bend right=15, looseness=1.00] (11) to (3);
		\draw [style=lipid, bend right=15, looseness=1.00] (26) to (8);
		\draw [style=lipid, bend left=15, looseness=1.00] (26) to (8);
		\draw [style=lipid, bend left=15, looseness=1.00] (27) to (8);
		\draw [style=lipid, bend right=15, looseness=1.00] (27) to (8);
		\draw [style=lipid, bend right=15, looseness=1.00] (12) to (29);
		\draw [style=lipid, bend left=15, looseness=1.00] (12) to (29);
		\draw [style=lipid, bend left=15, looseness=1.00] (24) to (29);
		\draw [style=lipid, bend right=15, looseness=1.00] (24) to (29);
		\draw [style=lipid, bend right=15, looseness=1.00] (25) to (18);
		\draw [style=lipid, bend left=15, looseness=1.00] (25) to (18);
		\draw [style=lipid, bend left=15, looseness=1.00] (41) to (18);
		\draw [style=lipid, bend right=15, looseness=1.00] (41) to (18);
		\draw [style=lipid, bend right=15, looseness=1.00] (30) to (37);
		\draw [style=lipid, bend left=15, looseness=1.00] (30) to (37);
		\draw [style=lipid, bend left=15, looseness=1.00] (2) to (37);
		\draw [style=lipid, bend right=15, looseness=1.00] (2) to (37);
		\draw [style=lipid] (6) to (7);
		\draw [style=lipid] (7) to (21);
		\draw [style=lipid] (21) to (19);
		\draw [style=lipid] (19) to (36);
		\draw [style=lipid] (28) to (15);
		\draw [style=lipid] (15) to (32);
		\draw [style=lipid] (32) to (22);
		\draw [style=lipid] (22) to (40);
			\end{pgfonlayer}
\end{tikzpicture}
\]
This is a morphism in $\DetBalMark$ from with input $Y=\{C'\}$ and output $Z=\{E\}$. Two open detailed balanced Markov processes can be composed if the detailed balanced equilibrium probabilities at the outputs of one match the detailed balanced equilibrium probabilities at the inputs of the other. This requirement guarantees that the composite of two open detailed balanced Markov process still admits a detailed balanced equilibrium.
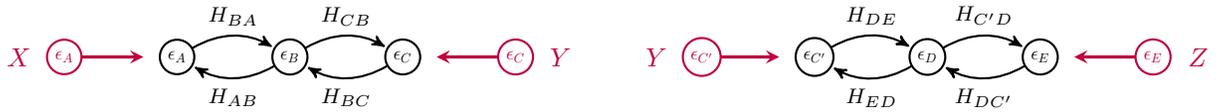
\[
\begin{tikzpicture}[->,>=stealth',shorten >=1pt,thick,scale=.75, font=\small]
\node[main node, scale=.5](A) at (-2,0) {$\epsilon_A$};
\node[main node, scale=.5](B) at (0,0) {$\epsilon_B$};
\node[main node, scale=.5](C) at (2,0) {$\epsilon_C$};
\node(input)[color=purple] at (-4.8,0) {$X$};
\node(input)[color=purple] at (4.8,0) {$Y$};
\node[terminal,scale=.5](X) at (-4,0) {$\epsilon_A$};
\node[terminal, scale=.5](Y) at (4,0) {$\epsilon_C$};

  \path[every node/.style={font=\scriptsize}, shorten >=1pt]
    (A) edge [bend left] node[above] {$H_{BA}$} (B)
    (B) edge [bend left] node[below] {$H_{AB}$} (A) 
    (B) edge [bend left] node[above] {$H_{CB}$} (C)
    (C) edge [bend left] node[below] {$H_{BC}$} (B);

\path[color=purple, very thick, shorten >=6pt, ->, >=stealth] (X) edge (A);
\path[color=purple, very thick, shorten >=6pt, ->, >=stealth] (Y) edge (C);

\node[main node, scale=.5](C') at (9.3,0) {$\epsilon_{C'}$};
\node[main node, scale=.5](B') at (11.3,0) {$\epsilon_D$};
\node[main node, scale=.5](A') at (13.3,0) {$\epsilon_{E}$};
\node(input)[color=purple] at (6.5,0) {$Y$};
\node(input)[color=purple] at (16.1,0) {$Z$};
\node[terminal,scale=.5](X') at (7.3,0) {$\epsilon_{C'}$};
\node[terminal, scale=.5](Y') at (15.3,0) {$\epsilon_{E}$};

  \path[every node/.style={font=\scriptsize}, shorten >=1pt]
    (C') edge [bend left] node[above] {$H_{DE}$} (B')
    (B') edge [bend left] node[below] {$H_{ED}$} (C') 
    (B') edge [bend left] node[above] {$H_{C'D}$} (A')
    (A') edge [bend left] node[below] {$H_{DC'}$} (B');

\path[color=purple, very thick, shorten >=6pt, ->, >=stealth] (X') edge (C');
\path[color=purple, very thick, shorten >=6pt, ->, >=stealth] (Y') edge (A');

\end{tikzpicture}
\]

If $q_C=q_{C'}$ in our two membrane models we can compose them by identifying $C$ with $C'$ to yield an open detailed balanced Markov process modeling the diffusion of neutral particles across membranes arranged in series: 
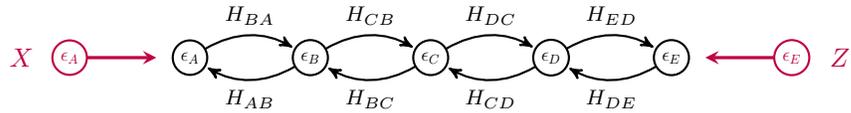
\[
\begin{tikzpicture}[->,>=stealth',shorten >=1pt,thick,scale=.8, font=\small]
\node[main node, scale=.5](A) at (-2,0) {$\epsilon_A$};
\node[main node, scale=.5](B) at (0,0) {$\epsilon_B$};
\node[main node, scale=.5](C) at (2,0) {$\epsilon_C$};
\node(input)[color=purple] at (-4.8,0) {$X$};
\node[terminal,scale=.5](X) at (-4,0) {$\epsilon_A$};

  \path[every node/.style={font=\scriptsize}, shorten >=1pt]
    (A) edge [bend left] node[above] {$H_{BA}$} (B)
    (B) edge [bend left] node[below] {$H_{AB}$} (A) 
    (B) edge [bend left] node[above] {$H_{CB}$} (C)
    (C) edge [bend left] node[below] {$H_{BC}$} (B);

\path[color=purple, very thick, shorten >=6pt, ->, >=stealth] (X) edge (A);

\node[main node, scale=.5](D) at (4,0) {$\epsilon_D$};
\node[main node, scale=.5](E) at (6,0) {$\epsilon_{E}$};
\node(input)[color=purple] at (8.8,0) {$Z$};
\node[terminal, scale=.5](Y') at (8,0) {$\epsilon_{E}$};

  \path[every node/.style={font=\scriptsize}, shorten >=1pt]
    (C) edge [bend left] node[above] {$H_{DC}$} (D)
    (D) edge [bend left] node[below] {$H_{CD}$} (C) 
    (D) edge [bend left] node[above] {$H_{ED}$} (E)
    (E) edge [bend left] node[below] {$H_{DE}$} (D);

\path[color=purple, very thick, shorten >=6pt, ->, >=stealth] (Y') edge (E);

\end{tikzpicture}
\]

Notice that the states corresponding to $C$ and $C'$ in each process have been identified and become internal states in the composite which is a morphism from $X=\{A\}$ to $Z=\{E\}$. This open Markov process can be thought of as modeling the diffusion across two membranes in series
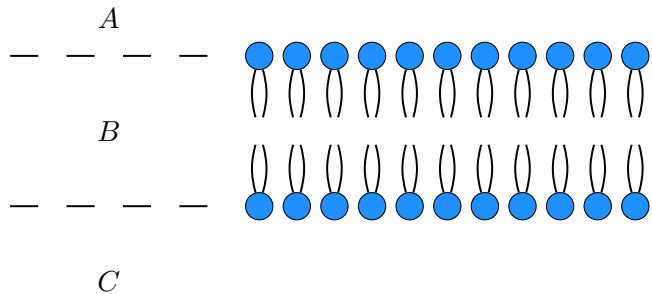
\[ \begin{tikzpicture}
	\begin{pgfonlayer}{nodelayer}
		\node [style=empty] (0) at (-0.5, -0) {};
		\node [style=shadecircle] (1) at (-2, -1) {};
		\node [style=shadecircle] (2) at (2.5, -1) {};
		\node [style=empty] (3) at (0.5, -0) {};
		\node [style=shadecircle] (4) at (-2, 1) {};
		\node [style=shadecircle] (5) at (-0.5, -1) {};
		\node [style=empty] (6) at (-3, 1) {};
		\node [style=empty] (7) at (-3.75, 1) {};
		\node [style=empty] (8) at (1, -0) {};
		\node [style=shadecircle] (10) at (-1.5, -1) {};
		\node [style=shadecircle] (11) at (0.5, -1) {};
		\node [style=shadecircle] (12) at (1.5, 1) {};
		\node [style=empty] (13) at (0, -0) {};
		\node [style=empty] (15) at (-3.75, -1) {};
		\node [style=shadecircle] (16) at (0.5, 1) {};
		\node [style=shadecircle] (17) at (-2.5, 1) {};
		\node [style=empty] (18) at (2, -0) {};
		\node [style=empty] (19) at (-5.25, 1) {};
		\node [style=shadecircle] (20) at (0, 1) {};
		\node [style=empty] (21) at (-4.5, 1) {};
		\node [style=empty] (22) at (-5.25, -1) {};
		\node [style=shadecircle] (23) at (-1.5, 1) {};
		\node [style=shadecircle] (24) at (1.5, -1) {};
		\node [style=shadecircle] (25) at (2, 1) {};
		\node [style=shadecircle] (26) at (1, 1) {};
		\node [style=shadecircle] (27) at (1, -1) {};
		\node [style=empty] (28) at (-3, -1) {};
		\node [style=empty] (29) at (1.5, -0) {};
		\node [style=shadecircle] (30) at (2.5, 1) {};
		\node [style=empty] (31) at (-2.5, -0) {};
		\node [style=empty] (32) at (-4.5, -1) {};
		\node [style=empty] (33) at (-1, -0) {};
		\node [style=shadecircle] (34) at (-1, 1) {};
		\node [style=shadecircle] (35) at (-0.5, 1) {};
		\node [style=empty] (36) at (-6, 1) {};
		\node [style=empty] (37) at (2.5, -0) {};
		\node [style=shadecircle] (38) at (0, -1) {};
		\node [style=empty] (39) at (-1.5, -0) {};
		\node [style=empty] (40) at (-6, -1) {};
		\node [style=shadecircle] (41) at (2, -1) {};
		\node [style=shadecircle] (42) at (-1, -1) {};
		\node [style=shadecircle] (44) at (-2.5, -1) {};
		\node [style=empty] (45) at (-2, -0) {};
		\node [style=empty] (9) at (-4.5, 1.5) {$A$};
		\node [style=empty] (14) at (-4.5, -0) {$B$};
		\node [style=empty] (43) at (-4.5, -2) {$C$};
	\end{pgfonlayer}
	\begin{pgfonlayer}{edgelayer}
		\draw [style=lipid, bend right=15, looseness=1.00] (17) to (31);
		\draw [style=lipid, bend left=15, looseness=1.00] (17) to (31);
		\draw [style=lipid, bend left=15, looseness=1.00] (44) to (31);
		\draw [style=lipid, bend right=15, looseness=1.00] (44) to (31);
		\draw [style=lipid, bend right=15, looseness=1.00] (4) to (45);
		\draw [style=lipid, bend left=15, looseness=1.00] (4) to (45);
		\draw [style=lipid, bend left=15, looseness=1.00] (1) to (45);
		\draw [style=lipid, bend right=15, looseness=1.00] (1) to (45);
		\draw [style=lipid, bend right=15, looseness=1.00] (23) to (39);
		\draw [style=lipid, bend left=15, looseness=1.00] (23) to (39);
		\draw [style=lipid, bend left=15, looseness=1.00] (10) to (39);
		\draw [style=lipid, bend right=15, looseness=1.00] (10) to (39);
		\draw [style=lipid, bend right=15, looseness=1.00] (34) to (33);
		\draw [style=lipid, bend left=15, looseness=1.00] (34) to (33);
		\draw [style=lipid, bend left=15, looseness=1.00] (42) to (33);
		\draw [style=lipid, bend right=15, looseness=1.00] (42) to (33);
		\draw [style=lipid, bend right=15, looseness=1.00] (35) to (0);
		\draw [style=lipid, bend left=15, looseness=1.00] (35) to (0);
		\draw [style=lipid, bend left=15, looseness=1.00] (5) to (0);
		\draw [style=lipid, bend right=15, looseness=1.00] (5) to (0);
		\draw [style=lipid, bend right=15, looseness=1.00] (20) to (13);
		\draw [style=lipid, bend left=15, looseness=1.00] (20) to (13);
		\draw [style=lipid, bend left=15, looseness=1.00] (38) to (13);
		\draw [style=lipid, bend right=15, looseness=1.00] (38) to (13);
		\draw [style=lipid, bend right=15, looseness=1.00] (16) to (3);
		\draw [style=lipid, bend left=15, looseness=1.00] (16) to (3);
		\draw [style=lipid, bend left=15, looseness=1.00] (11) to (3);
		\draw [style=lipid, bend right=15, looseness=1.00] (11) to (3);
		\draw [style=lipid, bend right=15, looseness=1.00] (26) to (8);
		\draw [style=lipid, bend left=15, looseness=1.00] (26) to (8);
		\draw [style=lipid, bend left=15, looseness=1.00] (27) to (8);
		\draw [style=lipid, bend right=15, looseness=1.00] (27) to (8);
		\draw [style=lipid, bend right=15, looseness=1.00] (12) to (29);
		\draw [style=lipid, bend left=15, looseness=1.00] (12) to (29);
		\draw [style=lipid, bend left=15, looseness=1.00] (24) to (29);
		\draw [style=lipid, bend right=15, looseness=1.00] (24) to (29);
		\draw [style=lipid, bend right=15, looseness=1.00] (25) to (18);
		\draw [style=lipid, bend left=15, looseness=1.00] (25) to (18);
		\draw [style=lipid, bend left=15, looseness=1.00] (41) to (18);
		\draw [style=lipid, bend right=15, looseness=1.00] (41) to (18);
		\draw [style=lipid, bend right=15, looseness=1.00] (30) to (37);
		\draw [style=lipid, bend left=15, looseness=1.00] (30) to (37);
		\draw [style=lipid, bend left=15, looseness=1.00] (2) to (37);
		\draw [style=lipid, bend right=15, looseness=1.00] (2) to (37);
		\draw [style=lipid] (6) to (7);
		\draw [style=lipid] (7) to (21);
		\draw [style=lipid] (21) to (19);
		\draw [style=lipid] (19) to (36);
		\draw [style=lipid] (28) to (15);
		\draw [style=lipid] (15) to (32);
		\draw [style=lipid] (32) to (22);
		\draw [style=lipid] (22) to (40);
			\end{pgfonlayer}
\end{tikzpicture}
\]
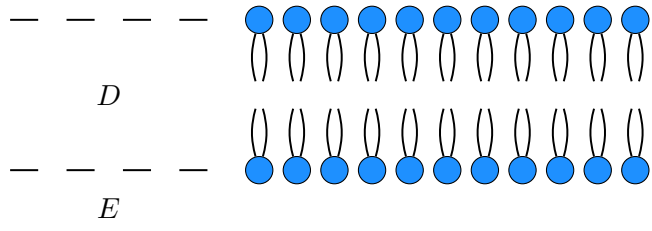
\[ \begin{tikzpicture}
	\begin{pgfonlayer}{nodelayer}
		\node [style=empty] (0) at (-0.5, -0) {};
		\node [style=shadecircle] (1) at (-2, -1) {};
		\node [style=shadecircle] (2) at (2.5, -1) {};
		\node [style=empty] (3) at (0.5, -0) {};
		\node [style=shadecircle] (4) at (-2, 1) {};
		\node [style=shadecircle] (5) at (-0.5, -1) {};
		\node [style=empty] (6) at (-3, 1) {};
		\node [style=empty] (7) at (-3.75, 1) {};
		\node [style=empty] (8) at (1, -0) {};
		\node [style=shadecircle] (10) at (-1.5, -1) {};
		\node [style=shadecircle] (11) at (0.5, -1) {};
		\node [style=shadecircle] (12) at (1.5, 1) {};
		\node [style=empty] (13) at (0, -0) {};
		\node [style=empty] (15) at (-3.75, -1) {};
		\node [style=shadecircle] (16) at (0.5, 1) {};
		\node [style=shadecircle] (17) at (-2.5, 1) {};
		\node [style=empty] (18) at (2, -0) {};
		\node [style=empty] (19) at (-5.25, 1) {};
		\node [style=shadecircle] (20) at (0, 1) {};
		\node [style=empty] (21) at (-4.5, 1) {};
		\node [style=empty] (22) at (-5.25, -1) {};
		\node [style=shadecircle] (23) at (-1.5, 1) {};
		\node [style=shadecircle] (24) at (1.5, -1) {};
		\node [style=shadecircle] (25) at (2, 1) {};
		\node [style=shadecircle] (26) at (1, 1) {};
		\node [style=shadecircle] (27) at (1, -1) {};
		\node [style=empty] (28) at (-3, -1) {};
		\node [style=empty] (29) at (1.5, -0) {};
		\node [style=shadecircle] (30) at (2.5, 1) {};
		\node [style=empty] (31) at (-2.5, -0) {};
		\node [style=empty] (32) at (-4.5, -1) {};
		\node [style=empty] (33) at (-1, -0) {};
		\node [style=shadecircle] (34) at (-1, 1) {};
		\node [style=shadecircle] (35) at (-0.5, 1) {};
		\node [style=empty] (36) at (-6, 1) {};
		\node [style=empty] (37) at (2.5, -0) {};
		\node [style=shadecircle] (38) at (0, -1) {};
		\node [style=empty] (39) at (-1.5, -0) {};
		\node [style=empty] (40) at (-6, -1) {};
		\node [style=shadecircle] (41) at (2, -1) {};
		\node [style=shadecircle] (42) at (-1, -1) {};
		\node [style=shadecircle] (44) at (-2.5, -1) {};
		\node [style=empty] (45) at (-2, -0) {};

		\node [style=empty] (14) at (-4.5, -0) {$D$};
		\node [style=empty] (43) at (-4.5, -1.5) {$E$};
	\end{pgfonlayer}
	\begin{pgfonlayer}{edgelayer}
		\draw [style=lipid, bend right=15, looseness=1.00] (17) to (31);
		\draw [style=lipid, bend left=15, looseness=1.00] (17) to (31);
		\draw [style=lipid, bend left=15, looseness=1.00] (44) to (31);
		\draw [style=lipid, bend right=15, looseness=1.00] (44) to (31);
		\draw [style=lipid, bend right=15, looseness=1.00] (4) to (45);
		\draw [style=lipid, bend left=15, looseness=1.00] (4) to (45);
		\draw [style=lipid, bend left=15, looseness=1.00] (1) to (45);
		\draw [style=lipid, bend right=15, looseness=1.00] (1) to (45);
		\draw [style=lipid, bend right=15, looseness=1.00] (23) to (39);
		\draw [style=lipid, bend left=15, looseness=1.00] (23) to (39);
		\draw [style=lipid, bend left=15, looseness=1.00] (10) to (39);
		\draw [style=lipid, bend right=15, looseness=1.00] (10) to (39);
		\draw [style=lipid, bend right=15, looseness=1.00] (34) to (33);
		\draw [style=lipid, bend left=15, looseness=1.00] (34) to (33);
		\draw [style=lipid, bend left=15, looseness=1.00] (42) to (33);
		\draw [style=lipid, bend right=15, looseness=1.00] (42) to (33);
		\draw [style=lipid, bend right=15, looseness=1.00] (35) to (0);
		\draw [style=lipid, bend left=15, looseness=1.00] (35) to (0);
		\draw [style=lipid, bend left=15, looseness=1.00] (5) to (0);
		\draw [style=lipid, bend right=15, looseness=1.00] (5) to (0);
		\draw [style=lipid, bend right=15, looseness=1.00] (20) to (13);
		\draw [style=lipid, bend left=15, looseness=1.00] (20) to (13);
		\draw [style=lipid, bend left=15, looseness=1.00] (38) to (13);
		\draw [style=lipid, bend right=15, looseness=1.00] (38) to (13);
		\draw [style=lipid, bend right=15, looseness=1.00] (16) to (3);
		\draw [style=lipid, bend left=15, looseness=1.00] (16) to (3);
		\draw [style=lipid, bend left=15, looseness=1.00] (11) to (3);
		\draw [style=lipid, bend right=15, looseness=1.00] (11) to (3);
		\draw [style=lipid, bend right=15, looseness=1.00] (26) to (8);
		\draw [style=lipid, bend left=15, looseness=1.00] (26) to (8);
		\draw [style=lipid, bend left=15, looseness=1.00] (27) to (8);
		\draw [style=lipid, bend right=15, looseness=1.00] (27) to (8);
		\draw [style=lipid, bend right=15, looseness=1.00] (12) to (29);
		\draw [style=lipid, bend left=15, looseness=1.00] (12) to (29);
		\draw [style=lipid, bend left=15, looseness=1.00] (24) to (29);
		\draw [style=lipid, bend right=15, looseness=1.00] (24) to (29);
		\draw [style=lipid, bend right=15, looseness=1.00] (25) to (18);
		\draw [style=lipid, bend left=15, looseness=1.00] (25) to (18);
		\draw [style=lipid, bend left=15, looseness=1.00] (41) to (18);
		\draw [style=lipid, bend right=15, looseness=1.00] (41) to (18);
		\draw [style=lipid, bend right=15, looseness=1.00] (30) to (37);
		\draw [style=lipid, bend left=15, looseness=1.00] (30) to (37);
		\draw [style=lipid, bend left=15, looseness=1.00] (2) to (37);
		\draw [style=lipid, bend right=15, looseness=1.00] (2) to (37);
		\draw [style=lipid] (6) to (7);
		\draw [style=lipid] (7) to (21);
		\draw [style=lipid] (21) to (19);
		\draw [style=lipid] (19) to (36);
		\draw [style=lipid] (28) to (15);
		\draw [style=lipid] (15) to (32);
		\draw [style=lipid] (32) to (22);
		\draw [style=lipid] (22) to (40);
			\end{pgfonlayer}
\end{tikzpicture}
\]

\chapter{Black-boxing open Markov processes}\label{ch:blackbox}

In this section we describe a functorial method for `black-boxing' open Markov processes. To do so we exploit an analogy between Markov processes with energies, (a.k.a detailed balanced Markov processes) and electrical circuits made of resistors connected by perfectly conductive wires. In his book, Kelly \cite{Kelly} describes a similar analogy between open Markov processes and electrical circuits. We consider the situation in which the probabilities of boundary states of an open Markov process are held externally fixed via the coupling to its environment, external systems or to various reservoirs. We then study the possible steady state flows of probability through the open Markov process as they depend of the prescribed boundary probabilities. One can imagine the analogous problem in which certain nodes of an electrical network are held at fixed potentials, inducing steady state currents flowing through the system. 

The crucial difference between Markov processes and electrical circuits made of resistors is that resistors, the `edges' of an electrical circuit, have no built in `directionality.' In particular if $C_{ij} \in (0,\infty)$ is the conductance between two nodes $i$ and $j$, then we have $C_{ij} = C_{ji}$. As we saw earlier, the rates of a Markov process need not satisfy this property, that is in general $H_{ij} \neq H_{ji}$. This is part of the reason we restrict our attention to Markov processes admitting an equilibrium distribution satisfying detailed balance. In this way we can `symmetrize' a Markov process by using the fact that a detailed balanced equilibrium $q$ satisfies $H_{ij} q_j = H_{ji} q_i$. The quantity $H_{ij}q_j$ is the equilibrium flow of probability from $j$ to $i$, which for a detailed balanced $q$ will equal the equilibrium flow of probability from $i$ to $j$. This quantity will play the role of the symmetric conductance $C_{ij}$, the inverse of the resistance between nodes $i$ and $j$, in a corresponding electrical network. If one then externally forces the boundary probabilities to values differing from their equilibrium values, the Markov process will approach a steady state in which probability flows through the system. We show that such steady state minimizes a quadratic form, which we call `dissipation,' analogous to the power dissipation functional of an electrical network made of resistors. 

We thus begin this chapter with an outline of our approach to black-boxing Markov processes with energies. First, recall that a detailed balanced Markov process is one whose rates satisfy Kolmogorov's criterion. The rates of a Markov process with energies $\epsilon \maps V \to \R$ satisfy $\frac{H{ij}}{H_{ji}} = e^{-\beta(\epsilon_i - \epsilon_j)}$ for some $\beta$. We saw in Section \ref{sec:balance} that this is equivalent to Kolmogorov's criterion. Such a process admits a detailed balanced equilibrium distribution of the form $q_i = \frac{e^{-\beta \epsilon_i}}{\mathcal{Z}}$ for some constant normalizing factor $\mathcal{Z}$. Thus, we consider Markov processes with a detailed balanced equilibrium $q \maps V \to (0,\infty) $ and Markov processes with energies $\epsilon \maps V \to \R$ interchangeably. 

We show that there is a functor sending any detailed balanced Markov process to an electrical circuit, this requires a description of the category $\Circ$ whose morphisms are electrical circuits made of resistors equipped with maps specifying certain boundary nodes. We then apply the black-boxing functor for electrical networks due to Baez and Fong \cite{BaezFongCirc}. We should note that recently a new proof of the black-box theorem for electrical circuits was given which uses PROPs \cite{BaezCoyaRebro}. This diagram summarizes our method of black boxing detailed balanced Markov processes: 
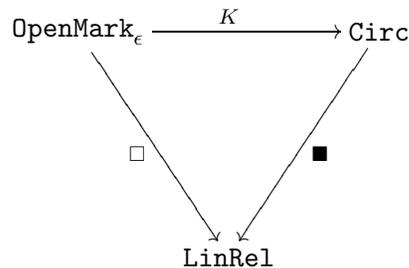
\[
   \xy
   (-20,20)*+{\DetBalMark}="1";
  (20,20)*+{\Circ}="2";
   (0,-10)*+{\LinRel}="5";
        {\ar^{K} "1";"2"};
        {\ar_{\square} "1";"5"};
        {\ar^{\blacksquare} "2";"5"};
\endxy
\]
Here we draw the black box functor $\square \maps \DetBalMark \to \LinRel$ as a white square $\square$ in order to distinguish it from the black-box functor for electrical circuits $\blacksquare \maps \Circ \to \LinRel$. Here, $\Circ$ is a decorated cospan category whose morphisms correspond to isomorphism classes of passive electrical networks made of resistors, capacitors, and inductors. We introduced the category $\LinRel$ of linear relations in Chapter 2 as an example of a behavior category. Objects in $\LinRel$ are vector spaces and morphisms are linear relations between vector spaces, which can be thought of as selecting out particular subspaces of the direct sum of the vector spaces. We now describe these categories in more detail before proving the black-box theorem for Markov processes with energies.

\section{Black-boxing open circuits}

A morphism in the category $\Circ$ is an electrical circuit made of resistors: that is,
a (directed) graph with each edge labelled by a `conductance' $c_e > 0$, again with specified input and output nodes:
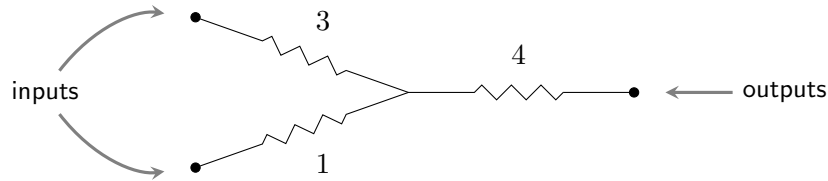
\[
\begin{tikzpicture}[circuit ee IEC, set resistor graphic=var resistor IEC graphic]
\node[contact] (I1) at (0,2) {};
\node[contact] (I2) at (0,0) {};
\node[contact] (O1) at (5.83,1) {};
\node(input) at (-2,1) {\small{\textsf{inputs}}};
\node(output) at (7.83,1) {\small{\textsf{outputs}}};
\draw (I1) 	to [resistor] node [label={[label distance=2pt]85:{$3$}}] {} (2.83,1);
\draw (I2)	to [resistor] node [label={[label distance=2pt]275:{$1$}}] {} (2.83,1)
				to [resistor] node [label={[label distance=3pt]90:{$4$}}] {} (O1);
\path[color=gray, very thick, shorten >=10pt, ->, >=stealth, bend left] (input) edge (I1);		\path[color=gray, very thick, shorten >=10pt, ->, >=stealth, bend right] (input) edge (I2);		
\path[color=gray, very thick, shorten >=10pt, ->, >=stealth] (output) edge (O1);
\end{tikzpicture}
\]

Finally, a morphism in the category $\LinRel$ is a linear relation $F \maps U \leadsto V$ between finite-dimensional real vector spaces $U$ and $V$; this is nothing but a linear subspace of $U \oplus V$.  In earlier work \cite{BaezFongCirc} we introduced the category $\Circ$ and the `black box functor' 
\[    \blacksquare \maps \Circ \to \LinRel . \]
The idea is that any circuit determines a linear relation between the potentials and net current flows at the inputs and outputs.   This relation describes the behavior of a circuit of resistors as seen from outside.

The functor $K$ converts a detailed balanced Markov process into an electrical circuit made of resistors.  This circuit is carefully chosen to reflect the steady-state behavior of the Markov process.   Its underlying graph is the same as that of the Markov process, so the `states' of the Markov process are the same as the `nodes' of the circuit.
Both the equilibrium probabilities at states of the Markov process and the rate constants labelling edges of the Markov process are used to compute the conductances of edges
of this circuit.  In the simple case where the Markov process has exactly one edge from
any state $i$ to any state $j$, the rule is
\[                 C_{i j} = H_{i j} q_j   \]
where:
\begin{itemize}
\item $q_j$ is the equilibrium probability of the $j$th state of the Markov process,
\item  $H_{i j}$ is the rate constant for the edge from the $j$th state to the $i$th state of the Markov process, and
\item  $C_{i j}$ is the conductance (that is, the reciprocal of the resistance) of the wire from the $j$th node to the $i$th node
of the resulting circuit.
\end{itemize}
The detailed balance condition for Markov processes says precisely that
the matrix $C_{i j}$ is symmetric.  This is just right for an electrical circuit made of resistors, since it means that the resistance of the wire from node $i$ to node $j$ 
equals the resistance of the same wire in the reverse direction, from node $j$ to node $i$.

The functor $\square$ maps any detailed balanced Markov process to the linear relation obeyed by probabilities and flows at the inputs and outputs in a steady state.  In short, it describes the steady state behavior of the Markov process `as seen from outside'.  We draw this functor as a white box merely to
distinguish it from the other black box functor.

The triangle of functors thus constructed does not commute!   However, 
a general lesson of category theory is that we should only expect diagrams of
functors to commute \emph{up to natural isomorphism}, and this is what happens here:
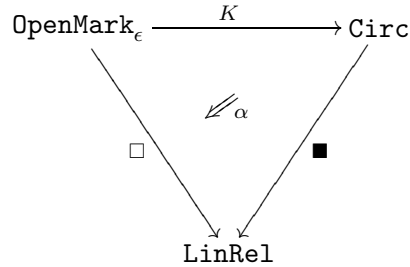
\[
   \xy
   (-20,20)*+{\DetBalMark}="1";
  (20,20)*+{\Circ}="2";
   (0,-10)*+{\LinRel}="5";
        {\ar^{K} "1";"2"};
        {\ar_{\square} "1";"5"};
        {\ar^{\blacksquare} "2";"5"};
        {\ar@{=>}^<<{\scriptstyle \alpha} (1,11); (-3,8)};
\endxy
\]
This `corrects' the black box functor for resistors to give the one for detailed
balanced Markov processes.   The functors $\square$ and $\blacksquare
\circ K$ are equal on objects.  An object in $\DetBalMark$ is a finite set $X$
with each element $i \in X$ labelled by an energy $\epsilon_i \in \R$; both these
functors map such an object to the vector space $\R^X \oplus \R^X$.    For the
functor $\square$, we think of this as a space of probability-flow pairs.  For
the functor $\blacksquare \circ K$, we think of it as a space of
potential-current pairs, since $K$ converts Markov processes to circuits made of
resistors.  The natural transformation $\alpha$ then gives a linear relation 
\[ \alpha_{X,\epsilon} \maps \R^X \oplus \R^X \leadsto \R^X \oplus \R^X ,\]
in fact an isomorphism of vector spaces, which converts potential-current pairs
into probability-flow pairs in a manner that depends on the $q_i$.  This
isomorphism maps any $n$-tuple of potentials and currents $(\phi_i, \iota_i)$
into the $n$-tuple of probabilities and flows $(p_i, j_i)$ given by
\[          p_i = \phi_i q_i,  \qquad   j_i = \iota_i  .\]

The naturality of $\alpha$ actually allows us to reduce the problem of computing
the functor $\square$ to the problem of computing $\blacksquare$.   Suppose $M
\maps (X,\epsilon) \to (Y,\upsilon)$ is any morphism in $\DetBalMark$.  The object $(X,\epsilon)$ is some finite set $X$ labelled by energies $\epsilon$, and $(Y,\upsilon)$ is some finite set
$Y$ labelled by energies $\upsilon$.  Then the naturality of $\alpha$ means that
this square commutes:
\[
  \xymatrix{
    \R^X \oplus \R^X \ar[rr]^{\blacksquare K(M)} \ar[dd]_{\alpha_{X,\epsilon}} &&  \R^Y \oplus \R^Y 
    \ar[dd]^{\alpha_{Y,\upsilon}} \\ \\ 
    \R^X \oplus \R^X \ar[rr]^{\square(M)} &&  \R^Y \oplus \R^Y 
  }
\]
Since $\alpha_{X,\epsilon}$ and $\alpha_{Y,\upsilon}$ are isomorphisms, we can solve for the
functor $\square$:
\[   \square(M) = \alpha_{Y,\upsilon}\circ \blacksquare K(M) \circ \alpha_{X,\epsilon}^{-1}  .
\]
This equation has a clear intuitive meaning: it says that to compute the behavior of
a detailed balanced Markov process, namely $\square(M)$, we convert it into a circuit made of resistors and compute the behavior of that, namely $\blacksquare K(M)$.  This is
not \emph{equal} to the behavior of the Markov process, but we can compute that 
behavior by converting the input probabilities and flows into potentials and
currents, feeding them into our circuit, and then converting the outputs back into probabilities and flows. 

\section{Principle of minimum dissipation}\label{sec:dissipation}

In Chapter \ref{ch:entropymark} we saw that by externally fixing the probabilities at boundary states of an open detailed balanced Markov process, one induces steady states which minimize a quadratic form which we called dissipation
\[ D(p) = \frac{1}{2}\sum_{i,j} H_{ij}q_j \left( \frac{p_j}{q_j} - \frac{p_i}{q_i} \right)^2 \]
where the sum is over all pairs of states.

We can also write the dissipation as a sum over edges in the open detailed balanced Markov process
\[ D(p) = \frac{1}{2} \sum_{e \in E} r_e q_{s(e)} \left( \frac{p_{s(e)}}{q_{s(e)}} - \frac{p_{t(e)}}{q_{t(e)}} \right)^2. \]

Given specified boundary probabilities, one can compute the steady state boundary flows by minimizing the dissipation subject to the boundary conditions.
\begin{defn}
We call a probability-flow pair a \define{steady state probability-flow pair} if the flows arise from a probability distribution which obeys the principle of minimum dissipation. 
\end{defn}
\begin{defn}
The \define{behavior} of an open detailed balanced Markov process with boundary $B$ is the set of all steady state probability-flow pairs $(p_B, J_B)$ along the boundary.
\end{defn}

The black-boxing functor $\square \maps \DetBalMark \to \LinRel$ maps open detailed balanced Markov processes to their steady state behaviors. The fact that this is a functor means that the behavior of a composite open detailed balanced Markov process can be computed as the composite of the behaviors in $\LinRel$. 

\section{From detailed balanced Markov processes to electrical circuits}\label{sec:reduction}

Comparing the dissipation of a detailed balanced open
Markov process:
\[ D(p) = \frac{1}{2} \sum_{e \in E} r_e q_{s(e)} \left( \frac{p_{s(e)}}{q_{s(e)}} - \frac{p_{t(e)}}{q_{t(e)}} \right)^2. \]

to the extended power functional of a circuit:
\[ P(\phi) = \frac{1}{2} \sum_{e \in E} c_e ( \phi_{s(e)} - \phi_{t(e)} )^2 \]
suggests the following correspondence:
\[ \begin{array}{ccc} 
\displaystyle{ \frac{p_i}{q_i} } &\leftrightarrow & \phi_i \\ \\
\displaystyle{ r_e q_{s(e)}} &\leftrightarrow & c_e . 
\end{array} \]
This sharpens this analogy between detailed balanced Markov processes and circuits made of resistors.   In this analogy, the probability of a state 
is \emph{roughly} analogous to the electric potential at a node.  However, it is really the `deviation', the ratio of the probability to the equilibrium probability, that 
is analogous to the electric potential.  Similarly, the rate constant of an edge is \emph{roughly}  analogous to the conductance of an edge.  However, in a circuit each edge allows for flow in both directions, while in a Markov process each edge allows for flows in only one direction.

We thus shall convert an open detailed balanced Markov process $M \maps X \to Y$, namely:
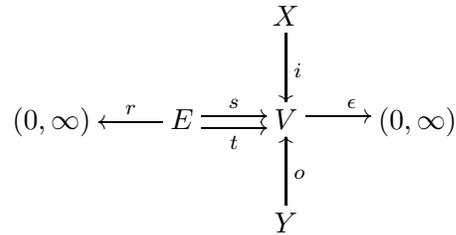
\[ 
  \xymatrix{ 
  && X \ar[d]^{i} \\
(0,\infty) & E \ar[l]_-r \ar[r]<-.5ex>_t \ar[r] <.5ex>^s & V \ar[r] <.5ex>^\epsilon &
(0,\infty) \\ 
&& Y \ar[u]_{o}} 
\]
into an open circuit $K(M) \maps X \to Y$, namely:
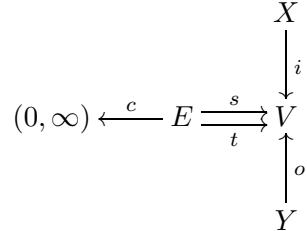
\[ 
  \xymatrix{  
    && X \ar[d]^{i} \\
    (0,\infty) & E \ar[l]_-{c} \ar[r]<-.5ex>_t  \ar[r] <.5ex>^s & V  \\
    && Y \ar[u]_{o}
  }
\]
where 
\[  c_e =  r_e q_{s(e)}. \]
For the open detailed balanced Markov process with two states depicted below this
map $K$ has the following effect:
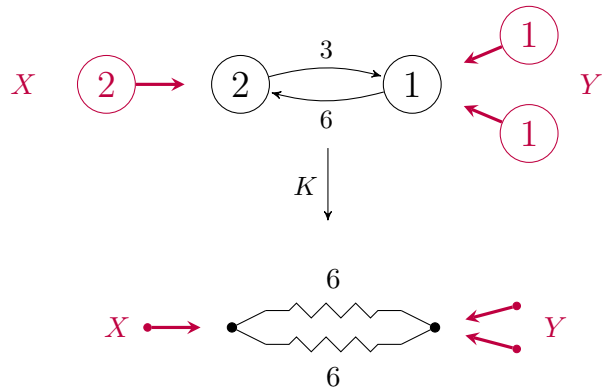
\[ 
\begin{tikzpicture}[circuit ee IEC, set resistor graphic=var resistor IEC graphic,->,>=stealth',shorten >=1pt,]
  \node[main node](1) {$2$};
  \node[main node](2) [right=1.5cm of 1] {$1$};

\node(A) [terminal, left=1cm of 1] {$2$};
\node(B) [terminal, above right=0.1cm and 1cm of 2] {$1$};
\node(C) [terminal, below right=0.1cm and 1cm of 2] {$1$};
\node[right=1.7 cm of 2,color=purple] {$Y$};
\node[left=0.4 cm of A,color=purple] {$X$};
  \path[every node/.style={font=\sffamily\small}, shorten >=1pt]
    (1) edge [bend left=15] node[above] {$3$}(2) 
    (2) edge [bend left=15] node[below] {$6$} (1);
    
\path[color=purple, very thick, shorten >=10pt, ->, >=stealth] (A) edge (1);
\path[color=purple, very thick, shorten >=10pt, ->, >=stealth] (B) edge (2);
\path[color=purple, very thick, shorten >=10pt, ->, >=stealth] (C) edge (2);

\node(D) [below right=.3 cm and 0.75cm of 1] {}; 
\node(E) [below=1cm of D] {};
  \path[every node/.style={font=\sffamily\small}, shorten >=1pt]
  (D) edge node[left] {$K$} (E);
\end{tikzpicture} 
\]
\[
\hskip 2em 
\begin{tikzpicture}[circuit ee IEC, set resistor graphic=var resistor IEC
    graphic,scale=.9]
    \node[contact]         (A) at (0,0) {};
    \node[contact]         (B) at (3,0) {};
    
        \coordinate         (ua) at (.5,.25) {};
    \coordinate         (ub) at (2.5,.25) {};
    \coordinate         (la) at (.5,-.25) {};
    \coordinate         (lb) at (2.5,-.25) {};
    \path (A) edge (ua);
    \path (A) edge (la);
    \path (B) edge (ub);
    \path (B) edge (lb);
    \path (ua) edge  [resistor] node[label={[label distance=1pt]90:{$6$}}] {} (ub);
    \path (la) edge  [resistor] node[label={[label distance=1pt]270:{$6$}}] {} (lb);
    
    \node(L) [left=1cm of A,circle,draw,inner sep=1pt,fill=gray,color=purple] {};     
\path[color=purple, very thick, shorten >=10pt, ->, >=stealth] (L) edge (A);  
\node[left=1.15 cm of A,color=purple] {$X$};
\node(R) [above right=0.2cm and 1cm of B,circle,draw,inner sep=1pt,fill=gray,color=purple] {};
\node(R2) [below right=0.2cm and 1cm of B,circle,draw,inner sep=1pt,fill=gray,color=purple] {};
\path[color=purple, very thick, shorten >=10pt, ->, >=stealth] (R) edge (B);
\path[color=purple, very thick, shorten >=10pt, ->, >=stealth] (R2) edge (B);
\node[right=1.25 cm of B,color=purple] {$Y$};
   \end{tikzpicture} 
\]

This analogy is stronger than a mere visual resemblance. The behavior for the Markov process $M$ is easily obtained from the behavior of the circuit $K(M)$.  Indeed, write $D_M$ for the extended dissipation functional of the open detailed balanced Markov process $M$, and $P_{K(M)}$ for the extended power functional of the open circuit $K(M)$. Then
\[
  D_M(p) = P_{K(M)}(\tfrac{p}{q}).
\]
Minimizing over the interior, we also have the equivalent fact for the
dissipation functional $F_M$ and power functional $Q_{K(M)}$:
\begin{lem}
  Let $M$ be an open detailed balanced Markov process, and let $K(M)$ be the
  corresponding open circuit. Then
\[
  F_M(p) = Q_{K(M)}(\tfrac{p}{q}),
\]
where $F_M$ is the dissipation functional for $M$ and $Q_{K(M)}$ is the power
functional for $K(M)$.
\end{lem}

Consider now $\Graph(\nabla F_M)$ and $\Graph(\nabla Q_{K(M)})$. These are both
subspaces of $\R^T \oplus \R^T$. For any set with probabilities $(T,q)$, define
the function
\begin{align*}  
\alpha_{T,q}\maps \R^T \oplus \R^T &\longrightarrow \R^T \oplus \R^T; \\
(\phi,\iota) &\longmapsto (q\phi,\iota).
\end{align*}
Then if $T$ is the set of terminals and $q\maps T \to (0,\infty)$ is the
restriction of the probabilities function $q\maps N \to (0,\infty)$ to the
terminals, we see that 
\[
  \Graph(\nabla F_M) = \alpha_{T,q}(\Graph(\nabla Q_{K(M)})).
\]
Note here that we are applying $\alpha_{T,q}$ pointwise to the subspace
$\Graph(\nabla Q_{K(M)})$ to arrive at the subspace $\Graph(\nabla F_M)$.

Observe that $\alpha_{T,q}$ acts as the identity on the `current' or `flow'
summand, while $S[i,o]$ acts simply as precomposition by $[i,o]$ on the
`potential' or `probability' summand. This implies the equality of the composite
relations
\[
  S[i,o] \alpha_{T,q} = (\alpha_{X,qi} \oplus \alpha_{Y,qo})S[i,o]
\]
as relations 
\[
  \R^T \oplus \R^T \leadsto \R^X \oplus \R^X \oplus \R^Y \oplus \R^Y.
\]
In summary, we have arrived at the following theorem.
\begin{thm} \label{thm.behaviors}
  Let $M$ be an open detailed balanced Markov process, and let $K(M)$ be the
  corresponding open circuit. Then
\[
  S[i,o]\Graph(\nabla F_M) = \big(\alpha_{X,qi} \oplus
  \alpha_{Y,qo}\big)S[i,o]\big(\Graph(\nabla Q_{K(M)})\big).
\]
where $F_M$ is the dissipation functional for $M$ and $Q_{K(M)}$ is the power
functional for $K(M)$.
\end{thm}
This makes precise the relationship between the two behaviors. Observe that
probability deviation is analogous to electric potential, and probability
flow is analogous to electric current.

\section{The black-box functor}
In this section we explain a `black box functor' that sends any open detailed balanced Markov process to a description of its steady-state behavior.

The first point of order is to define the category in which this steady-state
behavior lives. This is the category of linear relations. Already we have seen
that the steady states for an open detailed balanced Markov
process---that is, the solutions to the open master equation---form a linear
subspace of the vector space $\R^X \oplus \R^X \oplus \R^Y \oplus \R^Y$ of input 
and output probabilities and flows. To compose the steady states of two open Markov
processes is easy: we simply require that the probabilities and flows at the
outputs of our first Markov process are equal to the probabilities and flows at
the corresponding inputs of the second. It is also intuitive: it simply means
that we require any states we identify to have identical probabilities, and
require that at each output state all the outflow from the first Markov process 
flows into the second Markov process.

Luckily, this notion of composition for linear subspaces is already
well known: it is composition of linear relations.  We thus define the following
category:
\begin{defn}
  The category of linear relations, $\LinRel$, has finite-dimensional real
  vector spaces as objects and \define{linear relations} $L \maps U \leadsto V$,
  that is, linear subspaces $L \subseteq U \oplus V$, as morphisms from $U$ to $V$.
  The composite of linear relations $L \maps U \leadsto V$ and $L' \maps V \leadsto
  W$ is given by
  \[   L'L = \{ (u,w) : \exists v \in V \;\, (u,v) \in L \textrm{ and } (v,w) \in L' \}. \]
\end{defn}

\begin{prop}
  The category $\LinRel$ is a dagger compact category.
\end{prop}
\begin{proof}
  This is well known \cite{BaezEberleControl}. The tensor product is given by direct sum: if $L \maps U 
  \leadsto V$ and $L' \maps U' \leadsto V'$, then $L \oplus L' \maps U \oplus U' \leadsto
  V \oplus V'$ is the direct sum of the subspaces $L$ and $L'$.   The dagger is given 
  by relational transpose: if $L \maps U \leadsto V$, then 
  \[   L^\dagger = \{(v,u) : \; (u,v) \in L \} .  \qedhere \]
\end{proof}

\begin{defn}
The \define{black box functor} for detailed balanced Markov processes
\[   \square \maps \DetBalMark \to \LinRel  \]
maps each finite set with probabilities $(N,q)$ (or energies) to the vector space
 \[\square(N,q)=  \R^N \oplus \R^N\]
of boundary probabilities and boundary flows, and each open detailed balanced
Markov process $M \maps X \to Y$ to its behavior
\[ 
  \square(M) = S[i,o](\mathrm{Graph}(\nabla F)) \maps \R^X \oplus
  \R^X \leadsto \R^Y \oplus \R^Y,
\]
where $F$ is the dissipation functional of $M$. 
\end{defn}

We still need to prove that this construction actually gives a functor.  We do this
in Theorem \ref{thm:commuting_triangle} by relating this construction to 
the black box functor for circuits
\[  \blacksquare \maps \Circ \to \LinRel  ,\]
constructed by Baez and Fong in \cite{BaezFongCirc}.

To define the functor $\blacksquare$, we first construct a decorated cospan category
$\Circ$ in which the morphisms are open circuits.  In brief, let 
\[
  H\maps (\FinSet,+) \longrightarrow (\Set,\times)
\]
map each finite set $N$ to the set of circuits 
\[ \xymatrix{  (0,\infty) & E \ar[l]_-c \ar[r]<-.5ex>_t  \ar[r] <.5ex>^s & N }  \]
on $N$. This can be equipped with coherence maps to form a lax monoidal functor
in the same manner as Markov processes. Using this lax monoidal functor $H$, we
make the following definition.

\begin{defn}
The category $\Circ$ is the decorated cospan category where an object is a
finite set and a morphism is an isomorphism class of open circuits $C \maps X
\to Y$. 
\end{defn}

Again, we will often refer to a morphism as simply an open circuit;
we mean as usual the isomorphism class of the open circuit.

\begin{cor}
The category $\Circ$ is a dagger compact category.
\end{cor}

We utilize the main result from Baez and Fong's paper \cite{BaezFongCirc}: 

\begin{lem} \label{lem:blackbox}
  There exists a symmetric monoidal dagger functor, the \define{black box functor}  
for circuits: 
  \[ \blacksquare\maps \Circ \to \LinRel, \]
   mapping any finite set $X$ to the vector space
\[  \blacksquare(X) = \R^X \oplus \R^X, \] 
and any open circuit $C \maps X \to Y$ to its \define{behavior}, the linear relation 
  \[
   \blacksquare(C) = S[i,o](\mathrm{Graph}(\nabla Q))\maps \R^X \oplus \R^X
   \leadsto \R^Y \oplus \R^Y
  \]
where $Q$ is the power functional of $C$.
\end{lem}
\begin{proof}
This is a simplified version of \cite[Theorem 1.1]{BaezFongCirc}.  Note that we are treating the subspace
\[    S[i,o](\mathrm{Graph}(\nabla Q))  \subseteq 
\R^X \oplus \R^X \oplus \R^Y \oplus \R^Y .  \]
Now we are treating this subspace as a linear relation from $\blacksquare(X)$
to $\blacksquare(Y)$.
\end{proof}

\section{The functor from detailed balanced Markov processes to circuits}
\label{sec:reduction_2}

In Section \ref{sec:reduction} we described a way to model an open detailed
balanced Markov process using an open circuit, motivated by similarities between
dissipation and power. We now show that the analogy between these two structures
runs even deeper: first, this modelling process is functorial, and second, the
behaviors of corresponding Markov processes and circuits are naturally
isomorphic.

\begin{lem} \label{lem:K}
There is a symmetric monoidal dagger functor
\[ K \maps \DetBalMark \to \Circ 
\]
which maps a finite set with probabilities $(N,q)$ to the underlying finite set
$N$, and an open detailed balanced Markov process 
\[ 
  \xymatrix{ 
  && X \ar[d]^{i} \\
(0,\infty) & E \ar[l]_-r \ar[r]<-.5ex>_t \ar[r] <.5ex>^s & N \ar[r] <.5ex>^q &
(0,\infty) \\ 
&& Y \ar[u]_{o}} 
\]
to the open circuit
\[ 
  \xymatrix{  
    && X \ar[d]^{i} \\
    (0,\infty) & E \ar[l]_-{c} \ar[r]<-.5ex>_t  \ar[r] <.5ex>^s & N
    &\mbox{\phantom{$(0,\infty)$}} \\
    && Y. \ar[u]_{o}
  }
\]
where 
\[  c_e =  r_e q_{s(e)}. \]
\end{lem}

\begin{proof}  
  This is another simple application of Lemma \ref{lemma:decoratedfunctors}.
  To see that this gives a functor between the decorated cospan categories we need
  only check that the above function from detailed balanced Markov processes to
  circuits defines a monoidal natural transformation
  \[ 
    \xymatrix{ (\FinPopSet, + ) \ar[dd]_{(U,\upsilon)} \ar[drr]^{(G,\phi)} \ddrtwocell<\omit>{<0>_{\theta}} && \\
    && (\Set, \times) \\
    (\FinSet, + ) \ar[urr]_{(H,\phi')} && 
    } 
  \]
This is easy to check.
\end{proof}

In the above we have described two maps sending an open detailed balanced Markov process to a linear relation:
\[ \blacksquare \circ K \maps \DetBalMark \to \LinRel \]
and 
\[ \square \maps \DetBalMark \to \LinRel .\]
We know the first is a functor; for this second this remains to be proved.  We do this in the process of proving that these two maps are naturally isomorphic:

\begin{thm}
\label{thm:commuting_triangle}
There is a triangle of symmetric monoidal dagger functors \break between dagger compact categories: 
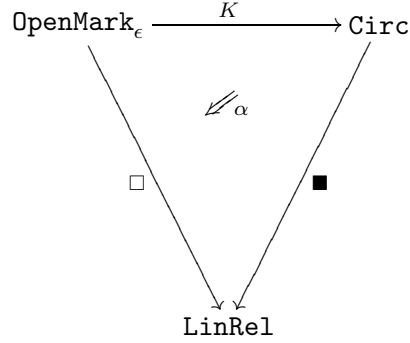
\[
   \xy
   (-20,20)*+{\DetBalMark}="1";
  (20,20)*+{\Circ}="2";
   (0,-20)*+{\LinRel}="5";
        {\ar^{K} "1";"2"};
        {\ar_{\square} "1";"5"};
        {\ar^{\blacksquare} "2";"5"};
        {\ar@{=>}^<<{\scriptstyle \alpha} (1,11); (-3,8)};
\endxy
\]
which commutes up to a monoidal natural isomorphism $\alpha$.  This natural
isomorphism assigns to each finite set with probabilities $(X,q)$ 
the linear relation $\alpha_{X,q}$ given by the linear map
\begin{align*}
  \alpha_{X,q}\maps \R^X \oplus \R^X &\longrightarrow \R^X \oplus \R^X \\
(\phi,\iota) &\longmapsto (q\phi,\iota)
\end{align*}
where $q\phi \in \R^X$ is the pointwise product of $q$ and $\phi$.
\end{thm}
\begin{proof}
We begin by simultaneously proving the functoriality of $\square$ and the naturality of
$\alpha$. The key observation is that we have the equality
\[   \square(M) = \alpha_{Y,r} \circ \blacksquare K(M) \circ \alpha_{X,q}^{-1}  \]
of linear relations $\R^X \oplus \R^X \to \R^Y \oplus \R^Y$. This is an
immediate consequence of Theorem \ref{thm.behaviors}, which relates the behavior
of the Markov process and the circuit:
\begin{align*}
  \square(M) &= S[i,o]\big(\Graph(\nabla F_M)\big) \\ 
  &= \big(\alpha_{X,q}\oplus \alpha_{Y,r}\big)S[i,o]\big(\Graph(\nabla Q_{K(M)})\big) \\
  &= \alpha_{Y,r} \circ S[i,o]\big(\Graph(\nabla Q_{K(M)})\big) \circ
  \alpha_{X,q}^{-1} \tag{$\ast$}\\
  &= \alpha_{Y,r} \circ \blacksquare K(M) \circ \alpha_{X,q}^{-1}
\end{align*}
The equation ($\ast$) may look a little unfamiliar, but is simply a switch
between two points of view: in the line above we apply the functions $\alpha$ to
the behavior, in the line below we compose the relations $\alpha$ with the
behavior. In either case the same subspace is obtained.

Another way of stating this `key observation' is as the commutativity of the
naturality square
\[
  \xymatrix{
    \R^X \oplus \R^X \ar[rr]^{\blacksquare K(M)} \ar[dd]_{\alpha_{X,q}} &&  \R^Y \oplus \R^Y 
    \ar[dd]^{\alpha_{Y,r}} \\ \\ 
    \R^X \oplus \R^X \ar[rr]^{\square(M)} &&  \R^Y \oplus \R^Y 
  }
\]
for $\alpha$. Thus if $\square$ is truly a functor, then $\alpha$ is a natural
transformation.

But the functoriality of $\square$ is now a consequence of the functoriality of
$\blacksquare$ and $K$. Indeed, for $M \maps (X,q) \to (Y,r)$ and $M' \maps (Y,r) \to (Z,s)$, we
have
\begin{align*}
 \square(M')\circ \square(M) 
 &=  \alpha_{Z,s} \circ \blacksquare K(M') \circ \alpha_{Y,r}^{-1} \circ \alpha_{Y,r} \circ
 \blacksquare K(M) \circ \alpha_{X,q}^{-1} \\
 &= \alpha_{Z,s} \circ \blacksquare K(M') \circ \blacksquare K(M) \circ
 \alpha_{X,q}^{-1} \\
 &= \alpha_{Z,s} \circ \blacksquare K(M' \circ M) \circ
 \alpha_{X,q}^{-1} \\
 &= \square(M'\circ M).
 \end{align*}
Thus $\alpha$ is a natural transformation. It is easily seen that $\alpha$ is
furthermore monoidal, and an isomorphism.  

As a consequence, the functor $\square$ can be given the structure of a symmetric 
monoidal dagger functor, in a way that makes the triangle commute up to $\alpha$.
\end{proof}

\chapter{Open reaction networks}\label{ch:rx}
We now turn our attention to constructing a compositional framework for describing `reaction networks,' certain types of labeled graphs which represent a network of interacting entities.  In the case of interacting chemical species these provide a graphical notation for a set of coupled differential equations. For reaction networks obeying mass-action kinetics these equations are generally non-linear, but the presence of conserved quantities in some cases simplifies the analysis of these equations. We are interested in open systems, where certain reaction participants are coupled to reservoirs which act as a source or sink for those participants. Much of the material of this Chapter can be found in \cite{BaezPRx}. 

In this Chapter we describe a decorated cospan category $\RxNet$ where the morphisms are isomorphism classes of open reaction networks with rates and a decorated cospan category $\Dynam$ where the morphisms are isomorphism classes of open dynamical systems. We prove that there is a `gray-boxing' functor $\graysquare \maps \RxNet \to \Dynam$ sending an open reaction network to the open dynamical system governing the evolution of the concentrations of the species in a reaction network as they evolve according to mass action kinetics. We then prove that there is a black-box functor 
\[        \blacksquare \maps \Dynam \to \SemialgRel \] 
to a category $\SemialgRel$ of `semi-algebraic relations' between real vector spaces, meaning relations defined by polynomials and inequalities. The composite of these functors provides a method for black-boxing open reaction networks.

\section{Reaction networks}
Typically, chemical reaction networks are drawn such that the vertices of the graph correspond to `complexes,' linear combinations of chemical species with natural number coefficients. These complexes are connected by directed edges which indicate the source and target of each reaction:
\[ \xymatrix{ A+B \ar[r]^-{r} & 2C} \] Each reaction is labeled by a non-negative rate constant $ r \in (0,\infty) $. There exists another graphical notation for such reaction networks as a certain type of bipartitie graph with one type of vertex corresponding to each species and another type of vertex corresponding to each reaction or transition. A directed edge from a species to a transition specifies that species as an input to that transition, i.e. that transition annihilates that species. A directed edge from a transition to a species specifies that species as an output of that transition, i.e. that transition creates members of that species: 
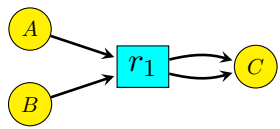
\[
\begin{tikzpicture}
	\begin{pgfonlayer}{nodelayer}
		\node [style=species] (0) at (-1, 0) {$C$};
		\node [style=species] (2) at (-4, -0.5) {$B$};
		\node [style=transition] (7) at (-2.5, 0) {$r_1$};
		\node [style=species] (18) at (-4, 0.5) {$A$};
	\end{pgfonlayer}
	\begin{pgfonlayer}{edgelayer}
		\draw [style=inarrow] (18) to (7);
		\draw [style=inarrow] (2) to (7);
		\draw [style=inarrow, bend right=15, looseness=1.00] (7) to (0);
		\draw [style=inarrow, bend left=15, looseness=1.00] (7) to (0);
	\end{pgfonlayer}
\end{tikzpicture}
\]

Omitting the rate constants labeling each transition, such structures are often called Petri nets. One should note that in the Petri net literature the vertices which chemists call species are often called `places.' Including the labels on the transitions gives a stochastic Petri net. An introduction to stochastic Petri nets can be found in \cite{StochPetri}. A Petri net can also be viewed as generating a specific monoidal category, however our approach differs from this one as we consider generalized open Petri nets as morphisms in a category themselves \cite{MeseguerMontanari, Sassone}. An introductory course on quantum techniques for stochastic Petri nets can be found in \cite{BaezBiamonte}. 

In the reaction network literature, graphs of this type were introduced under the name `SR-graphs,' short for `species-reaction graphs.' Properties of a reaction network's SR-graph can be used to determine whether a reaction network has the capacity to admit multiple steady states \cite{SRgraph}. Stronger results connecting the admissible behaviors of a reaction network to properties of its SR-graph can be found in \cite{SRgraph2, concordance}. These results depend only on the structure of the SR-graph, independent of any particular choice of rate constants. Results of this type began with the seminal work of Feinberg and Horn \cite{Feinberg, Horn}, resulting in the Deficiency Zero and Deficiency One Theorems \cite{def0,def1}. In this Chapter we show that these SR-graphs are morphisms in a decorated cospan category describing open chemical reaction networks.

\section{Open reaction networks}
In this section we describe a decorated cospan category where the morphisms are isomorphism classes of open reactoin networks. For open chemical reaction networks, the set $S$ will be the set of species, the concentration's of which evolve in time according to a set of coupled non-linear differential equations, the information for which can be encoded in a certain type of labeled graph structure which we depict below. For example, consider the reaction $ \xymatrix{ A+B \ar[r]^{r_1} & 2C }$ as an open system with three inputs $X$ and one output $Y$, which we draw in the following way:
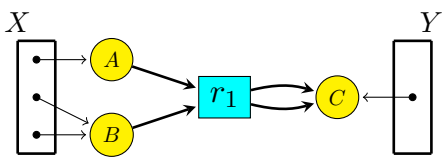
\[
\begin{tikzpicture}
	\begin{pgfonlayer}{nodelayer}
		\node [style=species] (0) at (-1, 0) {$C$};
		\node [style=none] (1) at (-0.25, 0.75) {};
		\node [style=species] (2) at (-4, -0.5) {$B$};
		\node [style=none] (3) at (-4.75, 0.75) {};
		\node [style=none] (4) at (-5.25, 0.75) {};
		\node [style=inputdot] (5) at (0, 0) {};
		\node [style=none] (6) at (0.25, 0.75) {};
		\node [style=transition] (7) at (-2.5, 0) {$r_1$};
		\node [style=none] (8) at (-5.25, -0.75) {};
		\node [style=empty] (10) at (-5.25, 1) {$X$};
		\node [style=none] (11) at (-0.25, -0.75) {};
		\node [style=inputdot] (12) at (-5, -0.5) {};
		\node [style=inputdot] (13) at (-5, 0.5) {};
		\node [style=none] (14) at (-4.75, -0.75) {};
		\node [style=none] (15) at (0.25, -0.75) {};
		\node [style=inputdot] (16) at (-5, 0) {};
		\node [style=empty] (17) at (0.25, 1) {$Y$};
		\node [style=species] (18) at (-4, 0.5) {$A$};
	\end{pgfonlayer}
	\begin{pgfonlayer}{edgelayer}
		\draw [style=inarrow] (18) to (7);
		\draw [style=inarrow] (2) to (7);
		\draw [style=inarrow, bend right=15, looseness=1.00] (7) to (0);
		\draw [style=inarrow, bend left=15, looseness=1.00] (7) to (0);
		\draw [style=inputarrow] (5) to (0);
		\draw [style=inputarrow] (13) to (18);
		\draw [style=inputarrow] (16) to (2);
		\draw [style=inputarrow] (12) to (2);
		\draw [style=simple] (4.center) to (3.center);
		\draw [style=simple] (3.center) to (14.center);
		\draw [style=simple] (14.center) to (8.center);
		\draw [style=simple] (8.center) to (4.center);
		\draw [style=simple] (1.center) to (6.center);
		\draw [style=simple] (6.center) to (15.center);
		\draw [style=simple] (15.center) to (11.center);
		\draw [style=simple] (11.center) to (1.center);
	\end{pgfonlayer}
\end{tikzpicture}
\]
Notice that the maps including the inputs and outputs into the states of the system need not be one to one.

The decoration of the cospan is accomplished via a lax monoidal functor $F \maps (\FinSet, +) \to (\Set, \times)$ sending a finite set $S$ to the set of all possible reaction networks on $S$ which we denote by $F(S)$. A particular element of $F(S)$, i.e. a particular reaction network is `picked out' by a function $ s \maps 1 \to F(s)$. A lax monoidal functor comes equipped with a natural transformation 
\[ \Phi_{S,S'} \maps F(S) \times F(S') \to F(S+S')\]
which sends a pair of reaction networks to a reaction network on the coproduct of their species sets. This functor provides a method of composing decorated cospans and therefore composing open chemical reaction networks. 

Given two systems
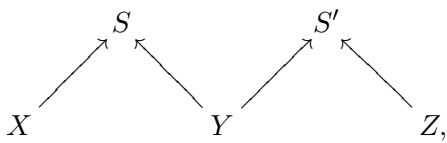
\[ \xymatrix{ & S & & S' & \\ X \ar[ur] & & Y \ar[ul] \ar[ur] & & Z, \ar[ul] } \]
if the outputs of one system match the inputs of the other system, the two can be combined to yield a new system. Composition of the cospans is done via the pushout
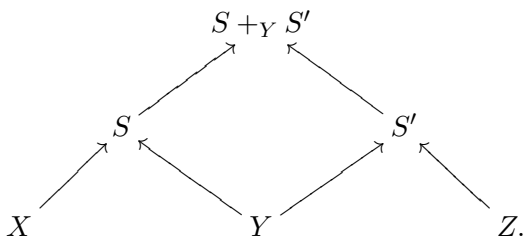
\[ \xymatrix{ & & S+_Y S' & & \\ & S \ar[ur] & & S' \ar[ul] \\
X \ar[ur] & & Y \ar[ur] \ar[ul] & & Z. \ar[ul] } \]

Consider the open chemical reactions $ \xymatrix{ A+B \ar[r]^{r_1} & 2C}$ and $\xymatrix {D \ar[r]^{r_2} & E+F }$. We can think of these as two systems, the first with species $S=\{ A,B,C\}$ the second with species $S'=\{D,E,F\}$, each with one reaction labeled with rate constants $r_1$ and $r_2$ respectively. We consider each of these as open systems, the first with 3 inputs $X$ and a single output $Y$ the second with a single input $Y$ and with a two outputs $Z$. 
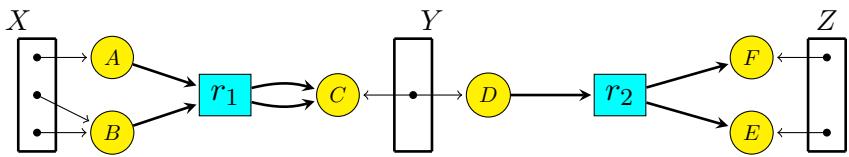
\[
\begin{tikzpicture}
	\begin{pgfonlayer}{nodelayer}
		\node [style=none] (0) at (-4.75, 0.75) {};
		\node [style=none] (1) at (-4.75, -0.75) {};
		\node [style=empty] (2) at (-5.25, 1) {$X$};
		\node[style=inputdot] (3) at (5.5,-0.5) {};
		\node [style=species] (4) at (-1, 0) {$C$};
		\node [style=none] (5) at (5.25, -0.75) {};
		\node [style=none] (6) at (0.25, -0.75) {};
		\node [style=species] (7) at (-4, -0.5) {$B$};
		\node [style=inputdot] (8) at (-5, -0.5) {};
		\node [style=none] (9) at (5.75, -0.75) {};
		\node [style=none] (10) at (-0.25, -0.75) {};
		\node [style=transition] (11) at (2.75, -0) {$r_2$};
		\node [style=species] (12) at (-4, 0.5) {$A$};
		\node [style=species] (13) at (1.0, 0) {$D$};
		\node [style=none] (14) at (5.75, 0.75) {};
		\node [style=none] (15) at (-0.25, 0.75) {};
		\node [style=inputdot] (16) at (5.5, 0.5) {};
		\node [style=inputdot] (17) at (0, 0) {};
		\node [style=none] (18) at (-5.25, -0.75) {};
		\node [style=none] (19) at (0.25, 0.75) {};
		\node [style=species] (20) at (4.5, 0.5) {$F$};
		\node [style=none] (21) at (5.25, 0.75) {};
		\node [style=inputdot] (22) at (-5, 0.5) {};
		\node [style=empty] (23) at (5.5, 1) {$Z$};
		\node [style=none] (24) at (-5.25, 0.75) {};
		\node [style=transition] (25) at (-2.5, 0) {$r_1$};
		\node [style=inputdot] (26) at (-5, 0) {};
		\node [style=empty] (27) at (0.25, 1) {$Y$};
		\node [style=species] (28) at (4.5, -0.5) {$E$};
	\end{pgfonlayer}
	\begin{pgfonlayer}{edgelayer}
		\draw [style=inarrow] (12) to (25);
		\draw [style=inarrow] (7) to (25);
		\draw[ style=inputarrow] (3) to (28);
		\draw [style=inarrow, bend right=15, looseness=1.00] (25) to (4);
		\draw [style=inarrow] (13) to (11);
		\draw [style=inarrow] (11) to (20);
		\draw [style=inarrow] (11) to (28);
		\draw [style=inarrow, bend left=15, looseness=1.00] (25) to (4);
		\draw [style=inputarrow] (17) to (13);
		\draw [style=inputarrow] (17) to (4);
		\draw [style=inputarrow] (22) to (12);
		\draw [style=inputarrow] (26) to (7);
		\draw [style=inputarrow] (8) to (7);
		\draw [style=inputarrow] (16) to (20);
		\draw [style=simple] (24.center) to (0.center);
		\draw [style=simple] (0.center) to (1.center);
		\draw [style=simple] (1.center) to (18.center);
		\draw [style=simple] (18.center) to (24.center);
		\draw [style=simple] (15.center) to (19.center);
		\draw [style=simple] (19.center) to (6.center);
		\draw [style=simple] (6.center) to (10.center);
		\draw [style=simple] (10.center) to (15.center);
		\draw [style=simple] (21.center) to (14.center);
		\draw [style=simple] (14.center) to (9.center);
		\draw [style=simple] (9.center) to (5.center);
		\draw [style=simple] (5.center) to (21.center);
	\end{pgfonlayer}
\end{tikzpicture}
\]
Since the outputs of the first process match the inputs of the second process, we can form the composite reaction network by gluing together the two reactions along their common overlap
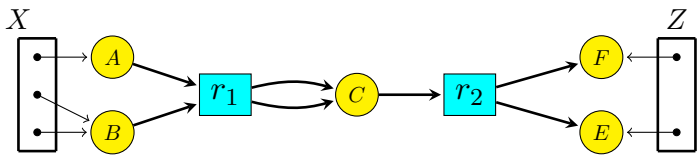
\[
\begin{tikzpicture}
	\begin{pgfonlayer}{nodelayer}
		\node [style=inputdot] (0) at (-4.25, -0) {};
		\node [style=species] (1) at (-3.25, 0.5) {$A$};
		\node [style=none] (2) at (-4, 0.75) {};
		\node [style=none] (3) at (4, -0.75) {};
		\node [style=transition] (4) at (-1.75, -0) {$r_1$};
		\node [style=none] (5) at (-4.5, 0.75) {};
		\node [style=none] (6) at (4, 0.75) {};
		\node [style=transition] (7) at (1.5, -0) {$r_2$};
		\node [style=inputdot] (8) at (4.25, 0.5) {};
		\node [style=none] (9) at (-4.5, -0.75) {};
		\node [style=species] (10) at (0, -0) {$C$};
		\node [style=none] (11) at (4.5, -0.75) {};
		\node [style=inputdot] (12) at (-4.25, 0.5) {};
		\node [style=none] (13) at (4.5, 0.75) {};
		\node [style=empty] (14) at (-4.5, 1) {$X$};
		\node [style=none] (15) at (-4, -0.75) {};
		\node [style=empty] (16) at (4.25, 1) {$Z$};
		\node [style=species] (17) at (3.25, 0.5) {$F$};
		\node [style=species] (18) at (-3.25, -0.5) {$B$};
		\node [style=inputdot] (19) at (-4.25, -0.5) {};
		\node[ style=species] (20) at (3.25, -0.5) {$E$};
		\node[ style=inputdot] (21) at (4.25, -0.5) {};
	\end{pgfonlayer}
	\begin{pgfonlayer}{edgelayer}
		\draw [style=inarrow] (1) to (4);
		\draw [style=inarrow] (18) to (4);
		\draw [style=inarrow, bend right=15, looseness=1.00] (4) to (10);
		\draw [style=inarrow] (7) to (17);
		\draw [style=inarrow, bend left=15, looseness=1.00] (4) to (10);
		\draw [style=inputarrow] (12) to (1);
		\draw [style=inputarrow] (0) to (18);
		\draw [style=inputarrow] (19) to (18);
		\draw [style=inputarrow] (8) to (17);
		\draw [style=simple] (5.center) to (2.center);
		\draw [style=simple] (2.center) to (15.center);
		\draw [style=simple] (15.center) to (9.center);
		\draw [style=simple] (9.center) to (5.center);
		\draw [style=simple] (6.center) to (13.center);
		\draw [style=simple] (13.center) to (11.center);
		\draw [style=simple] (11.center) to (3.center);
		\draw [style=simple] (3.center) to (6.center);
		\draw [style=inarrow] (10) to (7);
		\draw [style=inarrow]  (7) to (20);
		\draw[ style=inputarrow] (21) to (20);
	\end{pgfonlayer}
\end{tikzpicture}
\]

Notice that states $C$ and $D$ have been identified, leaving a composite system with species $S+_Y S' =\{ A,B,C,E,F\}$ and the reaction network above. We can think of this simply as a situation in which the output one reaction serves as an input to the next reaction. We describe a category where the morphisms are these `open reaction networks'. A functor from this category to the category of cospans decorated by algebraic relations describes the dynamical behavior of the concentrations of the species of the reaction network. 

Let us see how this works for the our example. We can write down the rate equation for each of the pieces of the above reaction network as
\[ \frac{d}{dt} \left( \begin{array}{c} A \\ B \\ C 

\end{array} \right)  = \left( \begin{array}{c} -r_1 AB \\ -r_1 AB \\ 2r_1 AB \end{array} \right) \]
and
\[ \frac{d}{dt} \left( \begin{array}{c} D \\ E \\ F 
\end{array} \right)  = \left( \begin{array}{c} -r_2 D \\ r_2 D \\ r_2 D \end{array} \right). \]
In terms of the differential equations, composition of the two process in which species $D$ and $E$ are both identified with species $C$, can be thought of has happening in two steps. First, one \emph{identifies} the concentrations of $C$ and $D$.
\[ C=D\]
Next, the contributions to the time rate of each of the identified concentrations are \emph{added}
\[ \frac{dC}{dt} + \frac{dD}{dt} = 2 r_1 AB - r_2 D. \]
Abusing notation and calling the resulting species $C$, we have the differential equation describing the evolution of the concentrations in the composite reaction network
\[ \frac{d}{dt} \left( \begin{array}{c} A \\ B \\ C \\ E \\ F
\end{array} \right)  = \left( \begin{array}{c} -r_1 AB \\ -r_1 AB \\ 2r_1 AB -r_2 C \\ r_2 C \\ r_2 C \end{array} \right). \]

This provides a compositional method for describing the dynamics on the whole reaction network. We further introduce a category where the morphisms are relations involving only the input and output species concentrations, effectively `black-boxing' the reaction network.  The composite reaction 
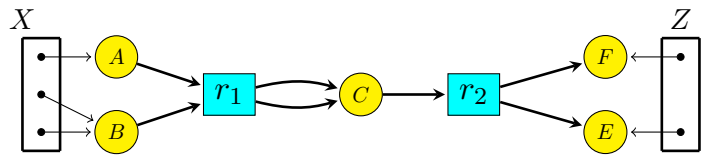
\[
\begin{tikzpicture}
	\begin{pgfonlayer}{nodelayer}
		\node [style=inputdot] (0) at (-4.25, -0) {};
		\node [style=species] (1) at (-3.25, 0.5) {$A$};
		\node [style=none] (2) at (-4, 0.75) {};
		\node [style=none] (3) at (4, -0.75) {};
		\node [style=transition] (4) at (-1.75, -0) {$r_1$};
		\node [style=none] (5) at (-4.5, 0.75) {};
		\node [style=none] (6) at (4, 0.75) {};
		\node [style=transition] (7) at (1.5, -0) {$r_2$};
		\node [style=inputdot] (8) at (4.25, 0.5) {};
		\node [style=none] (9) at (-4.5, -0.75) {};
		\node [style=species] (10) at (0, -0) {$C$};
		\node [style=none] (11) at (4.5, -0.75) {};
		\node [style=inputdot] (12) at (-4.25, 0.5) {};
		\node [style=none] (13) at (4.5, 0.75) {};
		\node [style=empty] (14) at (-4.5, 1) {$X$};
		\node [style=none] (15) at (-4, -0.75) {};
		\node [style=empty] (16) at (4.25, 1) {$Z$};
		\node [style=species] (17) at (3.25, 0.5) {$F$};
		\node [style=species] (18) at (-3.25, -0.5) {$B$};
		\node [style=inputdot] (19) at (-4.25, -0.5) {};
		\node[ style=species] (20) at (3.25, -0.5) {$E$};
		\node[ style=inputdot] (21) at (4.25, -0.5) {};
	\end{pgfonlayer}
	\begin{pgfonlayer}{edgelayer}
		\draw [style=inarrow] (1) to (4);
		\draw [style=inarrow] (18) to (4);
		\draw [style=inarrow, bend right=15, looseness=1.00] (4) to (10);
		\draw [style=inarrow] (7) to (17);
		\draw [style=inarrow, bend left=15, looseness=1.00] (4) to (10);
		\draw [style=inputarrow] (12) to (1);
		\draw [style=inputarrow] (0) to (18);
		\draw [style=inputarrow] (19) to (18);
		\draw [style=inputarrow] (8) to (17);
		\draw [style=simple] (5.center) to (2.center);
		\draw [style=simple] (2.center) to (15.center);
		\draw [style=simple] (15.center) to (9.center);
		\draw [style=simple] (9.center) to (5.center);
		\draw [style=simple] (6.center) to (13.center);
		\draw [style=simple] (13.center) to (11.center);
		\draw [style=simple] (11.center) to (3.center);
		\draw [style=simple] (3.center) to (6.center);
		\draw [style=inarrow] (10) to (7);
		\draw [style=inarrow]  (7) to (20);
		\draw[ style=inputarrow] (21) to (20);
	\end{pgfonlayer}
\end{tikzpicture}
\]
has $A$ and $B$ as inputs and $E$ and $F$ as output species. We imagine a situation in which these chemicals are coupled to reservoirs which serve to maintain the concentrations of the input and output species at fixed values, $A_0, B_0, E_0,$ and $F_0$. Then, one can write down the differential equation obeyed by the `internal' concentration $C$, namely
\[ \frac{dC}{dt} = 2r_1 A_0 B_0 - r_2 C \]
where $A_0$ and $B_0$ are constant in time. We see that in this case, the steady state concentration of species $C$, determined by $\frac{dC}{dt} = 0$, is given by
\[ C^* = \frac{2 r_1 A_0 B_0}{r_2}. \]

One can compute the flows between the system and the reservoir necessary to maintain this steady state, providing an effective `black-boxing' of the reaction which involves only the concentrations and flows along the boundary of the open reaction network. In the above example, the open reaction network admitted a unique steady state. In general, open reaction networks can admit multiple steady states with distinct boundary concentration-flow pairs.

Even for `closed reaction networks', namely those with no boundary species, multiple steady states can be present in relatively simple reaction networks. Such reaction networks often require features such as \define{autocatalysis}, meaning a certain species is present as both an input and an output to the same reaction. Another common feature of simple reaction networks admitting multiple steady states is the presence of \define{trimolecular reactions}, meaning that three of the same species appear either as an input or an output to a reaction. For spatially uniform systems where the velocity distribution of species molecules is Maxwellian, the probability of a tri-molecular collision becomes exceedingly rare, so in the simplest regime such reactions are somewhat unphysical. For example, the reaction 
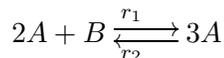
\[ \xymatrix{ 2A +B \ar@<0.5ex>[r]^{r_1}  & 3A \ar@<0.5ex>[l]^{r_2 } } \]
admits multiple steady states as a closed system. The corresponding rate equation for the concentrations of $A$ and $B$ is
\[ \frac{d}{dt} \left( \begin{array}{c} A \\ B \end{array} \right) = \left( \begin{array}{c} r_1 A^2 B - r_2 A^3 \\ -r_1 A^2 B +r_2 A^3 \end{array} \right).  \]
Steady states correspond to 
\[ r_1 A^2 B = r_2 A^3 \]
with the two solutions $A=0$ and $ r_2 A = r_1 B $. Notice the two solutions intersect only at $A=B=0$, the trivial solutions. Since the system of closed and consists of only a single reaction taking three molecules as the input and outputting three molecules, the total number of molecules will be conserved. 

This Chapter is structured as follows. In Section \ref{sec:rxnet}, we construct a decorated cospan category $ \RNet $ where an object is a finite set and a morphism is an isomorphism class of cospans of finite sets decorated by a reaction network. In Section \ref{sec:RxNet}, we construct a similar category $\RxNet$ where an object is a finite set and a morphism is an isomorphism class of cospans of finite sets decorated by a reaction network with \emph{rates}. The inclusion of these `rates' allow one to specify a partiuclar set of non-linear ordinary differential equations associated to a reaction network with rates. In Section \ref{sec:opendynam}, we construct a decorated cospan category $\Dynam$, the category of open dynamical systems where an object is a finite set and a morphism is an isomorphism class of cospans of finite sets decorated by an algebraic vector field, where by algebraic we mean a vector field defined component-wise in terms of polynomial relations in the concentrations of the species with real number coefficients.

In Section \ref{sec:gray}, we construct a `gray-boxing' functor $ \graysquare \maps \RxNet \to \Dynam$ sending a reaction network with rates on a finite set $S$ of species to a vector field $ v_S \maps \R^S \to \R^S$ which sends a vector of concentrations $c_S \in \R^S$ to the corresponding reaction velocities $v_S(c_S) \in \R^S$ which generate the time evolution of the concentrations via the additional relation
\[ \frac{dc_S}{dt} = v_S. \] This provides a functorial description of the dynamical behavior of the concentrations of all species in a reaction network. We then restrict our attention to the situation in which the concentrations of certain `boundary' species are held fixed via the coupling with some external `reservoir.' In this setting we define a `black-boxing' functor characterizing the behavior of a reaction network in terms of the behavior of its boundary concentrations and flows or velocities. 

\section{Reaction networks}
\label{sec:rxnet}

In this section we define a graph structure which we call a reaction network. These are certain types of bipartite graphs with species vertices and transition vertices. Edges either specify a species as an input to a transition or as an output to a transition. In the literature such graphs are often called Petri-Nets, place-transition nets, or SR-graphs. The term `stochastic Petri-net' was introduced in \cite{BaezBiamonte} for such graphs where the transitions are labeled by positive rate constants.  We construct the decorated cospan category $\RNet$ where objects are finite sets and morphisms are isomorphism classes of cospans in $\FinSet$ decorated by a reaction network.

\begin{defn}
A \define{reaction network} $(S,T,s,t)$ consists of:
\begin{itemize}
\item a finite set $S$,
\item a finite set $T$,
\item functions $s,t \maps T \to \N^S$.
\end{itemize}
We call the elements of $S$ \define{species}, those of $\N^S$ \define{complexes},
and those of $T$ \define{transitions}.  Any transition $\tau \in T$ has a \define{source}
$s(\tau)$ and \define{target} $t(\tau)$.  If $s(\tau) = \kappa$ and $t(\tau) = \kappa'$ 
we write $\tau \maps \kappa \to \kappa'$. 
\end{defn}

The set of complexes relevant to a given reaction network $(S,T,s,t)$ is
\[            K = \im(s) \cup \im(t) \subseteq \N^S. \]
The reaction network gives a graph with $K$ as its set of vertices and a directed 
edge from $\kappa \in K$ to $\kappa' \in K$ for each transition $\tau \maps \kappa 
\to \kappa'$.     This graph may have multiple edges or self-loops.  It is thus
the kind of graph sometimes called a `directed multigraph' or
`quiver'.   However, a graph of this kind can only arise from a reaction network if 
every vertex is the source or target of some edge.

We are interested in viewing reaction networks as open systems. 

\begin{defn}
Given finite sets $X$ and $Y$, an \define{open reaction network from $X$ to $Y$} is a cospan of finite sets
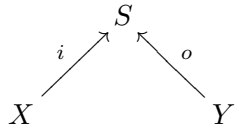
\[ \xymatrix{  & S  &  \\ X \ar[ur]^{i} & & Y \ar[ul]_{o} } \] 
together with a reaction network $\mathtt{R}$ on $S$. We often abbreviate such an open reaction network as $\mathtt{R} \maps X \to Y$. We say $X$ is the set of \define{inputs} of the open reaction network and $Y$ is its set of \define{outputs}. We call a species in $B = i(X) \cup o(Y)$ a \define{boundary species} of the open reaction network while a species in $S-B$ is called \define{internal}.
\end{defn}

Let $F(S)$ denote the set of all possible reaction networks on $S$.

\begin{lem}
Given an function $f \maps S \to S'$, there is a function $F(f) \maps F(S) \to F(S')$ that maps a reaction network on $S$ to a reaction network on $S'$. Moreover we have that
\[ F(fg) = F(f)F(g) \]
and
\[ F(1_S) = 1_{F(S)} \]
when $1_S \maps S \to S$ is the identity map.
\end{lem}

Note this lemma actually says that $F \maps (\FinSet,+) \to (\Set, \times)$ is a functor. \begin{lem}
\label{lemma:RFunctor}
There is a functor $F \maps \FinSet \to \Set$ such that:
\begin{itemize}
\item For any finite set $S$, $F(S)$ is the set of all reaction networks on $S$.
\item For any function $f \maps S \to S'$ between finite sets and any reaction network $(S,T,s,t) \in F(S)$, we have
\[    F(f)(S,T,s,t) = (S',T, f_*(s), f_*(t)) \]
where 
\beq  
f_*(s)(\tau)(\sigma') =\sum_{\{\sigma \in S : f(\sigma) = \sigma' \}} s(\tau)(\sigma)  
\label{stilde}
\eeq
and 
\beq  
f_*(t)(\tau)(\tilde{\sigma}) =\sum_{\{\sigma \in S : f(\sigma) = \sigma' \}} t(\tau)(\sigma) .
\label{ttilde}
\eeq
\end{itemize}
\end{lem}
\begin{proof}
To prove that $F$ is a functor we need only check that $F$ preserves composition
and sends identity functions to identity functions. Both are straightforward calculations.
\end{proof}

To get a decorated cospan category whose morphisms are open reaction networks we need that $F \maps \FinSet \to \Set$ be a lax monoidal functor.

\begin{lem}\label{lemma:RLax}
For any pair of finite sets $S$ and $S'$, there is a map $\varphi_{S,S'} \maps F(S) \times F(S') \to F(S+S')$ sending any pair consisting of a reaction network on $S$ and a reaction network on $S'$ to a reaction network on $S+S'$. This map makes $F \maps \FinSet \to \Set$ into a lax monoidal functor.
\end{lem}

\begin{proof}
We define the natural transformation \[   \varphi_{S,S'} \maps F(S) \times F(S') \to F(S + S') \]
that sends a pair of reaction networks, one on $S$ and one on $S'$, to one
on the disjoint union $S + S'$ of these two sets.   

The map $\varphi_{S,S'}$ is defined by
\beq
\label{eq:Phi}
    \varphi_{S,S'} ((S,T,s,t), (S',T',s',t')) = (S + S', T + T', s+s', t+t').  
\eeq
In more detail, for the reaction network on the right hand side:
\begin{itemize}
\item The set of species is the disjoint union $S + S'$.
\item The set of transitions is the disjoint union $T + T'$.
\item The source map $s + s' \maps T + T' \to S + S'$ sends any transition $\tau \in T$ 
to $s(\tau) \in S$ and any transition $\tau' \in T'$ to $s'(\tau') \in S'$.  
\item The target map $t + t' \maps T + T' \to S + S'$ sends any transition $\tau \in T$ 
to $t(\tau) \in S$ and any transition $\tau' \in T'$ to $t'(\tau') \in S'$.  
\end{itemize}

We define the unit map assigning the trivial decoration
\[  \varphi \maps  1  \to F(\emptyset) \]
where $1$ is a one-element set, sending this one element to the unique reaction network with no species and no transitions. Commutativity of the hexagon and left/right unitor squares required for laxness of $F$ follows from the universal property of the coproduct in $\FinSet$.
\end{proof}

Suppose we have an open reaction network $\mathtt{R} \maps X \to Y$ and an open reaction network $\mathtt{R'} \maps Y \to Z$. Thus we have cospans of finite sets
\[ \xymatrix { & S & & S' & \\ X \ar[ur]^{i} & &  Y \ar[ul]_o \ar[ur]^{i'} & & Z \ar[ul]_{o'} } \] decorated with open reaction networks $\mathtt{R} \in F(S)$, $\mathtt{R'} \in F(S')$. To define the \define{composite} open reaction network $\mathtt{R'R} \maps X \to Z$ we need to decorate the composite cospan
\[ \xymatrix{ & S +_Y S' & \\ X \ar[ur]^{j \circ i} & & \ar[ul]_{j' \circ o'} Z. } \]
We have the map $[j,j'] \maps S+S' \to S+_Y S'$ including the disjoint union $S+S'$ into the pushout $S+_Y S'$. To choose an element of $F(S +_Y S')$, we take $\Phi_{S,S'}(\mathtt{R},\mathtt{R'})$, an open reaction network on $S+S'$ and apply the map 
\[F([j,j']) \maps F(S+S') \to F(S+_Y S').\] This gives the required reaction network on $S +_Y S'$.

This method of composition \emph{almost} gives a category with finite sets as objects and open reaction networks $ \mathtt{R} \maps X \to Y$ as morphisms. However since the disjoint union of sets is associative only up to isomorphism, then so is composition of open reaction networks. We are really dealing with a structure called a bicategory. For now we wish to work with a mere category. To this end we take \emph{isomorphism classes} of open reaction networks $\mathtt{R} \maps X \to Y$ as morphisms from $X$ to $Y$. In fact, it was recently shown that decorated cospan categories always give rise to certain bicategorical structures \cite{Kenny}.

\begin{thm}
There is a category $\RNet$ where:
\begin{itemize}
\item an object is a finite set,
\item a morphism from $X$ to $Y$ is an equivalence class of open reaction networks 
from $X$ to $Y$, 
\item Given morphisms represented by an open reaction network from $X$ to $Y$ and one from $Y$ to $Z$:
 \[
    (X \stackrel{i}\longrightarrow S \stackrel{o}\longleftarrow Y, R) 
    \quad \textrm{ and } \quad
    (Y \stackrel{i'}\longrightarrow S' \stackrel{o'}\longleftarrow Z, R'), 
  \]
their composite is the equivalence class of this cospan constructed via a pushout:
  
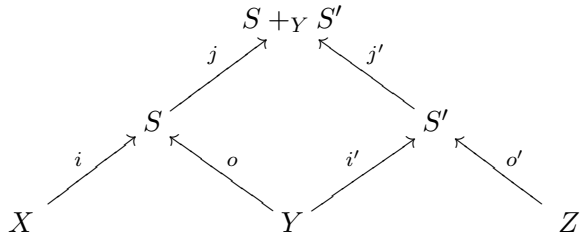
\[
    \xymatrix{
      && S +_Y S' \\
      & S \ar[ur]^{j} && S' \ar[ul]_{j'} \\
      \quad X\quad \ar[ur]^{i} && Y \ar[ul]_{o} \ar[ur]^{i'} &&\quad Z \quad \ar[ul]_{o'}
    }
  \]
together with the reaction network on $S +_Y S'$ obtained by applying the map
\[      
\xymatrix{      F(S) \times F(S') \ar[rr]^-{\varphi_{S,S'}} && 
                     F(S + S') \ar[rr]^-{F([j,j'])} && F(S +_Y S') } \]
to the pair $(R,R') \in F(S) \times F(S')$.  
\end{itemize}
\end{thm}

\begin{proof}
This follows from Lemma \ref{lemma:RLax} together with Lemma 
\ref{lemma:fcospans}, where we explain the equivalence relation in detail.  
\end{proof}

In fact, Fong's machinery proves more.  We can take the `tensor product' of  two open reaction networks by setting them side by side.  They then act in parallel with no interaction between them.  This makes $\RNet$ into symmetric monoidal category. 
In fact it is one of a very nice sort, called a `hypergraph category'.  For more on this concept, see Fong's thesis \cite{FongThesis}.

\begin{thm}
The category $\RNet$ is a symmetric monoidal category where the tensor product of objects $X$ and $Y$ is their disjoint union $X + Y$, while the tensor product of the morphisms
\[
    (X \stackrel{i}{\longrightarrow} S \stackrel{o}{\longleftarrow} Y, R) 
    \quad \textrm{ and } \quad
    (X' \stackrel{i'}{\longrightarrow} S' \stackrel{o'}{\longleftarrow} Y', R) 
  \]
is defined to be
\[  ( X + X' \stackrel{i+i'}{\longrightarrow} S + S'  \stackrel{o + o'}{\longleftarrow} Y + Y', \;
\varphi_{S,S'}(R,R') ) .\]
In fact $\RNet$ is a hypergraph category, and thus a dagger-compact category.
\end{thm}

\begin{proof}
This follows from Theorem 3.4 of Fong's paper on decorated cospans \cite{Fong}.
\end{proof}

\section{Open reaction networks with rates}
\label{sec:RxNet}

To get a dynamical system from an open reaction network we need to equip its with nonnegative real numbers called `rate constants':

\begin{defn}
A \define{reaction network with rates} $R = (S,T,s,t,r)$ consists of a reaction network
$(S,T,s,t)$ together with a function $r \maps T \to (0,\infty)$ specifying a \define{rate constant} for each transition.   We call any reaction network with rates having $S$ as its set of species a reaction network with rates \define{on} $S$.
\end{defn}

\noindent
Just as reaction networks are equivalent to Petri nets, reaction networks with rates
are equivalent to Petri nets where each transition is equipped with a rate constant.  
These are usually called `stochastic Petri nets', because they can be used to define 
stochastic processes \cite{BaezBiamonte, GossPeccoud, Haas}.   

The results of the last section easily generalize to reaction networks with rates: 

\begin{defn}
Given finite sets $X$ and $Y$, an \define{open reaction network with rates from} $X$ \define{to} $Y$ is a cospan of finite sets
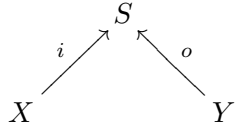
\[ \xymatrix{  & S  &  \\ X \ar[ur]^{i} & & Y \ar[ul]_{o} } \] 
together with a reaction network with rates on $R$ on $S$. We often abbreviate all this data as $R \maps X \to Y$. 
\end{defn}

\begin{lem}
\label{lemma:RxFunctor}
There is a functor $F \maps \FinSet \to \Set$ such that:
\begin{itemize}
\item For any finite set $S$, $F(S)$ is the set of all reaction networks with rates on $S$.
\item For any function $f \maps S \to S'$ between finite sets and any reaction network with rates $(S,T,s,t,r) \in F(S)$, we have
\[    F(f)(S,T,s,t,r) = (S',T, f_*(s), f_*(t), r) \]
where $f_*(s)$ and $f_*(t)$ are defined as in Equations \eqref{stilde} and \eqref{ttilde}.
\end{itemize}
\end{lem}

\begin{proof}
This is a slight variation on Lemma \ref{lemma:RFunctor}.
\end{proof}

\begin{lem}
\label{lemma:RxLax}
The functor $F$ can be made lax symmetric monoidal from $(\FinSet, +)$ to $(\Set, \times)$.  To do this we equip it with the map $\varphi \maps 1 \to F(\emptyset)$ sending the one element of $1$ to the unique reaction network with rates having no species and no transitions, together with natural transformation $\varphi_{S,S'} \maps F(S) \times F(S') \to F(S + S')$ 
such that
\[
    \varphi_{S,S'} ((S,T,s,t,r), (S',T',s',t',r')) = (S + S', T + T', s+s', t+t',[r,r']).  
\]
where:
\begin{itemize}
\item The map $s + s' \maps T + T' \to S + S'$ sends any transition $\tau \in T$ 
to $s(\tau) \in S$ and any transition $\tau' \in T'$ to $s'(\tau') \in S'$.  
\item The map $t + t' \maps T + T' \to S + S'$ sends any transition $\tau \in T$ 
to $t(\tau) \in S$ and any transition $\tau' \in T'$ to $t'(\tau') \in S'$.  
\item The map $[r,r'] \maps T + T' \to [0,\infty)$ sends any transition $\tau \in T$ to $r(\tau)$ and any transition $\tau' \in T'$ to $r'(\tau')$.  
\end{itemize}
\end{lem}

\begin{proof}
This is a slight variation on Lemma \ref{lemma:RLax}. 
\end{proof}

We now come to the star of the show, the category $\RxNet$:

\begin{thm}
\label{thm:RxNet}
There is a category $\RxNet$ where:
\begin{itemize}
\item an object is a finite set,
\item a morphism from $X$ to $Y$ is an equivalence class of open reaction networks 
with rates from $X$ to $Y$, 
\item Given morphisms represented by an open reaction network with rates from $X$ to $Y$ and one from $Y$ to $Z$:
 \[
    (X \stackrel{i}\longrightarrow S \stackrel{o}\longleftarrow Y, R) 
    \quad \textrm{ and } \quad
    (Y \stackrel{i'}\longrightarrow S' \stackrel{o'}\longleftarrow Z, R'), 
  \]
their composite consists of the equivalence class of this cospan:
  \[
    \xymatrix{
      & S +_Y S' \\
      \quad X\quad \ar[ur]^{ji} && \quad Z \quad \ar[ul]_{j'o'}
    }
  \]
together with the reaction network with rates on $S +_Y S'$ 
obtained by applying the map
\[      
\xymatrix{      F(S) \times F(S') \ar[rr]^-{\varphi_{S,S'}} && 
                     F(S + S') \ar[rr]^-{F([j,j'])} && F(S +_Y S') } \]
to the pair $(R,R') \in F(S) \times F(S')$.  
\end{itemize}
The category $\RxNet$ is a symmetric monoidal category where the tensor product of objects $X$ and $Y$ is their disjoint union $X + Y$, while the tensor product of the morphisms
\[
    (X \stackrel{i}{\longrightarrow} S \stackrel{o}{\longleftarrow} Y, R) 
    \quad \textrm{ and } \quad
    (X' \stackrel{i'}{\longrightarrow} S' \stackrel{o'}{\longleftarrow} Y', R) 
  \]
is defined to be
\[  ( X + X' \stackrel{i+i'}{\longrightarrow} S + S'  \stackrel{o + o'}{\longleftarrow} Y + Y', \; \varphi_{S,S'}(R,R') ) .\]
In fact $\RxNet$ is a hypergraph category.
\end{thm}

\begin{proof}
This follows from Lemmas \ref{lemma:RxLax} and \ref{lemma:fcospans}, 
where we explain the equivalence relation in detail.  
\end{proof}

\section{The open rate equation}
\label{sec:openrate}

In chemistry, a reaction network with rates is frequently used as a tool to specify a dynamical system.   A dynamical system is often defined as a smooth manifold $M$ whose points are `states', together with a smooth vector field on $M$ saying how these states evolve in time.   In chemistry we take $M = [0,\infty)^S$ where $S$ is the set of species: a point $c \in [0,\infty)^S$ describes the \define{concentration} $c_\sigma$ of each species $\sigma \in S$.   The corresponding dynamical system is a first-order differential equation called the `rate equation':
\[   \frac{dc(t)}{dt} = v(c(t)) \]
where now $c \maps \R \to [0,\infty)^S$ describes the concentrations as a function of 
time, and $v$ is a vector field on $[0,\infty)^S$.   Of course, $[0,\infty)^S$ is not a smooth manifold.  On the other hand, the vector field $v$ is better than smooth: its components are polynomials, so it is \define{algebraic}.   For mathematical purposes this lets us treat $v$ as an vector field on all of $\R^S$, even though negative concentrations are unphysical.  

In more detail, suppose $R = (S,T,s,t,r)$ is a reaction network with rates. 
Then the rate equation is determined by a rule called the \define{law of mass action}. 
This says that each transition $\tau \in T$ contributes 
to $dc(t)/dt$ by the product of:
\begin{itemize}
\item the rate constant $r(\tau)$, 
\item the concentration of each species $\sigma$ raised to the power given by the number of times $\sigma$ appears as an input to $\tau$, namely $s(\tau)(\sigma)$, and
\item the vector $t(\tau) - s(\tau) \in \R^S$ whose $\sigma$th component is the change
in the number of items of the species $\sigma \in S$ caused by the transition $\tau$,
\end{itemize}
The second factor is a product over all species, and it deserves an abbreviated notation: 
given $c \in \R^S$ we write
\begin{equation}
\label{eq:power}
 c^{s(\tau)} = \prod_{\sigma \in S} {c_\sigma}^{s(\tau)(\sigma)} .
\end{equation}
Thus:

\begin{defn} 
We say the time-dependent concentrations $c \maps \R \to \R^S$ obey the
\define{rate equation} for the reaction network with rates $R = (S,T,s,t,r)$ if
\[
\frac{dc(t)}{d t} = 
\sum_{\tau \in T} r(\tau) (t(\tau) - s(\tau)) c(t)^{s(\tau)} .
\]
\end{defn}

\noindent
For short, we can write the rate equation as
\[    \frac{dc(t)}{d t} = v^R(c(t)) \]
where $v^R$ is the vector field on $\R^S$ given by
\begin{equation}
\label{eq:v^R}
v^R(c) = \sum_{\tau \in T} r(\tau) \, ( t(\tau) - s(\tau) ) c^{s(\tau)} 
\end{equation}
at any point $c \in \R^S$.  We call the components of this vector field \define{reaction velocities}, since in the rate equation they describe rates of change of concentrations.   

Given an \emph{open} reaction network with rates, we can go further: we can obtain an \emph{open} dynamical system.   We give a specialized definition of this concept suited to the case at hand:

\begin{defn}
Given finite sets $X$ and $Y$, an \define{open dynamical system from} $X$ \define{to} $Y$ is a cospan of finite sets
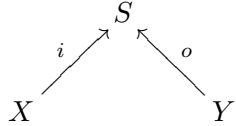
\[ \xymatrix{  & S  &  \\ X \ar[ur]^{i} & & Y \ar[ul]_{o} } \] 
together with an algebraic vector field $v$ on $\R^S$. 
\end{defn}

\noindent
The point is that given an open dynamical system of this sort, we can write down a generalization of the rate equation that takes into account `inflows' and `outflows' as well as the intrinsic dynamics given by the vector field $v$.   

To make this precise, let the
\define{inflows} $I \maps \R \to \R^X $ and \define{outflows} $O \maps \R \to \R^Y$ be arbitrary smooth functions of time.   We write the inflow at the point $x \in X$ as $I_x(t)$ or simply $I_x$, and similarly for the outflows.   Given an open dynamical system and a choice of inflows and outflows, we define the pushforward $i_*(I) \maps \R \to \R^S$ by
\[ i_*(I)_\sigma = \sum_{ \{ x : i(x) = \sigma \} } I_x \]
and define $o_*(O) \maps \R \to \R^S$ by
\[   o_*(O)_\sigma = \sum_{ \{ y : o(y) = \sigma \} } O_y \]
With this notation, the \define{open rate equation} is
\[ \frac{dc(t)}{dt} = v(c(t)) + i_*(I(t)) - o_*(O(t)). \]
The pushforwards here say that for any species $\sigma \in S$, the time derivative of 
the concentration $c_\sigma(t)$ takes into account the sum of all inflows at $x \in X$ such that $i(x) = \sigma$, minus the sum of outflows at $y \in Y$ such that $i(y) = \sigma$.  

To make these ideas more concrete, let us see in an example how to go from an open reaction network with rates to an open dynamical system and then its open rate equation.  Let $R$ be the following reaction network:
\[   \xymatrix{ A+B \ar[r]^{\tau} & C + D. }  \] 
The set of species is $S = \{A,B,C,D\}$ and the set of transitions is just $T = \{\tau\}$.   We can make $R$ into a reaction network with rates by saying the rate constant of $\tau$ is some positive number $r$.   This gives a vector field 
\[      v^R(A,B,C,D) = (-r A B, -r A B, r A B, r A B) \]
where we abuse notation in a commonly practiced way and use $(A,B,C,D)$ as the  coordinates for a point in $\R^S$: that is, the concentrations of the four species
with the same names.  The resulting rate equation is
\[ 
\begin{array}{rcl} 
\displaystyle{\frac{dA(t)}{dt}} &=& - r A(t) B(t) \\ \\
\displaystyle{\frac{dB(t)}{dt}} &=& - r A(t) B(t) \\ \\
\displaystyle{\frac{dC(t)}{dt}}&=& r A(t) B(t) \\ \\
\displaystyle{\frac{dD(t)}{dt}} &=& r A(t) B(t) .
\end{array}
\]
Next, we can make $R$ into an open reaction network $R \maps X \to Y$ as follows:
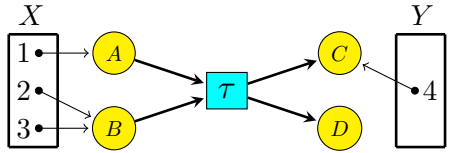
\[
\begin{tikzpicture}
	\begin{pgfonlayer}{nodelayer}
		\node [style=species] (A) at (-4, 0.5) {$A$};
		\node [style=species] (B) at (-4, -0.5) {$B$};
		\node [style=species] (C) at (-1, 0.5) {$C$};
		\node [style=species] (D) at (-1, -0.5) {$D$};
             \node [style=transition] (a) at (-2.5, 0) {$\tau$}; 
		
		\node [style=empty] (X) at (-5.1, 1) {$X$};
		\node [style=none] (Xtr) at (-4.75, 0.75) {};
		\node [style=none] (Xbr) at (-4.75, -0.75) {};
		\node [style=none] (Xtl) at (-5.4, 0.75) {};
             \node [style=none] (Xbl) at (-5.4, -0.75) {};
	
		\node [style=inputdot] (1) at (-5, 0.5) {};
		\node [style=empty] at (-5.2, 0.5) {$1$};
		\node [style=inputdot] (2) at (-5, 0) {};
		\node [style=empty] at (-5.2, 0) {$2$};
		\node [style=inputdot] (3) at (-5, -0.5) {};
		\node [style=empty] at (-5.2, -0.5) {$3$};

		\node [style=empty] (Y) at (0.1, 1) {$Y$};
		\node [style=none] (Ytr) at (.4, 0.75) {};
		\node [style=none] (Ytl) at (-.25, 0.75) {};
		\node [style=none] (Ybr) at (.4, -0.75) {};
		\node [style=none] (Ybl) at (-.25, -0.75) {};

		\node [style=inputdot] (4) at (0, 0) {};
		\node [style=empty] at (0.2, 0) {$4$};
		
		
	\end{pgfonlayer}
	\begin{pgfonlayer}{edgelayer}
		\draw [style=inarrow] (A) to (a);
		\draw [style=inarrow] (B) to (a);
		\draw [style=inarrow] (a) to (C);
		\draw [style=inarrow] (a) to (D);
		\draw [style=inputarrow] (1) to (A);
		\draw [style=inputarrow] (2) to (B);
		\draw [style=inputarrow] (3) to (B);
		\draw [style=inputarrow] (4) to (C);
		\draw [style=simple] (Xtl.center) to (Xtr.center);
		\draw [style=simple] (Xtr.center) to (Xbr.center);
		\draw [style=simple] (Xbr.center) to (Xbl.center);
		\draw [style=simple] (Xbl.center) to (Xtl.center);
		\draw [style=simple] (Ytl.center) to (Ytr.center);
		\draw [style=simple] (Ytr.center) to (Ybr.center);
		\draw [style=simple] (Ybr.center) to (Ybl.center);
		\draw [style=simple] (Ybl.center) to (Ytl.center);
	\end{pgfonlayer}
\end{tikzpicture}
\]
Here $X = \{1,2,3\}$ and $Y = \{4,5\}$, while the functions $i \maps X \to S$ and $o \maps Y \to S$ are given by
\[      i(1) = A, \; i(2) = i(3) = B, \; o(4) = C. \]
The corresponding open dynamical system is the span
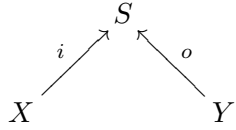
\[ \xymatrix{  & S  &  \\ X \ar[ur]^{i} & & Y \ar[ul]_{o} } \]
decorated by the vector field $v^R$ on $\R^S$.   Finally, the corresponding open rate equation is
\[ 
\begin{array}{rcl} 
\displaystyle{\frac{dA(t)}{dt}} &=& - r A(t) B(t)  + I_1(t)\\ \\
\displaystyle{\frac{dB(t)}{dt}} &=& - r A(t) B(t) + I_2(t) + I_3(t) \\ \\
\displaystyle{\frac{dC(t)}{dt}} &=& r A(t) B(t) - O_4(t) \\ \\
\displaystyle{\frac{dD(t)}{dt}} &=& r A(t) B(t).
\end{array}
\]
Note that $dB/dt$ involves the sum of two inflow terms $I_1$ and $I_2$ 
since $i(2) = i(3) = B$, while $dD/dt$ involves neither inflow nor outflow terms 
since $D$ is in the range of neither $i$ nor $o$.  

\section{The category of open dynamical systems}
\label{sec:opendynam}

There is a category $\Dynam$ where the morphisms are open dynamical systems ----or more precisely, certain equivalence classes of these.  We compose two open dynamical systems by connecting the outputs of the first to the inputs of the second.  

To construct the category $\Dynam$, we again use the machinery of decorated cospans.  
For this we need a lax monoidal functor $D \maps \FinSet \to \Set$ sending any finite set $S$ to the set of algebraic vector fields on $\R^S$.

\begin{lem}
\label{lemma:Dfunctor}
There is a functor $D \maps \FinSet \to \Set$ such that:
\begin{itemize}
\item $D$ maps any finite set $S$ to 
\[ D(S) = \{ v \maps \R^S \to \R^S : \; v \textrm{ is algebraic}  \}. \] 
\item $D$ maps any function $f \maps S \to S'$ between finite sets to the function $D(f) \maps D(S) \to D(S')$ given as follows:
\[ D(f)(v) = f_* \circ v \circ f^* \]
where the pullback $ f^* \maps \R^{S'} \to \R^S $ is given by
\[ f^*(c)(\sigma) = c(f(\sigma)) \] 
while the pushforward $ f_* \maps \R^{S} \to \R^{S'} $ is given by
\[ f_*(c)(\sigma') = \sum_{ \{ \sigma \in S : f(\sigma) = \sigma' \} } c(\sigma). \]
\end{itemize}
\end{lem}

\begin{proof}
The functoriality of $D$ follows from the fact that pushforward is a covariant functor and pullback is a contravariant functor:
\[   D(f)D(g)(v) = f_* \circ g_* \circ v \circ g^* \circ f^* = (f\circ g)_* \circ v \circ (f\circ g)^* = D(fg)(v).  \qedhere \]
\end{proof}

\begin{lem}
\label{lemma:DLax}
The functor $D$ becomes lax symmetric monoidal from $(\FinSet, +)$ to $(\Set, \times)$ if we equip it with the natural transformation 
\[ \delta_{S,S'} \maps D(S) \times D(S') \to D(S + S') \]
given by
\[  \delta_{S,S'}(v,v') = i_* \circ v \circ i^* + i'_* \circ v' \circ {i'}^* \]
together with the unique map $\delta \maps 1 \to D(\emptyset)$.
Here $i \maps S \to S+S'$ and $i' \maps S' \to S+S'$ are the inclusions of $S$ and $S'$ into their disjoint union, and we add vector fields in the usual way.
\end{lem}

\begin{proof}
By straightforward calculations one can verify all the conditions in the definition of lax symmetric monoidal functor, Def.\ \ref{defn.lmf}.   \end{proof}

\begin{thm}
\label{thm:dynam}
There is a category $\Dynam$ where:
\begin{itemize}
\item an object is a finite set,
\item a morphism from $X$ to $Y$ is an equivalence class of open dynamical
systems from $X$ to $Y$, 
\item Given an open dynamical system from $X$ to $Y$ and one from $Y$ to $Z$:
 \[
    (X \stackrel{i}\longrightarrow S \stackrel{o}\longleftarrow Y, v) 
    \quad \textrm{ and } \quad
    (Y \stackrel{i'}\longrightarrow S' \stackrel{o'}\longleftarrow Z, v'), 
  \]
their composite consists of the equivalence class of this cospan:
  \[
    \xymatrix{
      & S +_Y S' \\
      \quad X\quad \ar[ur]^{ji} && \quad Z \quad \ar[ul]_{j'o'}
    }
  \]
together with the algebraic vector field on $\R^{S+_Y S'}$ obtained by applying the map
\[      
\xymatrix{      D(S) \times D(S') \ar[rr]^-{\delta_{S,S'}} && 
                     D(S + S') \ar[rr]^-{D([j,j'])} && D(S +_Y S') } \]
to the pair $(v,v') \in D(S) \times D(S')$.  
\end{itemize}
The category $\Dynam$ is a symmetric monoidal category where the tensor product of objects $X$ and $Y$ is their disjoint union $X + Y$, while the tensor product of the morphisms
\[
    (X \stackrel{i}{\longrightarrow} S \stackrel{o}{\longleftarrow} Y, v) 
    \quad \textrm{ and } \quad
    (X' \stackrel{i'}{\longrightarrow} S' \stackrel{o'}{\longleftarrow} Y', v) 
  \]
is defined to be
\[  ( X + X' \stackrel{i+i'}{\longrightarrow} S + S'  \stackrel{o + o'}{\longleftarrow} Y + Y', \; \delta_{S,S'}(v,v') ) .\]
In fact $\Dynam$ is a hypergraph category.
\end{thm}

\begin{proof}
This follows from Lemmas \ref{lemma:fcospans} and \ref{lemma:DLax}.
\end{proof}

\section{The gray-boxing functor}
\label{sec:gray}

Now we are ready to describe the `gray-boxing' functor $\graysquare \maps \RxNet \to \Dynam$.  This sends any open reaction network to the open dynamical system that it determines.  The functoriality of this process says that we can first compose networks and then find the open dynamical system of the resulting larger network, or first find the open dynamical system for each network and then compose these systems: either way, the result is the same.

To construct the gray-boxing functor we again turn to Fong's theory of decorated
cospans.  Just as this theory gives decorated cospan categories from lax symmetric 
monoidal functors, it gives functors between such categories from monoidal natural transformations.

\begin{thm}
There is a symmetric monoidal functor $\graysquare \maps \RxNet \to \Dynam$ that is
the identity on objects and sends each morphism represented by an open reaction network $( X \stackrel{i}{\to} S \stackrel{o}{\leftarrow} Y , R )$ to the morphism represented by the open dynamical system $( X \stackrel{i}{\to} S \stackrel{o}{\leftarrow} Y , v^R )$, where $v^R$ is defined by Equation \ref{eq:v^R}.  Moreover, $\graysquare$ is a hypergraph functor.
\end{thm}

\begin{proof}
Recall that the functors $F,D \maps (\FinSet,+) \to (\Set,\times)$ assign to a set $S$ the set of all possible reaction networks on $S$ and the set of all algebraic vector fields on $S$, respectively.  By Lemma \ref{lemma:decoratedfunctors}, we may obtain a hypergraph functor $\graysquare \maps \RxNet \to \Dynam$ from a monoidal natural transformation $\theta_S \maps F(S) \to D(S)$.    For any $R \in F(S)$, let us define $\theta_S(R) \in D(S)$ by
\[  \theta_S(R) = v^R .\]
where the vector field $v^R$ is given by Equation \ref{eq:v^R}.

To check the naturality of $\theta$, we must prove that the following square commutes:
\[  \xymatrix{   F(S) \ar_{\theta_S}[d] \ar^{F(f)}[r] & F(S') \ar^{\theta_{S'} }[d] \\
D(S) \ar_{D(f)}[r] & D(S') } \]
where $F(f)$ was defined in Lemma \ref{lemma:RxFunctor} and $D(f)$ was defined in Lemma \ref{lemma:Dfunctor}.   So, consider any element $R \in F(S)$: that is, any reaction network with 
rates 
\[    R = (S,T,s,t,r) .\]
Let $R' = F(f)(R)$.   Thus,
\[   R' = (S',T,f_*(s),f_*(t),r).\]
We need to check that $D(f)(v^R) = v^{R'}$, or in other words,
\[    f_* \circ v^R \circ f^* = v^{R'}  .\]

To do this, recall from Equation \ref{eq:v^R} that for any concentrations $c' \in \R^{S'}$ we have
\[    v^{R'}(c') = \sum_{\tau \in T} r(\tau) (f_*(t)(\tau) - f_*(s)(\tau)) \;{c'}^{f_*(s)(\tau)} .\]
Using Equation (\ref{eq:power}) and the definition of pushforward and pullback, we obtain
\[  
\begin{array}{ccl}
{c'}^{f_*(s)(\tau)} &=& \displaystyle{\prod_{\sigma' \in S'} {c'_{\sigma'}}^{f_*(s)(\tau)(\sigma')} }
\\ \\ 
&=& \displaystyle{ \prod_{\sigma' \in S'} {c'_{\sigma'}}^{\sum_{\{\sigma : \; f(\sigma) = \sigma'\}} s(\tau)(\sigma)} } \\ \\
&=& \displaystyle{ \prod_{\sigma' \in S'}  \prod_{\{\sigma : \; f(\sigma) = \sigma'\}} 
{c'_{\sigma'}}^{s(\tau)(\sigma)} } \\ \\
&=& \displaystyle{  \prod_{\sigma \in S} {c'_{f(\sigma)}}^{s(\tau)(\sigma)} }  \\ \\
&=& \displaystyle{  \prod_{\sigma \in S} {f^*(c')_\sigma}^{s(\tau)(\sigma)} }  \\ \\
&=& f^*(c')^{s(\tau)} .
\end{array} 
\]
Thus, 
\[
\begin{array}{ccl}
   v^{R'}(c') &=&
\displaystyle{ \sum_{\tau \in T} r(\tau) (f_*(t)(\tau) - f_*(s)(\tau))  \; f^*(c')^{s(\tau)} } \\ \\
&=& f_*(v^R (f^*(c'))) .
\end{array}
\]
so $v^{R'} = f_* \circ v^R \circ f^*$ as desired.

We must also check that $\theta$ is monoidal.   By Definition \ref{def:monnattran}, this means showing
that
   \[
    \xymatrix{
      F(S) \times F(S') \ar[r]^-{\varphi_{S,S'}} \ar[d]_{\theta_S \times
      \theta_{S'}} & F(S + S') \ar[d]^{\theta_{S + S'}} \\
      D(S) \times D(S') \ar[r]^-{\delta_{S,S'}} & D(S+S')
    }
  \]
commutes for all $S,S'\in \FinSet$, where $\varphi$ was defined in Lemma \ref{lemma:RxLax} and
$\delta$ was defined in Lemma \ref{lemma:DLax}.   This is straighforward.
\end{proof}

The idea behind this theorem is best explained with an example.  Consider two composable open reaction networks with rates.   The first, $R \maps X \to Y$, is this:
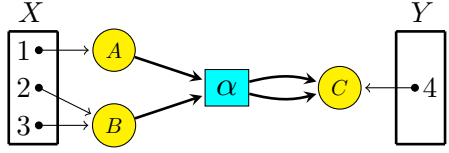
\[
\begin{tikzpicture}
	\begin{pgfonlayer}{nodelayer}
		\node [style=species] (A) at (-4, 0.5) {$A$};
		\node [style=species] (B) at (-4, -0.5) {$B$};
		\node [style=species] (C) at (-1, 0) {$C$};
             \node [style=transition] (a) at (-2.5, 0) {$\alpha$}; 
		
		\node [style=empty] (X) at (-5.1, 1) {$X$};
		\node [style=none] (Xtr) at (-4.75, 0.75) {};
		\node [style=none] (Xbr) at (-4.75, -0.75) {};
		\node [style=none] (Xtl) at (-5.4, 0.75) {};
             \node [style=none] (Xbl) at (-5.4, -0.75) {};
	
		\node [style=inputdot] (1) at (-5, 0.5) {};
		\node [style=empty] at (-5.2, 0.5) {$1$};
		\node [style=inputdot] (2) at (-5, 0) {};
		\node [style=empty] at (-5.2, 0) {$2$};
		\node [style=inputdot] (3) at (-5, -0.5) {};
		\node [style=empty] at (-5.2, -0.5) {$3$};

		\node [style=empty] (Y) at (0.1, 1) {$Y$};
		\node [style=none] (Ytr) at (.4, 0.75) {};
		\node [style=none] (Ytl) at (-.25, 0.75) {};
		\node [style=none] (Ybr) at (.4, -0.75) {};
		\node [style=none] (Ybl) at (-.25, -0.75) {};

		\node [style=inputdot] (4) at (0, 0) {};
		\node [style=empty] at (0.2, 0) {$4$};
		
	\end{pgfonlayer}
	\begin{pgfonlayer}{edgelayer}
		\draw [style=inarrow] (A) to (a);
		\draw [style=inarrow] (B) to (a);
		\draw [style=inarrow, bend left =15] (a) to (C);
		\draw [style=inarrow, bend right =15] (a) to (C);
		\draw [style=inputarrow] (1) to (A);
		\draw [style=inputarrow] (2) to (B);
		\draw [style=inputarrow] (3) to (B);
		\draw [style=inputarrow] (4) to (C);
		\draw [style=simple] (Xtl.center) to (Xtr.center);
		\draw [style=simple] (Xtr.center) to (Xbr.center);
		\draw [style=simple] (Xbr.center) to (Xbl.center);
		\draw [style=simple] (Xbl.center) to (Xtl.center);
		\draw [style=simple] (Ytl.center) to (Ytr.center);
		\draw [style=simple] (Ytr.center) to (Ybr.center);
		\draw [style=simple] (Ybr.center) to (Ybl.center);
		\draw [style=simple] (Ybl.center) to (Ytl.center);
	\end{pgfonlayer}
\end{tikzpicture}
\]
It has species $S = \{A,B,C\}$ and transitions $T = \{\alpha\}$.  
The vector field describing its dynamics is 
\begin{equation}
\label{eq_v_1}
 v^R(A,B,C) = ( -r(\alpha) AB, -r(\alpha) AB , 2r(\alpha) AB). 
\end{equation}
The corresponding open rate equation is 
\begin{equation}
\label{eq:open_rate_1}
\begin{array}{rcl} 
\displaystyle{\frac{dA(t)}{dt}} &=& - r(\alpha) A(t) B(t)  + I_1(t)\\ \\
\displaystyle{\frac{dB(t)}{dt}} &=& - r(\alpha) A(t) B(t) + I_2(t) + I_3(t) \\ \\
\displaystyle{\frac{dC(t)}{dt}} &=& 2r(\alpha) A(t) B(t) - O_4(t).
\end{array}
\end{equation}

The second open reaction network with rates, $R' \maps Y \to Z$, is this:
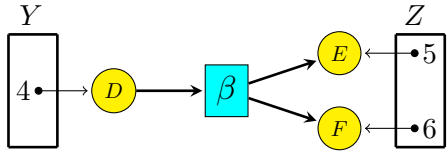
\[
\begin{tikzpicture}
	\begin{pgfonlayer}{nodelayer}
		\node [style = species] (D) at (1, 0) {$D$};
		\node [style = transition] (b) at (2.5, 0) {$\beta$};
		\node [style = species] (E) at (4,0.5) {$E$};
		\node [style = species] (F) at (4,-0.5) {$F$};

		\node [style=empty] (Y) at (-0.1, 1) {$Y$};
		\node [style=none] (Ytr) at (.25, 0.75) {};
		\node [style=none] (Ytl) at (-.4, 0.75) {};
		\node [style=none] (Ybr) at (.25, -0.75) {};
		\node [style=none] (Ybl) at (-.4, -0.75) {};

		\node [style=inputdot] (4) at (0, 0) {};
		\node [style=empty] at (-0.2, 0) {$4$};
		
		\node [style=empty] (Z) at (5, 1) {$Z$};
		\node [style=none] (Ztr) at (4.75, 0.75) {};
		\node [style=none] (Ztl) at (5.4, 0.75) {};
		\node [style=none] (Zbl) at (5.4, -0.75) {};
		\node [style=none] (Zbr) at (4.75, -0.75) {};

		\node [style=inputdot] (5) at (5, 0.5) {};
		\node [style=empty] at (5.2, 0.5) {$5$};	
		\node [style=inputdot] (6) at (5, -0.5) {};
		\node [style=empty] at (5.2, -0.5) {$6$};	

	\end{pgfonlayer}
	\begin{pgfonlayer}{edgelayer}
		\draw [style=inarrow] (D) to (b);
		\draw [style=inarrow] (b) to (E);
		\draw [style=inarrow] (b) to (F);
		\draw [style=inputarrow] (4) to (D);
		\draw [style=inputarrow] (5) to (E);
		\draw [style=inputarrow] (6) to (F);
		\draw [style=simple] (Ytl.center) to (Ytr.center);
		\draw [style=simple] (Ytr.center) to (Ybr.center);
		\draw [style=simple] (Ybr.center) to (Ybl.center);
		\draw [style=simple] (Ybl.center) to (Ytl.center);
		\draw [style=simple] (Ztl.center) to (Ztr.center);
		\draw [style=simple] (Ztr.center) to (Zbr.center);
		\draw [style=simple] (Zbr.center) to (Zbl.center);
		\draw [style=simple] (Zbl.center) to (Ztl.center);
	\end{pgfonlayer}
\end{tikzpicture}
\]
It has species $S'=\{D,E,F\}$ and transitions $T' = \{\beta\}$.    The vector field describing its dynamics is
\begin{equation}
\label{eq:v_2}
 v^{R'}(D,E,F) = ( -r(\beta) D , r(\beta) D , r(\beta) D). 
\end{equation}
The corresponding open rate equation is
\begin{equation}
\label{eq:open_rate_2}
\begin{array}{rcl} 
\displaystyle{\frac{dD(t)}{dt}} &=& - r(\beta) D(t)  + I_4(t)\\ \\
\displaystyle{\frac{dE(t)}{dt}} &=& r(\beta) D(t) - O_5(t) \\ \\
\displaystyle{\frac{dF(t)}{dt}} &=& r(\beta) D(t) - O_6(t).
\end{array}
\end{equation}

Composing $R$ and $R'$ gives $R' R \maps X \to Z$, which looks like this:
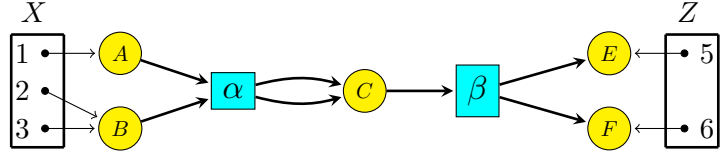
\[
\begin{tikzpicture}
	\begin{pgfonlayer}{nodelayer}
		\node [style=inputdot] (0) at (-4.25, 0) {};
		\node [style=empty] at (-4.55, 0) {$2$};
		\node [style=species] (1) at (-3.25, 0.5) {$A$};
		\node [style=none] (2) at (-4, 0.75) {};
		\node [style=none] (3) at (4, -0.75) {};
		\node [style=transition] (4) at (-1.75, -0) {$\alpha$};
		\node [style=none] (5) at (-4.7, 0.75) {};
		\node [style=none] (6) at (4, 0.75) {};
		\node [style=transition] (7) at (1.5, -0) {$\beta$};
		\node [style=inputdot] (8) at (4.25, 0.5) {};
		\node [style=empty] at (4.55, 0.5) {$5$};
		\node [style=none] (9) at (-4.7, -0.75) {};
		\node [style=species] (10) at (0, -0) {$C$};
		\node [style=none] (11) at (4.7, -0.75) {};
		\node [style=inputdot] (12) at (-4.25, 0.5) {};
		\node [style=empty] at (-4.55, 0.5) {$1$};
		\node [style=none] (13) at (4.7, 0.75) {};
		\node [style=empty] (14) at (-4.4, 1.05) {$X$};
		\node [style=none] (15) at (-4, -0.75) {};
		\node [style=empty] (16) at (4.3, 1.05) {$Z$};
		\node [style=species] (17) at (3.25, 0.5) {$E$};
		\node [style=species] (18) at (-3.25, -0.5) {$B$};
		\node [style=inputdot] (19) at (-4.25, -0.5) {};
		\node [style=empty] at (-4.55, -0.5) {$3$};
		\node[ style=species] (20) at (3.25, -0.5) {$F$};
		\node[ style=inputdot] (21) at (4.25, -0.5) {};
		\node [style=empty] at (4.55, -0.5) {$6$};
	\end{pgfonlayer}
	\begin{pgfonlayer}{edgelayer}
		\draw [style=inarrow] (1) to (4);
		\draw [style=inarrow] (18) to (4);
		\draw [style=inarrow, bend right=15, looseness=1.00] (4) to (10);
		\draw [style=inarrow] (7) to (17);
		\draw [style=inarrow, bend left=15, looseness=1.00] (4) to (10);
		\draw [style=inputarrow] (12) to (1);
		\draw [style=inputarrow] (0) to (18);
		\draw [style=inputarrow] (19) to (18);
		\draw [style=inputarrow] (8) to (17);
		\draw [style=simple] (5.center) to (2.center);
		\draw [style=simple] (2.center) to (15.center);
		\draw [style=simple] (15.center) to (9.center);
		\draw [style=simple] (9.center) to (5.center);
		\draw [style=simple] (6.center) to (13.center);
		\draw [style=simple] (13.center) to (11.center);
		\draw [style=simple] (11.center) to (3.center);
		\draw [style=simple] (3.center) to (6.center);
		\draw [style=inarrow] (10) to (7);
		\draw [style=inarrow]  (7) to (20);
		\draw[ style=inputarrow] (21) to (20);
	\end{pgfonlayer}
\end{tikzpicture}
\]
Note that the species $C$ and $D$ have been identified, and we have arbitrarily called the resulting species $C$.   Thus, $R'R$ has species $S+_Y S' =\{ A,B,C,E,F\}$.   If we had chosen to call the resulting state something else, we would obtain an equivalent open reaction network with rates, and thus the same morphism in $\RxNet$.

At this point we can either compute the vector field $v^{R'R}$ or combine the vector fields $v^R$ and $v^{R'}$ following the procedure given in Theorem \ref{thm:dynam}.  Because $\graysquare$ is a functor, we should get the same answer either way.

The vector field $v^{R'R}$ can be read off from the above picture of $R'R$.  It is
\begin{equation}
\label{eq:v_3}
     v^{R'R}(A,B,C,E,F) = 
\end{equation}
\[
(-r(\alpha) A B, \; -r(\alpha) A B, \; 2 r(\alpha) AB - r(\beta) C, \; r(\beta) C, \; r(\beta) C ) .
\]
On the other hand, the procedure in Theorem \ref{thm:dynam} is to apply the composite of these maps:
\[      
\xymatrix{      D(S) \times D(S') \ar[rr]^-{\delta_{S,S'}} && 
                     D(S + S') \ar[rr]^-{D([j,j'])} && D(S +_Y S') } \]
to the pair of vector fields $(v^R,v^{R'}) \in D(S) \times D(S')$.    The first map
was defined in Lemma \ref{lemma:DLax}, and it yields
\[   \begin{array}{ccl}
  \delta_{S,S'}(v^R,v^{R'}) &=& i_* \circ v^R \circ i^* + i'_* \circ v^{R'} \circ {i'}^* \\  \\
&=&  ( -r(\alpha) AB, \; -r(\alpha) AB , \; 2r(\alpha) AB, \, 0, \; 0, \; 0) \; + \\ 
&& (0, \; 0,\; 0,\, -r(\beta) D , \;r(\beta) D , \; r(\beta) D) \\ \\
&=& ( -r(\alpha) AB,\; -r(\alpha) AB ,\; 2r(\alpha) AB, \; -r(\beta) D , \; r(\beta) D ,\; r(\beta) D).
\end{array}
\]
If we call this vector field $u$, the second map yields
\[  D([j,j'])(u) = [j,j']_* \circ u \circ [j,j']^*.\]
Applying $[j,j']^*$ to any vector of concentrations $(A,B,C,E,F) \in \R^{S +_Y S'}$ yields
$(A,B,C,C,E,F) \in \R^{S+S'}$, since the species $C$ and $D$ are identified by $[j,j']$.   Thus, 
\[   u \circ [j,j']^* =  ( -r(\alpha) AB, -r(\alpha) AB , 2r(\alpha) AB, -r(\beta) C , r(\beta) C , r(\beta) C). \]
Applying $[j,j']_*$ to this, we sum the third and fourth components, again because $C$
and $D$ are identified by $[j,j']$.  Thus,
\[   [j,j']_* \circ u \circ [j,j']^* =  ( -r(\alpha) AB, \, -r(\alpha) AB , \, 2r(\alpha) AB -r(\beta) C , \, r(\beta) C , \, r(\beta) C). \]
As expected, this vector field equals $v^{R'R}$.     The open rate equation of the composite open dynamical system is
\begin{equation}
\label{eq:open_rate_3}
\begin{array}{rcl} 
\displaystyle{\frac{dA(t)}{dt}} &=& - r(\alpha) A(t) B(t)  + I_1(t)\\ \\
\displaystyle{\frac{dB(t)}{dt}} &=& - r(\alpha) A(t) B(t) + I_2(t) + I_3(t) \\ \\
\displaystyle{\frac{dC(t)}{dt}} &=& 2r(\alpha) A(t) B(t) - r(\beta) C(t)  \\ \\
\displaystyle{\frac{dD(t)}{dt}} &=& r(\beta) C(t) - O_5(t) \\ \\
\displaystyle{\frac{dE(t)}{dt}} &=& r(\beta) C(t) - O_6(t) .
\end{array}
\end{equation}

One general lesson here is that when we compose open reaction networks, the process of identifying some of their species via the map $[j,j'] \maps S + S' \to S +_Y S'$ has two effects: copying concentrations and summing reaction velocities.  Concentrations are copied via the pullback $[j,j']^*$, while reaction velocities are summed via the pushfoward $[j,j']_*$.  A similar phenomenon occurs the compositional framework for electrical circuits, where voltages are copied and currents are summed \cite{BaezFongCirc}.  For a deeper look at this, see Section 6.6 of Fong's thesis \cite{FongThesis}.

\section{The black-boxing functor}
\label{sec:black}

The open rate equation describes the behavior of an open dynamical system for any choice of inflows and outflows.  One option is to choose these flows so that the input and output concentrations do not change with time.  In chemistry this is called `chemostatting'.    There will then frequently---though not always---be solutions of the open rate equation where \emph{all} concentrations are constant in time.   These are called `steady states'.  

In this section we take an open dynamical system and extract from it the relation between input and output concentrations and flows that holds in steady state.  We call the process of extracting this relation `black-boxing', since it discards information that cannot be seen at the inputs and ouputs.  The relation thus obtained is always `semi-algebraic', meaning that it can be described by polynomials and inequalities.  In fact, black-boxing defines a functor 
\[        \blacksquare \maps \Dynam \to \SemialgRel \]
where $\SemialgRel$ is the category of semi-algebraic relations between real vector spaces.    The functoriality of black-boxing means that we can compose two open dynamical systems and then black-box them, or black-box each one and compose the resulting relations: either way, the final answer is the same.   

We can also black-box open reaction networks with rates.  To do this, we simply compose the gray-boxing functor with the black-boxing functor:
\[   \RxNet \stackrel{\graysquare}{\longrightarrow} \Dynam \stackrel{\blacksquare}{\longrightarrow} \SemialgRel .\]

We begin by explaining how to black-box an open dynamical system.

\begin{defn}
Given an open dynamical system
$ (X \stackrel{i}\longrightarrow S \stackrel{o}\longleftarrow Y, v) $
we define the \define{boundary} species to be those in $B = i(X) \cup o(Y)$, and the \define{internal} species to be those in $S - B$.
\end{defn}

The open rate equation says
\[ \frac{dc(t)}{dt} = v(c(t)) + i_*(I(t)) - o_*(O(t)) \]
but if we fix $c$, $I$ and $O$ to be constant in time, this reduces to
\[     v(c) + i_*(I) - o_*(O) = 0 .\]   
This leads to the definition of `steady state':

\begin{defn}
Given an open dynamical system 
$ (X \stackrel{i}\longrightarrow S \stackrel{o}\longleftarrow Y, v) $  
together with $I \in \R^X$ and $O \in \R^Y$,
a \define{steady state} with inflows $I$ and outflows $O$ is an element $c \in \R^S$ such that 
\[     v(c) + i_*(I) - o_*(O) = 0 .\]   
\end{defn}

Thus, in a steady state, the inflows and outflows conspire to exactly compensate for
the reaction velocities.    In particular, we must have
\[   \left. v(c)\right|_{S - B} = 0 \]
since the inflows and outflows vanish on internal species.  

\begin{defn}
Given a morphism $F \maps X \to Y$ iin $\Dynam$ represented by the open dynamical system 
\[     (X \stackrel{i}\longrightarrow S \stackrel{o}\longleftarrow Y, v), \]
define its \define{black-boxing} to be the set
\[   \blacksquare(F) \subseteq \R^X \oplus \R^X \oplus \R^Y \oplus \R^Y \]
consisting of all 4-tuples $(i^*(c),I,o^*(c),O)$ where $c \in \R^S$ is a steady state
with inflows $I \in \R^X$ and outflows $O \in \R^Y$.
\end{defn}

We call $i^*(c)$ the \define{input concentrations} and $o^*(c)$ the \define{output concentrations}.   Thus, black-boxing records the relation between input concentrations, inflows, output concentrations and outflows that holds in steady state.  This is the `externally observable steady state behavior' of the open dynamical system.

Category theory enters the picture because relations are morphisms in a category.  For any sets $X$ and $Y$, a relation $A \maps X \relto Y$ is a subset $A \subseteq X \times Y$.   Given relations $A \maps X \relto Y$ and $B \maps Y \relto Z$, their composite $B \circ A \maps X \relto Z$ is the set of all pairs $(x,z) \in X \times Z$ such that there exists $y \in Y$ with $(x,y) \in A$ and $(y,z) \in B$.   This gives a category $\Rel$ with sets as objects and relations as morphisms.   

Black-boxing an open dynamical system $F \maps X \to Y$ gives a relation
\[   \blacksquare(F) \maps \R^X \oplus \R^X \relto \R^Y \oplus \R^Y  .\]
This immediately leads to the question of whether black-boxing is a functor from $\Dynam$ to 
$\Rel$.  

The answer is yes.  To get a sense for this, consider the example from Section \ref{sec:gray}, where we composed two open dynamical systems.   We first considered this open reaction network with rates:
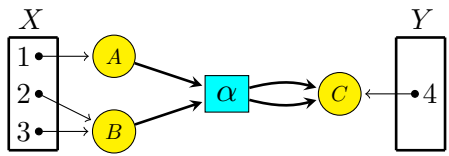
\[
\begin{tikzpicture}
	\begin{pgfonlayer}{nodelayer}
		\node [style=species] (A) at (-4, 0.5) {$A$};
		\node [style=species] (B) at (-4, -0.5) {$B$};
		\node [style=species] (C) at (-1, 0) {$C$};
             \node [style=transition] (a) at (-2.5, 0) {$\alpha$}; 
		
		\node [style=empty] (X) at (-5.1, 1) {$X$};
		\node [style=none] (Xtr) at (-4.75, 0.75) {};
		\node [style=none] (Xbr) at (-4.75, -0.75) {};
		\node [style=none] (Xtl) at (-5.4, 0.75) {};
             \node [style=none] (Xbl) at (-5.4, -0.75) {};
	
		\node [style=inputdot] (1) at (-5, 0.5) {};
		\node [style=empty] at (-5.2, 0.5) {$1$};
		\node [style=inputdot] (2) at (-5, 0) {};
		\node [style=empty] at (-5.2, 0) {$2$};
		\node [style=inputdot] (3) at (-5, -0.5) {};
		\node [style=empty] at (-5.2, -0.5) {$3$};

		\node [style=empty] (Y) at (0.1, 1) {$Y$};
		\node [style=none] (Ytr) at (.4, 0.75) {};
		\node [style=none] (Ytl) at (-.25, 0.75) {};
		\node [style=none] (Ybr) at (.4, -0.75) {};
		\node [style=none] (Ybl) at (-.25, -0.75) {};

		\node [style=inputdot] (4) at (0, 0) {};
		\node [style=empty] at (0.2, 0) {$4$};
		
	\end{pgfonlayer}
	\begin{pgfonlayer}{edgelayer}
		\draw [style=inarrow] (A) to (a);
		\draw [style=inarrow] (B) to (a);
		\draw [style=inarrow, bend left =15] (a) to (C);
		\draw [style=inarrow, bend right =15] (a) to (C);
		\draw [style=inputarrow] (1) to (A);
		\draw [style=inputarrow] (2) to (B);
		\draw [style=inputarrow] (3) to (B);
		\draw [style=inputarrow] (4) to (C);
		\draw [style=simple] (Xtl.center) to (Xtr.center);
		\draw [style=simple] (Xtr.center) to (Xbr.center);
		\draw [style=simple] (Xbr.center) to (Xbl.center);
		\draw [style=simple] (Xbl.center) to (Xtl.center);
		\draw [style=simple] (Ytl.center) to (Ytr.center);
		\draw [style=simple] (Ytr.center) to (Ybr.center);
		\draw [style=simple] (Ybr.center) to (Ybl.center);
		\draw [style=simple] (Ybl.center) to (Ytl.center);
	\end{pgfonlayer}
\end{tikzpicture}
\]
Gray-boxing this gives a morphism in $\Dynam$, say $F \maps X \to Y$, represented by the
open dynamical system
\[         (X \stackrel{i}\longrightarrow S \stackrel{o}\longleftarrow Y, v^R) \]
where the cospan is visible in the figure and $v^R$ is the vector field on $\R^S$ given in
Equation (\ref{eq_v_1}).  If we now black-box $F$, we obtain the relation
\[   \blacksquare(F) = \{  (i^*(c),I,o^*(c),O) : \; v^R(c) + i_*(I) - i_*(O) = 0 \}. \] 
Here the inflows and outflows are
\[   I = (I_1, I_2, I_3) \in \R^X, \qquad O = O_4 \in \R^Y, \]
and vector of concentrations is $c = (A,B,C) \in \R^S$, so the input and output 
concentrations are 
\[   i^*(c) = (A,B,B) \in \R^X, \qquad o^*(c) = C \in \R^Y .\]
To find steady states with inflows $I$ and outflows $O$ we take the open rate equation, Equation (\ref{eq:open_rate_1}), and set all concentrations, inflows and outflows to constants:
\[
\begin{array}{rcr} 
 I_1 &=& r(\alpha) AB\, \\ 
 I_2 + I_3 &=& r(\alpha) AB\, \\ 
 O_4 &=& 2r(\alpha) AB.
\end{array}
\]
Thus, 
\begin{equation}   
\label{eq:black_1}
\blacksquare(F) = 
\end{equation}
\[ \{(A,B,B,I_1,I_2,I_3,C,O_4): \; I_1 = I_2 + I_3 = r(\alpha) AB , 
 O_4 = 2r(\alpha) AB \}.\]

Next we considered this open reaction network with rates:
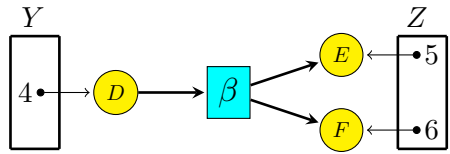
\[
\begin{tikzpicture}
	\begin{pgfonlayer}{nodelayer}
		\node [style = species] (D) at (1, 0) {$D$};
		\node [style = transition] (b) at (2.5, 0) {$\beta$};
		\node [style = species] (E) at (4,0.5) {$E$};
		\node [style = species] (F) at (4,-0.5) {$F$};

		\node [style=empty] (Y) at (-0.1, 1) {$Y$};
		\node [style=none] (Ytr) at (.25, 0.75) {};
		\node [style=none] (Ytl) at (-.4, 0.75) {};
		\node [style=none] (Ybr) at (.25, -0.75) {};
		\node [style=none] (Ybl) at (-.4, -0.75) {};

		\node [style=inputdot] (4) at (0, 0) {};
		\node [style=empty] at (-0.2, 0) {$4$};
		
		\node [style=empty] (Z) at (5, 1) {$Z$};
		\node [style=none] (Ztr) at (4.75, 0.75) {};
		\node [style=none] (Ztl) at (5.4, 0.75) {};
		\node [style=none] (Zbl) at (5.4, -0.75) {};
		\node [style=none] (Zbr) at (4.75, -0.75) {};

		\node [style=inputdot] (5) at (5, 0.5) {};
		\node [style=empty] at (5.2, 0.5) {$5$};	
		\node [style=inputdot] (6) at (5, -0.5) {};
		\node [style=empty] at (5.2, -0.5) {$6$};	

	\end{pgfonlayer}
	\begin{pgfonlayer}{edgelayer}
		\draw [style=inarrow] (D) to (b);
		\draw [style=inarrow] (b) to (E);
		\draw [style=inarrow] (b) to (F);
		\draw [style=inputarrow] (4) to (D);
		\draw [style=inputarrow] (5) to (E);
		\draw [style=inputarrow] (6) to (F);
		\draw [style=simple] (Ytl.center) to (Ytr.center);
		\draw [style=simple] (Ytr.center) to (Ybr.center);
		\draw [style=simple] (Ybr.center) to (Ybl.center);
		\draw [style=simple] (Ybl.center) to (Ytl.center);
		\draw [style=simple] (Ztl.center) to (Ztr.center);
		\draw [style=simple] (Ztr.center) to (Zbr.center);
		\draw [style=simple] (Zbr.center) to (Zbl.center);
		\draw [style=simple] (Zbl.center) to (Ztl.center);
	\end{pgfonlayer}
\end{tikzpicture}
\]
Gray-boxing this gives a morphism $F' \maps X \to Y$ in $\Dynam$ represented by the open dynamical system 
\[         (Y \stackrel{i'}\longrightarrow S' \stackrel{o'}\longleftarrow Y, v^{R'}) \]
where $v^{R'}$ is given by Equation (\ref{eq:v_2}).   To black-box $F'$ we can follow the same
procedure as for $F$.  We take the open rate equation, Equation (\ref{eq:open_rate_2}), and look
for steady-state solutions:
\[
\begin{array}{rcr} 
I_4 &=& r(\beta) D\,  \\
O_5 &=& r(\beta) D\, \\
O_6 &=& r(\beta) D.
\end{array}
\]
Then we form the relation between input concentrations, inflows, output concentrations and outflows
that holds in steady state:
\begin{equation}
\label{eq:black_2}
\blacksquare(F') = 
 \{ (D,I_4,E,F,O_5,O_6) : \; I_4 = O_5 = O_6 = r(\beta) D \} .  
\end{equation}

Finally, we can compose these two open reaction networks with rates:
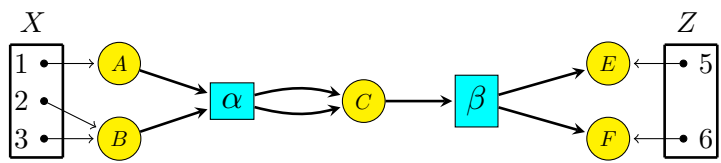
\[
\begin{tikzpicture}
	\begin{pgfonlayer}{nodelayer}
		\node [style=inputdot] (0) at (-4.25, 0) {};
		\node [style=empty] at (-4.55, 0) {$2$};
		\node [style=species] (1) at (-3.25, 0.5) {$A$};
		\node [style=none] (2) at (-4, 0.75) {};
		\node [style=none] (3) at (4, -0.75) {};
		\node [style=transition] (4) at (-1.75, -0) {$\alpha$};
		\node [style=none] (5) at (-4.7, 0.75) {};
		\node [style=none] (6) at (4, 0.75) {};
		\node [style=transition] (7) at (1.5, -0) {$\beta$};
		\node [style=inputdot] (8) at (4.25, 0.5) {};
		\node [style=empty] at (4.55, 0.5) {$5$};
		\node [style=none] (9) at (-4.7, -0.75) {};
		\node [style=species] (10) at (0, -0) {$C$};
		\node [style=none] (11) at (4.7, -0.75) {};
		\node [style=inputdot] (12) at (-4.25, 0.5) {};
		\node [style=empty] at (-4.55, 0.5) {$1$};
		\node [style=none] (13) at (4.7, 0.75) {};
		\node [style=empty] (14) at (-4.4, 1.05) {$X$};
		\node [style=none] (15) at (-4, -0.75) {};
		\node [style=empty] (16) at (4.3, 1.05) {$Z$};
		\node [style=species] (17) at (3.25, 0.5) {$E$};
		\node [style=species] (18) at (-3.25, -0.5) {$B$};
		\node [style=inputdot] (19) at (-4.25, -0.5) {};
		\node [style=empty] at (-4.55, -0.5) {$3$};
		\node[ style=species] (20) at (3.25, -0.5) {$F$};
		\node[ style=inputdot] (21) at (4.25, -0.5) {};
		\node [style=empty] at (4.55, -0.5) {$6$};
	\end{pgfonlayer}
	\begin{pgfonlayer}{edgelayer}
		\draw [style=inarrow] (1) to (4);
		\draw [style=inarrow] (18) to (4);
		\draw [style=inarrow, bend right=15, looseness=1.00] (4) to (10);
		\draw [style=inarrow] (7) to (17);
		\draw [style=inarrow, bend left=15, looseness=1.00] (4) to (10);
		\draw [style=inputarrow] (12) to (1);
		\draw [style=inputarrow] (0) to (18);
		\draw [style=inputarrow] (19) to (18);
		\draw [style=inputarrow] (8) to (17);
		\draw [style=simple] (5.center) to (2.center);
		\draw [style=simple] (2.center) to (15.center);
		\draw [style=simple] (15.center) to (9.center);
		\draw [style=simple] (9.center) to (5.center);
		\draw [style=simple] (6.center) to (13.center);
		\draw [style=simple] (13.center) to (11.center);
		\draw [style=simple] (11.center) to (3.center);
		\draw [style=simple] (3.center) to (6.center);
		\draw [style=inarrow] (10) to (7);
		\draw [style=inarrow]  (7) to (20);
		\draw[ style=inputarrow] (21) to (20);
	\end{pgfonlayer}
\end{tikzpicture}
\]
Gray-boxing the composite gives a morphism $F'F \maps X \to Y$ represented by the open dynamical system
\[         (Y \stackrel{i'}\longrightarrow S' \stackrel{o'}\longleftarrow Y, v^{R'R}) \]
where $v^{R'R}$ is given in Equation (\ref{eq:v_3}).   To black-box $F'F$ we take its open rate equation, Equation (\ref{eq:open_rate_3}), and look for steady state solutions:
\[
\begin{array}{rcl} 
I_1 &=& r(\alpha) AB \\
I_2 + I_3 &=& r(\alpha) AB \\ 
r(\beta) C &=& 2r(\alpha) AB  \\
O_5 &=& r(\beta) C \\
O_6 &=& r(\beta) C .
\end{array}
\]
The concentrations of internal species play only an indirect role after we black-box an open
dynamical system, since black-boxing only tells us the steady state relation between input concentrations, inflows, output concentrations and outflows.  In $F$ and $F'$ there were no internal species.  In $F'F$ there is one, namely $C$.  However, in this particular example the concentration $C$ is completely determined by the other data, so we can eliminate it from the above equations.  This is not true in every example.  But we can take advantage of this special feature here, obtaining these equations:
\[
\begin{array}{rcl} 
I_1 &=& r(\alpha) AB \\
I_2 + I_3 &=& r(\alpha) AB \\ 
O_5 &=& 2r(\alpha) AB  \\
O_5 &=& O_6 .
\end{array}
\]
We thus obtain
\begin{equation}
\label{eq:black_3}
\blacksquare(F'F)= 
\end{equation}
\[  \{ (A,B,B,I_1,I_2,I_3,E,F,O_5,O_6) : \; I_1  = I_2 + I_3 = r(\alpha) AB, O_5 = O_6 = 2 r(\alpha) AB \}   .\]
We leave it to the reader to finish checking the functoriality of black-boxing in this example:
\[         \blacksquare(F' F) = \blacksquare(F') \blacksquare(F) .\]
To do this, it suffices to compose the relations $\blacksquare(F)$ given in Equation (\ref{eq:black_1}) and $\blacksquare(F')$ given in Equation (\ref{eq:black_2}).

This example was a bit degenerate, because in each open dynamical system considered there was at most one steady state compatible with any choice of input concentrations, inflows, output concentrations and outflows.   In other words, even when there was an internal species, its concentration was determined by this `boundary data'.    This is far from generally true! 
Even for relatively simple `closed' reaction networks, namely those with no boundary species, multiple steady states may be possible.  Such reaction networks often involve features such as `autocatalysis', meaning that a certain species is present as both an input and an output to the same reaction.   
We expect the study of open reaction networks to give a new outlook on these questions.  However, our proof of the functoriality of black-boxing sidesteps this issue.

Before proving this result, it is nice to refine the framework slightly.  The black-boxing of an open dynamical system is far from an arbitrary relation: it is always `semialgebraic'. To understand this, we need a lightning review of semialgebraic geometry \cite{Coste}.

Let us use `vector space' to mean a finite-dimensional real vector space.  Given a vector space $V$, the collection of \define{semialgebraic subsets} of $V$ is the smallest collection that contains all sets of the form $\{P(v) = 0\}$ and $\{P(v) > 0\}$, where $P \maps V \to \R$ is any polynomial, and is closed under finite intersections, finite unions and complements.   The Tarski--Seidenberg theorem says that if $S \subseteq V \oplus W$ is semialgebraic then so is its projection to $V$, that is, the subset of $V$ given by
\[                \{v \in V :\; \exists w \in W \; (v,w) \in S \} .\]

If $U$ and $V$ are vector spaces, a \define{semialgebraic relation} $A \maps U \relto V$ is a semialgebraic subset $A \subseteq U \oplus V$.    If $A \maps U \relto V$ and $B \maps V \relto W$ are semialgebraic relations, so is their composite
\[      B \circ A = \{(u,w) : \; \exists v \in V \; (u,v) \in A \textrm{ and } (v,w) \in B \} \]
thanks to the Tarski--Seidenberg theorem.   The identity relation on any vector space is
also semialgebraic, so we obtain a category:

\begin{defn}
Let $\SemialgRel$ be the category with vector spaces as objects and semialgebraic relations as morphisms.
\end{defn}

We can now state the main theorem about black-boxing:

\begin{thm}
\label{thm:black}
There is a symmetric monoidal functor $\blacksquare \maps \Dynam \to \SemialgRel$
sending any finite set $X$ to the vector space $\R^X \oplus \R^X$ and any morphism $F \maps X \to Y$ to its black-boxing $\blacksquare(F)$.
\end{thm}

\begin{proof}
For any morphism $F \maps X \to Y$ in $\Dynam$ represented by the open dynamical system 
\[         (X \stackrel{i}\longrightarrow S \stackrel{o}\longleftarrow Y, v) \]
the set 
\[       \{(c, i^*(c),I,o^*(c),O) : \; v(c) + i_*(I) - i_*(O) = 0 \} \subseteq
\R^S \oplus \R^X \oplus \R^X \oplus \R^Y \oplus \R^Y \] 
is defined by polynomial equations, since $v$ is algebraic.   Thus, by the Tarski--Seidenberg theorem, the set 
\[    \blacksquare(F) = \{  (i^*(c),I,o^*(c),O) : \; v^R(c) + i_*(I) - i_*(O) = 0 \} \]
is semialgebraic.   

Next we prove that $\blacksquare$ is a functor.  Consider composable morphisms $F \maps X \to Y$ and $F' \maps Y \to Z$ in $\Dynam$.   We know that $F$ is represented by some open dynamical system 
\[         (X \stackrel{i}\longrightarrow S \stackrel{o}\longleftarrow Y, v) \]
while $F'$ is represented by some
\[         (Y \stackrel{i'}\longrightarrow S' \stackrel{o'}\longleftarrow Z, v') .\]
To compose these, we form the pushout
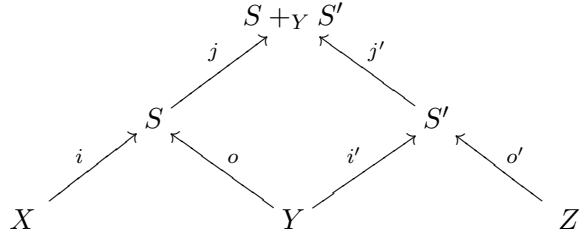
\[
    \xymatrix{
      && S +_Y S' \\
      & S \ar[ur]^{j} && S' \ar[ul]_{j'} \\
      \quad X\quad \ar[ur]^{i} && Y \ar[ul]_{o} \ar[ur]^{i'} &&\quad Z \quad \ar[ul]_{o'}
    }
\]
Then $F'F \maps X \to Z$ is represented by the open dynamical system
\[ (X \stackrel{j i}{\longrightarrow} S +_Y S' \stackrel{j' o'}{\longrightarrow} Z, u ) \]
where 
\[    u = j_* \circ v \circ j^* + {j'}_* \circ v' \circ {j'}^*  .\]

To prove that $\blacksquare$ is a functor, we first show that 
\[ \blacksquare(F'F) \subseteq \blacksquare(F') \blacksquare(F) \]  
Thus, given
\[     (i^*(c),I,o^*(c),O) \in \blacksquare(F), \qquad  ({i'}^*(c'),I',{o'}^*(c'),O') \in \blacksquare(F') \]
with 
\[   o^*(c) = {i'}^*(c'), \qquad O = I' \]
we need to prove that 
\[     (i^*(c),I,{o'}^*(c'),O') \in \blacksquare(F'F) \]
To do this, it suffices to find concentrations $b \in \R^{S +_Y S'}$ such that 
\[   (i^*(c),I,{o'}^*(c'),O') = ((ji)^*(b), I, {(j'o')}^*(b), O') \]
and $b$ is a steady state of $F'F$ with inflows $I$ and outflows $O'$.

Since  $o^*(c) = {i'}^*(c'),$ this diagram commutes:
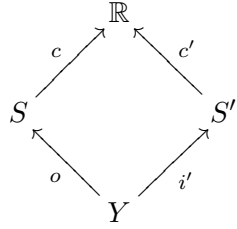
\[
    \xymatrix{
      && \R \\
      & S \ar[ur]^{c} && S' \ar[ul]_{c'} \\
       && Y \ar[ul]^{o} \ar[ur]_{i'} &&
    }
\]
so by the universal property of the pushout there is a unique map $b \maps S +_Y S' \to \R$ such that
this commutes:
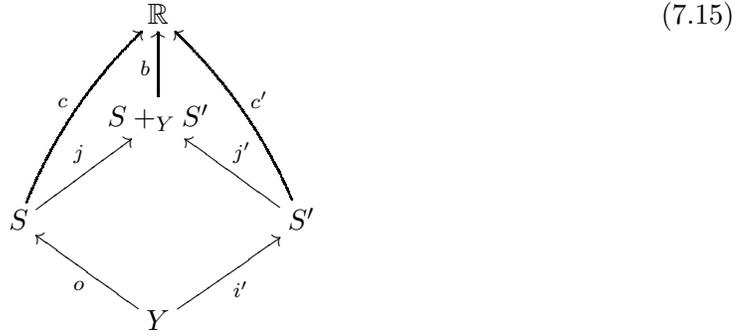
\begin{equation}
\label{eq:pushout}
    \xymatrix{
      && \R \\
     && S +_Y  S' \ar[u]^b \\      
      & S \ar@/^/[uur]^{c} \ar[ur]^{j} && S' \ar@/_/[uul]_{c'} \ar[ul]_{j'} \\
       && Y \ar[ul]^{o} \ar[ur]_{i'} &&
    }
\end{equation}
This simply says that because the concentrations $c$ and $c'$ agree on the `overlap' of our two
open dynamical systems, we can find a concentration $b$ for the composite system that restricts
to $c$ on $S$ and $c'$ on $S'$.

We now prove that $b$ is a steady state of the composite open dynamical system with
inflows $I$ and outflows $O'$:
\begin{equation}
\label{eq:steady_state_3}
   u(b) + (ji)_*(I) - (j'o')_*(O') = 0.
\end{equation}
To do this we use the fact that $c$ is a steady state of $F$ with inflows $I$ and outflows $O$:
\begin{equation}
\label{eq:steady_state_1}
   v(c) + i_*(I) - {o}_*(O) = 0
\end{equation}
and $c'$ is a steady state of $F'$ with inflows $I'$ and outflows $O'$:
\begin{equation}
\label{eq:steady_state_2}
   v'(c') + {i'}_*(I') - {o'}_*(O') = 0.
\end{equation}
We push forward Equation (\ref{eq:steady_state_1}) along $j$, push forward Equation
(\ref{eq:steady_state_2}) along $j'$, and sum them:
\[   j_*(v(c))  + (ji)_*(I) - (jo)_*(O) + j'_*(v'(c')) + (j'i')_*(I') - (j'o')_*(O') = 0. \]
Since $O = I'$ and $jo = j'i'$, two terms cancel, leaving us with
\[     j_*(v(c))  + (ji)_*(I) + j'_*(v'(c')) - (j'o')_*(O') = 0. \]
Next we combine the terms involving the vector fields $v$ and $v'$, with the help of Equation (\ref{eq:pushout}) and the definition of $u$:
\begin{equation}
\label{eq:u}
   \begin{array}{ccl}
  j_*(v(c)) + j'_*(v'(c')) &=& j_*(v(b \circ j)) + j'_*(v'(b \circ j')) \\
                                    &=& (j_* \circ v \circ j^* + j'_* \circ v' \circ j'^*)(b) \\
                                    &=& u(b)  .
\end{array}
\end{equation}
This leaves us with
\[         u(b) +  (ji)_*(I) - (j'o')_*(O') = 0 \]
which is Equation (\ref{eq:steady_state_3}), precisely what we needed to show.

To finish showing that $\blacksquare$ is a functor, we need to show that 
\[   \blacksquare(F'F) \subseteq \blacksquare(F') \blacksquare(F)  .\] 
So, suppose we have 
\[    ((ji)^*(b), I, {(j'o')}^*(b), O') \in \blacksquare(F'F) .\]
We need to show
\begin{equation}
\label{eq:composite}
  ((ji)^*(b), I, {(j'o')}^*(b), O') = (i^*(c),I,{o'}^*(c'),O') 
\end{equation}
where 
\[     (i^*(c),I,o^*(c),O) \in \blacksquare(F), \qquad  ({i'}^*(c'),I',{o'}^*(c'),O') \in \blacksquare(F') \]
and
\[   o^*(c) = {i'}^*(c'), \qquad O = I' .\]

To do this, we begin by choosing
\[   c = j^*(b), \qquad c' = {j'}^*(b) .\]
This ensures that Equation (\ref{eq:composite}) holds, and since $jo = j'i'$, it also ensures that 
\[  o^*(c) = (jo)^*(b) = (j'i')^*(b) = {i'}^*(c')  .\]
So, to finish the job, we only need to find an element $O = I' \in \R^Y$ such that $c$ is a steady state of $F$ with inflows $I$ and outflows $O$ and $c'$ is a steady state of $F'$ with inflows $I'$ and outflows $O'$.  Of course, we are given the fact that $b$ is a steady state of $F'F$ with inflows $I$ and outflows $O'$.   

In short, we are given Equation (\ref{eq:steady_state_3}), and we want to find $O = I'$ such that Equations (\ref{eq:steady_state_1}) and (\ref{eq:steady_state_2}) hold.  Thanks to our choices of $c$ and $c'$,  we can use Equation (\ref{eq:u}) and rewrite Equation (\ref{eq:steady_state_3}) as
\begin{equation}
\label{eq:steady_state_3'}
  j_*(v(c) + i_*(I)) \; + \; {j'}_*(v'(c') - {o'}_*(O')) = 0 .  
\end{equation}
Equations  (\ref{eq:steady_state_1}) and (\ref{eq:steady_state_2}) say that
\begin{equation}
\label{eq:steady_state_1'2'}
\begin{array}{lcl}
   v(c) + i_*(I) - {o}_*(O) &=& 0 \\  \\
   v'(c') + {i'}_*(I') - {o'}_*(O') &=& 0.
\end{array}
\end{equation}

Now we use the fact that 
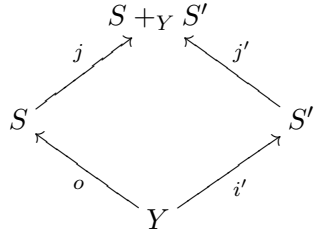
\[
    \xymatrix{
      & S +_Y S' \\
       S \ar[ur]^{j} && S' \ar[ul]_{j'} \\
       & Y \ar[ul]^{o} \ar[ur]_{i'} &
    }
\]
is a pushout.  Applying the `free vector space on a finite set' functor, which preserves colimits, this implies that
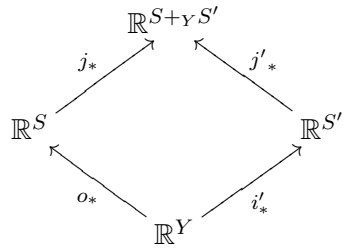
\[
    \xymatrix{
      & \R^{S +_Y S'} \\
       \R^S \ar[ur]^{j_*} && \R^{S'} \ar[ul]_{{j'}_*} \\
       & \R^Y \ar[ul]^{o_*} \ar[ur]_{i'_*} &
    }
\]
is a pushout in the category of vector spaces.   Since a pushout is formed by taking first a coproduct and then a coequalizer, this implies that 
\[
     \xymatrix{
      \R^Y \ar@<-.5ex>[rr]_-{(0,i'_*)} \ar@<.5ex>[rr]^-{(o_*,0)} && \R^S \oplus \R^{S'} \ar[rr]^{j_* + j'_*}
   && \R^{S +_Y S'}
}
\]
is a coequalizer.  Thus, the kernel of $j_* + j'_*$ is the image of $(o_*,0) - (0,i'_*)$.   Equation (\ref{eq:steady_state_3'}) says precisely that 
\[    (v(c) + i_*(I), v'(c') - o'_*(O')) \in \ker(j_* + j'_*)  .\]
Thus, it is in the image of $o_* - i'_*$.  In other words, there exists some element $O = I' \in \R^Y$
such that 
\[   (v(c) + i_*(I), v'(c') - o'_*(O')) = (o_*(O), -i'_*(I')).\]
This says that Equations (\ref{eq:steady_state_1}) and (\ref{eq:steady_state_2}) hold, as desired.

Finally, we need to check that $\blacksquare$ is symmetric monoidal.  But this is a straightforward calculation, so we leave it to the reader.
\end{proof}

It is worth comparing our black-boxing theorem, Theorem \ref{thm:black}, to Spivak's work on open dynamical systems \cite{Spivak}.  He describes various  categories where the morphisms are open dynamical systems, and constructs functors from these categories to $\Rel$, whcih describe the steady state relations between inputs and outputs.  None of his results subsume ours, but they are philosophically very close.  Both are doubtless special cases of a more general theorem that is yet to be formulated. It would be interesting to connect this line of work with recent results on the thermodynamics of open chemical reaction networks which connects the existence and interpretation of various thermodynamic quantities with topological properties of the open reaction network \cite{PoletinniCRN}.
\chapter{Conclusions}\label{ch:conc}

At the graphical level, Petri Nets or reaction networks provide a very general syntax, utilized not only to represent sets of coupled non-linear differential equations, but also to reason about models of concurrent computing and distributed systems \cite{modgenpet}. The category $\RNet$ provides a new framework for the construction of open versions of such Petri Nets and likely admits many interesting assignments of semantics or behavior.

Restricting to single-reactant to single-product transitions, reaction networks reduce to Markov processes as labeled graphs:

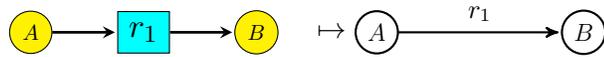
\[ \begin{tikzpicture}
	\begin{pgfonlayer}{nodelayer}
	       \node[style=none] (1) at (0,0) {$ \mapsto $ };
		\node [style=species] (0) at (-1, 0) {$B$};
		\node [style=transition] (7) at (-2.5, 0) {$r_1$};
		\node [style=species] (18) at (-4, 0) {$A$};
	\end{pgfonlayer}
	\begin{pgfonlayer}{edgelayer}
		\draw [style=inarrow] (18) to (7);
		\draw [style=inarrow] (7) to (0);
	\end{pgfonlayer}

\begin{scope}[shift={(5,0)},->,>=stealth',shorten >=1pt,thick,scale=1.1]
\node[main node, scale=.65](A) at (-4,0) {$A$};
\node[main node, scale=.65](B) at (-1.5,0) {$B$};

  \path[every node/.style={font=\sffamily\small}, shorten >=1pt]
    (A) edge [] node[above] {$r_1$} (B);
\end{scope}
\end{tikzpicture}
\]
In such a case, the rate equation
\[ \frac{dc}{dt} = \sum_{\tau \in T} r(\tau) (t(\tau) - s(\tau)) c(t)^{s(\tau)} \]
also reduces to the master equation
\[ \frac{dc_i}{dt} = \sum_j \left( H_{ij}c_j - H_{ji}c_i \right) \]
where the concentrations in a reaction network play the role of the probabilities in a Markov process.

Since each transition has a single input and a single output, under the correspondence $T \equiv E$ we can write
\[ \frac{dc}{dt} = \sum_{e \in E} r_e( t(e) - s(e)) c(t)^{s(e)}. \]
Taking a particular component, we have
\[ \frac{dc_i}{dt} = \sum_{e \in E} r_e( t(e)_i - s(e)_i ) c(t)^{s(e)} \] 
\[ \frac{dc_i}{dt} = \sum_j \left( \sum_{e \maps j \to i} r_e c_j - \sum_{e \maps i \to j} r_e c_i \right). \]
Recalling that for Markov processes, $H_{ij} = \sum_{e \maps j \to i} r_e$ we can write this as
\[ \frac{dc_i}{dt} = \sum_j \left( H_{ij}c_j - H_{ji}c_i \right) \]
which we recognize as the master equation.

At the level of categories, this means that we have an inclusion functor
\[ I \maps \OpenMark \to \RxNet \]
and that composing it with the functor
\[ \graysquare \maps \RxNet \to \Dynam \]
and applying the composite to an open Markov process yields the open master equation.

Our approach to black-boxing open Markov processes and open reaction networks differed substantially. To tackle open Markov processes we first restricted our attention to a subcategory $\DetBalMark$ of open detailed balanced Markov processes. We did this essentially because it is possible to start with two open Markov processes whose underlying Markov processes both satisfy Kolmogorov's criterion and compose them to give a process whose underlying Markov process violates Kolmogorov's criterion. By adding additional structure to the inputs and outputs of an open Markov process, namely by labeling the states with energies, we arrived at a notion of composition which preserved the property of admitting a detailed balanced equilibrium. Since every open detailed balanced Markov process has an underlying open Markov process, there is a forgetful functor 
\[ F \maps \DetBalMark \to \OpenMark . \] 
We then showed that there is a functor
\[ K \maps \DetBalMark \to \Circ \]
sending an open detailed balanced Markov process to a corresponding electrical circuit. 
Utilizing the fact that non-equilibrium steady states minimize dissipation in an open Markov process together with the existing result that there is a black-boxing functor for electrical circuits \cite{BaezFongCirc}
\[ \vdarkgraysquare \maps \Circ \to \LinRel\]
we arrived at a functor
\[ \square \maps \DetBalMark \to \LinRel \] 
sending any open detailed balanced Markov process to the subspace of possible steady state boundary probabilities and probability flows, viewed as a linear relation. A key step in this construction is that the steady-states of the open Markov processes we were black-boxing obeyed a variational principle provided by the minimization of the dissipation, which we saw approximates the rate of entropy production for steady-states near the underlying detailed balanced equilibrium. 

For reaction networks we took a more general approach by first showing that there is a category of open dynamical systems and a functor
\[ \graysquare \maps \RxNet \to \Dynam \]
sending an open reaction network to its open reaction network viewed as a morphism in the decorated cospan category $\Dynam$.
Composing this with
\[ \blacksquare \maps \Dynam \to \SemialgRel, \] 
yields a functor sending an open dynamical system to the semialgebraic relation corresponding to possible steady state boundary concentrations and flows.  Every linear relation is in fact semialgebraic, meaning that there is a functor
\[ U \maps \LinRel \to \SemialgRel. \]

We can summarize the above results involving various categories and functors between them with a single diagram:
\[ 
\xymatrix{
\DetBalMark \ar[dd]_K \ar[r]^-F \ar[ddr]_{\square} & \Mark \ar[r]^I  \ar[dd]^{\darkgraysquare} & \RNet  \ar[d]^{\graysquare} \\ 
& &   \Dynam \ar[d]^{\blacksquare} \\
\Circ \ar[r]_{\vdarkgraysquare} & \LinRel \ar[r]_-U & \SemialgRel 
}
\]
This diagram commutes up to natural transformation. 

Our approach to black-boxing open reaction networks applies to any open reaction network and therefore any Markov process, even those which do admit a detailed balanced equilibrium. Thus this approach is more general and makes no use of a variational principle. However, we saw that relative entropy serves as a Lyapunov function for Markov processes
\[ \frac{d}{dt} I(p(t),q) \leq 0 \]
and that relative entropy approximates dissipation for steady-states near equilibrium $\epsilon_i = 1+\frac{p_i}{q_i}$
\[ \frac{d}{dt} I(p(t),q) \approx -D(p)+ O(\epsilon^2).\] 
In addition, we saw that the open master equation can be written as a gradient flow with dissipation serving as a potential function
\[ \frac{dp}{dt} = -\nabla D(p) \]
on a space where the gradient $\nabla$ is given by $\nabla_i = \frac{q_i}{2} \frac{\partial}{\partial p_i}$.

In chemical reaction network theory, there is a class of reaction networks which admit a particularly nice equilibrium state called a `complex balanced equilibrium.' A reaction network admitting such an equilibrium is called a complex balanced reaction network. The existence of such an equilibrium again amounts to a condition on the structure and the rates of the reaction network. This condition guarantees that there exist values of the concentrations for which if a complex is annihilated at some net rate in the network, then it is created elsewhere at the same net rate, and vice versa. For detailed balanced Markov processes, the entropy relative to the equilibrium state serves as a Lyapunov function. In such a situation this relative entropy is in fact the difference in free energy of a non-equilibrium steady state from the equilibrium free energy. For complex balanced reaction networks, there exists a Lyapunov function corresponding to a type of free energy \cite{Horn}. In addition, mass-action kinetics for complex balanced reaction networks can be written as a gradient flow involving a certain related potential function \cite{Mielke}. 

Thus while not included presently there is likely a category of open complex balanced reaction networks and a corresponding method of black-boxing such networks which does involve a variational principle. The morphisms in the category of open complex balanced reaction networks would all correspond to morphisms in $\RNet$, however in order to ensure that composition of open complex balanced reaction networks resulted in an open complex balanced reaction network, the objects require more structure than that of finite sets. This is analogous to our assignment of energies to states in a Markov process order to close the category of open detailed balanced Markov processes under the operation of composition. An analogous assignment for complex balanced reaction networks would relate the existence of a free energy serving as a Lyapunov function to that of a potential function whose gradient on certain space generates the rate equation. This would be a natural next step in this line of work.

\nocite{*}
 \bibliographystyle{alpha}

\begin{thebibliography}{9}

\bibitem{AC} S.\ Abramsky and B.\ Coecke, A categorical semantics of quantum protocols, in {\sl Proceedings of the 19th IEEE Conference on Logic in Computer Science (LiCS04)}, IEEE Computer Science Press, 2004, 415--425. Available as \href{http://arxiv.org/abs/quant-ph/0402130}{arXiv:quant-ph/0402130}.

\bibitem{alberti1982stochasticity} P.\ M.\ Alberti and A.\ Uhlmann, \textsl{Stochasticity and Partial Order: Doubly Stochastic Maps and Unitary Mixing}, D.\ Reidel: Dordrecht, 1982.

\bibitem{AndrieuxGaspard} D.\ Andrieux, and P.\ Gaspard, Fluctuation theorem for currents and Schnakenberg network theory, \textsl{ J.\ Stat.\ Phys.}, 127,  2007, 107--131.

\bibitem{BaezBiamonte} J.\ C.\ Baez and J.\ D.\ Biamonte, Quantum techniques for stochastic mechanics. Available as \href{http://arxiv.org/abs/1209.3632}{arXiv:1209.3632}.

\bibitem{BaezCoyaRebro} J.\ C.\ Baez, B.\ Coya, and F.\ Rebro, Props in network theory. Available as \href{https://arxiv.org/abs/1707.08321}{arXiv:1707.08321}.

\bibitem{BaezEberleControl} J.\ C.\ Baez and J.\ Eberle, Categories in control, \textsl{ Theory Appl.\ Cat.}, 30, 2015, 836--881. Available as \href{https://arxiv.org/abs/1405.6881}{arXiv:1405.6881}.

\bibitem{BaezFongCirc} J.\ C.\ Baez and B.\ Fong, A compositional framework for passive linear networks. Available as \href{http://arxiv.org/abs/1504.05625}{arXiv:1504.05625}.

\bibitem{BaezFongP} J.\ C.\ Baez, B.\ Fong, and B.\ S.\ Pollard, A compositional framework for open Markov processes, \textsl{ Jour.\ Math.\ Phys.}, 57, 2016, 033301. Available as \href{http://arxiv.org/abs/1508.06448}{arXiv:1508.06448}.


\bibitem{BaezPRx} J.\ C.\ Baez and B.\ S.\ Pollard, A compositional framework for reaction networks, \textsl{ Rev.\ Math.\ Phys.}, 29, 2017, 1750028. Available as \href{https://arxiv.org/abs/1704.02051}{arXiv:1704.02051}.

\bibitem{BaezStay} J.\ C.\ Baez and M.\ Stay,  Physics, topology, logic and computation: a Rosetta Stone, in {\sl New Structures for Physics}, ed.\ B.\ Coecke, Lecture Notes in Physics, 813, Springer: Berlin, 2011, 173--286.  Available as \href{http://arxiv.org/abs/0903.0340}{arXiv:0903.0340}.


\bibitem{SRgraph2} M.\ Banaji and G.\ Craciun, Graph-theoretic criterion for injectivity and unique equilibria in general chemical reaction systems, \textsl{Adv.\ Appl.\ Math.}, 44, 2010, 168--184. Available as \href{http://arxiv.org/abs/0809.1308}{arXiv:0809.1308}.

\bibitem{Benabou} J.\ B\'enabou, Introduction to bicategories, in \textsl{Reports of the Midwest Category Seminar}, eds.\ J.\ B\'enabou et.\ al., Springer Lecture Notes in Mathematics, 47, Springer: New York, 1967, 1--77.

\bibitem{Bott} R.\ Bott and J.\ Mayberry, Matrices and trees, \textsl{ Economic Activity Analysis}, ed. O.\ Morgenstern, John Wiley \& Sons: New York, 1954,  391--400. 

\bibitem{BMN} S.\ Bruers, C.\ Maes, and K.\ Neto\v cn\'y, On the validity of entropy production principles for linear electrical circuits. \textsl{ Jour.\ Stat.\ Phys.},  2007, 129, 725--740.

\bibitem{cohen1993relative} J.\ E.\ Cohen, Y.\ Iwasa, G.\ Rautu, M.\ B.\ Ruskai, E.\ Seneta, and G.\ Zbaganu, Relative entropy under mappings by stochastic matrices, \textsl{Linear Algebra Appl.}, 179, 1993, 211--235.

\bibitem{cohenmajorization}  J.\ E.\ Cohen,  Y.\ Derriennic, and G.\ Zbaganu,  Majorization, monotonicity of relative entropy, and stochastic matrices, \textsl{Contemp.\ Math.}, 149, 1993, 251--259.

\bibitem{Coste} M.\ Coste, An introduction to semialgebraic geometry, RAAG Network School, 2002.  Available at \href{http://gcomte.perso.math.cnrs.fr/M2/CosteIntroToSemialGeo.pdf}{http://gcomte.perso.math.cnrs.fr/M2/CosteIntroToSemialGeo.pdf}.

\bibitem{Kenny} K.\ Courser, A bicategory of decorated cospans, to appear in \textsl{Theory Appl.\ Cat.}, 32, 2017, 985--1027. Available as \href{https://arxiv.org/abs/1605.08100}{arXiv:1605.08100}.

\bibitem{cover1994processes} T.\ M.\ Cover, Which processes satisfy the Second Law?, in \textsl{Physical Origins of Time Asymmetry}, eds.\ J.\ J.\ Halliwell, J.\ Perez-Mercader and W.\ H.\ Zurek, Cambridge University Press: New York, 1994, 98--107.

\bibitem{SRgraph} G.\ Craciun, Y.\ Tang, and M.\ Feinberg, Understanding bistability in complex enzyme-driven reaction networks, \textsl{Proc.\ Natl.\ Acad.\ Sci.\ USA}, 103, 2006, 8697--8702. Available at \href{http://www.pnas.org/content/103/23/8697.abstract}{http://www.pnas.org/content/103/23/8697.abstract}.

\bibitem{csiszar1963} I.\ Csisz{\'a}r, Eine informationstheoretische Ungleichung und ihre Anwendung auf den Beweis der Ergodizitat von Markoffschen Ketten, \textsl{Publ.\ Math.\ Inst.\ Hungar.\ Acad.\ Sci.}, 8, 1963, 85--108. 

\bibitem{dupuis2012construction} P.\ Dupuis and M.\ Fischer, On the construction of Lyapunov functions for nonlinear Markov processes via relative entropy, preprint 2012.

\bibitem{JasonThesis} J.\ M.\ Erbele, Categories in Control: Applied PROPs, Ph.D. thesis, University of California Riverside, 2016. Available as \href{https://arxiv.org/abs/1611.07591}{arXiv:1611.07591}.

\bibitem{Feinberg} M.\ Feinberg, Complex balancing in general kinetic systems, \textsl{Arch.\ Rational Mech.\ Anal.}, 49, 1972, 187--194.

\bibitem{def0} M.\ Feinberg, The existence and uniqueness of steady states for a class of chemical reaction networks, \textsl{Arch.\ Rational Mech.\ Anal.}, 132, 1995, 311--370.

\bibitem{def1} M.\ Feinberg, Multiple steady states for chemical reaction networks of deficiency one, \textsl{Arch.\ Rational Mech.\ Anal.}, 132, 1995, 371--406.

\bibitem{Horn} M.\ Feinberg and F.\ J.\ Horn, Chemical mechanism structure and the coincidence of the stochiometric and kinetic subspaces, \textsl{Arch.\ Rational Mech.\ Anal.}, 66, 1977, 83--97. 

\bibitem{Fong} B.\ Fong, Decorated cospans, \textsl{Theory Appl.\ Cat.}, 30, 2015, 1096--1120. Available as \href{https://arxiv.org/abs/1502.00872}{arXiv:1502.00872}.

\bibitem{FongThesis} B.\ Fong, \textsl{The Algebra of Open and Interconnected Systems}, Ph.D. thesis, University of Oxford, 2016.  Available as \href{https://arxiv.org/abs/1609.05382}{arXiv:1609.05382}.

\bibitem{Gardiner} C.\ W.\ Gardiner, \textsl{Handbook of Stochastic Methods: for Physics, Chemistry, and the Natural Sciences}, ed.\ H. Haken, Springer Series in Synergetics, 13, Springer: Berlin, 1985.

\bibitem{GP} P.\ Glandsorf and I.\ Prigogine, \textsl{ Thermodynamic Theory of Structure, Stability and Fluctuations}, Wiley-Interscience: New York, 1971.

\bibitem{gorban2010entropy} A.\ N.\ Gorban, P.\ A.\ Gorban, and G.\ Judge, Entropy: The Markov ordering approach, \textsl{Entropy}, 12, 2010, 1145--1193. Available as \href{https://arxiv.org/abs/1003.1377}{arXiv:1003.1377}.

\bibitem{StochPetri} P.\ J.\ E.\ Goss and J.\ Peccoud, Quantitative modeling of stochastic systems in molecular biology by using stochastic Petri nets, \textsl{Proc.\ Natl.\ Acad.\ Sci.\ USA}, 98, 1998, 6750--6755. Available at \href{http://www.pnas.org/content/95/12/6750.full.pdf}{http://www.pnas.org/content/95/12/6750.full.pdf}.

\bibitem{DeGrootM} S.\ R.\ de Groot and P.\ Mazur, \textsl{ Non-equilibrium Thermodynamics}, North-Holland Publishing Company: Amsterdam, 1962.

\bibitem{Haas} P.\ J.\ Haas, \textsl{Stochastic Petri Nets: Modelling, Stability, Simulation}, Springer: Berlin, 2002.

\bibitem{HillTree} T.\ L.\ Hill, Studies in irreversible thermodynamics IV: Diagrammatic representation of steady state fluxes for unimolecular systems, \textsl{ Jour.\ Theor.\ Bio.}, 10, 1966, 442--459. 

\bibitem{Hill} T.\ L.\ Hill, \textsl{ Free Energy Transduction in Biology: The Steady-State Kinetic and Thermodynamic Formalism}, Academic Press: New York, 1977.

\bibitem{HillDover} T.\ L.\ Hill, \textsl{ Free Energy Transduction and Biochemical Cycle Kinetics}, Springer-Verlag: New York, 1989, reprinted, Dover: New York, 2005.

\bibitem{HillScience} T.\ L.\ Hill and E.\ Eisenberg, Muscle contraction and free energy transduction in biological systems, \textsl{ Science}, 227, 1985, 999-1006.

\bibitem{Qians} D.\ Jiang, M.\ Qian, and M.\ P.\ Qian, \textsl{ Mathematical Theory of Nonequilibrium Steady States}, Springer: Berlin, 2004.

\bibitem{PSWalters} P.\ Katis, N.\ Sabadini, and R.\ F.\ C.\ Walters, Representing place/transition nets in Span(Graph), in \textsl{Proc. 5th AMAST Conf.}, LNCS, 1349, 1997, 322--336.  

\bibitem{Kelly} F.\ P.\ Kelly, \textsl{Reversibility and Stochastic Networks}, Wiley: Chichester, 1979,  reprinted, Cambridge University Press: New York, 2011.

\bibitem{King} E.\ L.\ King and C.\ Altman, A schematic method of deriving the rate laws for enzyme-catalyzed reactions, \textsl{J.\ Phys.\ Chem.\ }, 60, 1956, 1375--1378.

\bibitem{Kingman} J.\ F.\ C.\ Kingman, Markov population processes, \textsl{Jour.\ Appl.\ Prob.}, 6, 1969, 1--18.  

\bibitem{kullback1951information} S.\ Kullback and R.\ A.\ Leibler, On information and sufficiency, \textsl{Ann. Math. Statist.}, 22, 1951, 79--86.

\bibitem{Landauer} R.\ Landauer, Inadequacy of entropy and entropy derivatives in characterizing the steady state. \textsl{ Phys.\ Rev.\ A}, 12, 1975, 636--638.

\bibitem{Landauer2} R.\ Landauer, Stability and entropy production in electrical circuits. \textsl{Jour.\ Stat.\ Phys.}, 13, 1975, 1--16.

\bibitem{Lawvere}  F.\ W.\ Lawvere, Functorial semantics of algebraic theories and some algebraic problems in the context of functorial semantics of algebraic theories, 1963, \textsl{Reprints in Theory Appl.\ Categ.}, 5, 2005, 1--121. Available at \href{http://www.tac.mta.ca/tac/reprints/articles/5/tr5abs.html}{http://www.tac.mta.ca/tac/reprints/articles/5/tr5abs.html}.

\bibitem{LJ} G.\ Lebon and D.\ Jou, \textsl{Understanding Non-equilibrium Thermodynamics}, Springer: Berlin, 2008.

\bibitem{Leinster} T.\ Leinster, Basic bicategories. Available as \href{http://arxiv.org/abs/math/9810017}{math.CT/9810017}.
    
\bibitem{Lindblad} G.\ Lindblad, \textsl{Non-equilibrium Entropy and Irreversibility}, D. Reidel: Dordecht, Holland, 1983.

\bibitem{MacLane} S.\ Mac Lane, \textsl{Categories for the Working Mathematician}, Springer: Berlin, 1998.

 \bibitem{MN} C.\ Maes and K.\ Neto\v cn\'y, Minimum entropy production principle from a dynamical fluctuation law, \textsl{Jour.\ Math.\ Phys.}, 48, 2007, 053306.

\bibitem{modgenpet}  M.\ A.\ Marsan, G.\ Balbo, G.\ Conte, S.\ Donatelli, and G.\ Franceschinis, \textsl{Modelling with Generalized Stochastic Petri Nets},  John Wiley \& Sons: New York, 1994.

\bibitem{merhav2011data}  N.\ Merhav, Data processing theorems and the second law of thermodynamics, \textsl{IEEE Trans.\ Inform.\ Theory}, 5, 2011, 4926--4939.

\bibitem{MeseguerMontanari} J.\ Meseguer and U.\ Montanari, Petri nets are monoids, \textsl{Inform.\ and Comput.}, 88 1990, 105--155. 

\bibitem{Mielke} A.\ Mielke, A gradient structure for reaction-diffusion systems and for energy-drift diffusion systems, \textsl{Nonlinearity}, 24, 2011, 1329--1346. 

\bibitem{moran1961entropy} P.\ A.\ P.\ Moran, Entropy, Markov processes and Boltzmann's H-theorem, \textsl{Proc.\ Cambridge Philos.\ Soc.}, 57, 1961, 833--842.

\bibitem{morimoto1963markov}  T.\ Morimoto, Markov processes and the H-theorem, \textsl{J.\ Phys.\ Soc.\ Japan}, 18, 1963, 328--331. 


\bibitem{OPK} G.\ Oster, A.\ Perelson, and A.\ Katchalsky, Network thermodynamics, \textsl{Nature}, 234, 1971, 393--399.

\bibitem{OPKBio} G.\ Oster, A.\ Perelson, and A.\ Katchalsky, Network thermodynamics: dynamic modeling of biophysical systems, \textsl{Quart.\ Rev.\ Biophys.}, 1, 1973, 1--134.

\bibitem{OPChem} A.\ Perelson and G.\ Oster, Chemical reaction networks, \textsl{IEEE Trans.\ Circ.\ Sys.}, 21, 1974, 709--721.

\bibitem{PoletinniCRN} M.\ Poletinni, and M.\ Esposito, Irreversible thermodynamics of open chemical networks I: Emergent cycles and broken conservation laws, \textsl{J.\ Chem.\ Phys.}, 141, 2014, 024117.

\bibitem{Pollard} B.\ S.\ Pollard,  A Second Law for open Markov processes, \textsl{Open Syst.\ Inf.\ Dyn.}, 23, 2016, 1650006. Available as \href{http://arxiv.org/abs/1410.6531}{arXiv:1410.6531}.

\bibitem{PollardBio} B.\ S.\ Pollard, Open Markov processes: A compositional perspective on non-equilibrium steady states in biology, \textsl{Entropy}, 18, 2016, 140. Available as \href{https://arxiv.org/abs/1601.00711}{arXiv:1601.00711}.

\bibitem{Prigogine} I.\ Prigogine, \textsl{Non-Equilibrium Statistical Mechanics}, Interscience Publishers: New York, 1962.

\bibitem{PrigogineEnt} I.\ Prigogine, \textsl{Etud{\'e} Thermodynamique des ph{\'e}nom{\'e}nes irr{\'e}versibles}, Dunod: Paris and Desoer: Li{\'e}ge, 1947. 

\bibitem{Qian1} H.\ Qian, Open-system nonequilibrium steady state: statistical thermodynamics, fluctuations, and chemical oscillations, \textsl{J.\ Phys.\ Chem.\ B}, 31, 2006, 15063--74.

\bibitem{Qian2} H.\ Qian and D.\ A.\ Beard, Thermodynamics of stoichiometric biochemical networks in living systems far from equilibrium, \textsl{ Biophys.\ Chem.}, 114, 2005, 213--220.

\bibitem{Qian3} H.\ Qian and L.\ Bishop, The chemical master equation approach to nonequilibrium steady-state of open biochemical systems: Linear single-molecule enzyme kinetics and nonlinear biochemical reaction networks, \textsl{ Int.\ J.\ Mol.\ Sci.},  11, 2010, 3472--3500.

\bibitem{Sassone} V.\ Sassone, On the category of Petri net computations,  in \textsl{TAPSOFT'95: Proc. Intl. Joint Conference on Theory and Practice of Software Development}, LNCS, 915, 1995, 334--348. 

\bibitem{SchnakenRev} J.\ Schnakenberg, Network theory of microscopic and macroscopic behavior of master equation systems, \textsl{ Rev.\ Mod.\ Phys.}, 48, 1976, 571--585.

\bibitem{SchnakenBook} J.\ Schnakenberg, \textsl{ Thermodynamic Network Analysis of Biological Systems}, Springer: Berlin, 1981. 

\bibitem{Se} P.\ Selinger, Dagger compact closed categories and completely positive maps, in {\sl QPL2005: Proceedings of the 3rd International Workshop on Quantum Programming Languages}, ENTCS, 170, 2007, 139--163.  Available at \href{https://ncatlab.org/nlab/files/SelingerPositiveMaps.pdf}{https://ncatlab.org/nlab/files/SelingerPositiveMaps.pdf}.

\bibitem{concordance} G.\ Shinar and M.\ Feinberg, Concordant chemical reaction networks and the Species-Reaction Graph, \textsl{Math.\ Bio.}, 241, 2013, 1--23. 

\bibitem{SpivakWiring} D.\ Spivak, The operad of wiring diagrams: formalizing a graphical language for databases, recursion, and plug-and-play circuits. Available as \href{https://arxiv.org/abs/1305.0297}{arXiv:1305.0297}.

\bibitem{Spivak} D.\ Spivak, The steady states of coupled dynamical systems compose according to matrix arithmetic.  Available as \href{https://arxiv.org/abs/1512.00802}{arXiv:1512.00802}.

\bibitem{Tutte} W.\ Tutte, The dissection of equilateral triangles into equilateral triangles, \textsl{ Proc.\ Cambridge Philos.\ Soc.}, 44, 1948, 463--482.

\bibitem{VanKampen} N.\ G.\ Van Kampen, \textsl{Stochastic Processes in Physics and Chemistry}, North Holland: Amsterdam, 1981.










\end{thebibliography}


\end{document}